\documentclass[11pt, reqno, oneside]{amsart}
\usepackage[margin=1.00in]{geometry}
\usepackage{comment}    
\usepackage[T1]{fontenc}
\usepackage[utf8]{inputenc}
\usepackage{courier}
\usepackage{bm}          
\usepackage{dsfont}      
\usepackage{braket}     
\usepackage{lipsum}      
\usepackage[section]{placeins}
\usepackage{booktabs}
\usepackage{amsmath,amssymb,amsthm}   
\usepackage{mathtools}                
\usepackage{esvect}                   
\usepackage{makecell}
\usepackage{enumitem}

\setlist[enumerate,1]{
    label=(\arabic*),
    leftmargin=*
}
\setlist[enumerate,2]{
    label=(\roman*),
    leftmargin=*
}
\setlist[itemize]{leftmargin=*}
\setlist[enumerate]{leftmargin=*}
\setlist[itemize]{leftmargin=*}    
\numberwithin{equation}{section}
\usepackage[ruled]{algorithm2e}
\DontPrintSemicolon
\usepackage{xcolor, graphicx}
\usepackage[caption=false]{subfig}
\usepackage{tikz}
\usepackage{tikz-cd}

\definecolor{darkgreen}{rgb}{0.0,0.72,0.06}
\definecolor{cerisepink}{rgb}{0.93,0.23,0.51}
\definecolor{red-brown}{rgb}{0.65,0.16,0.16}
\definecolor{refblue}{RGB}{0,90,230}
\definecolor{magenta}{RGB}{200,40,160}

\usepackage{aliascnt}
\usepackage[colorlinks]{hyperref}
\usepackage[capitalize,noabbrev]{cleveref}
\hypersetup{
    colorlinks=true,
    citecolor=darkgreen,
    linkcolor=refblue,
    filecolor=magenta,
    urlcolor=cerisepink
}

\usepackage[backend=biber,style=numeric-comp,sorting=none,url=false,maxbibnames=99,doi=false]{biblatex}
\theoremstyle{plain}
\newtheorem{thm}{Theorem}[section]

\newaliascnt{lem}{thm}
\newtheorem{lem}[lem]{Lemma}
\aliascntresetthe{lem}

\newaliascnt{prop}{thm}
\newtheorem{prop}[prop]{Proposition}
\aliascntresetthe{prop}

\newaliascnt{cor}{thm}
\newtheorem{cor}[cor]{Corollary}
\aliascntresetthe{cor}

\theoremstyle{definition}

\newaliascnt{defn}{thm}
\newtheorem{defn}[defn]{Definition}
\aliascntresetthe{defn}

\newaliascnt{assumption}{thm}

\aliascntresetthe{assumption}

\newaliascnt{notation}{thm}

\aliascntresetthe{notation}

\newaliascnt{example}{thm}
\newtheorem{example}[example]{Example}
\aliascntresetthe{example}

\newaliascnt{remark}{thm}
\newtheorem{remark}[remark]{Remark}
\aliascntresetthe{remark}

\newaliascnt{result}{thm}
\newtheorem{result}[result]{Result}
\aliascntresetthe{result}

\newaliascnt{fact}{thm}
\newtheorem{fact}[fact]{Fact}
\aliascntresetthe{fact}

\newcommand{\C}{\mathbb{C}}

\newcommand{\Z}{\mathbb{Z}}

\let\oldtocsection=\tocsection\let\oldtocsubsection=\tocsubsection\renewcommand{\tocsection}[2]{\hspace{0em}\oldtocsection{#1}{#2}}\renewcommand{\tocsubsection}[2]{\hspace{1em}\oldtocsubsection{#1}{#2}}

\begin{document}

\title{A Cohomological Characterization of the Clifford Hierarchy}

\author{Junaid Aftab}  
\address{Department of Mathematics, University of Maryland, College Park, MD 20742, USA} \email{\href{mailto:junaida@umd.edu}{junaida@umd.edu}}

\address{Department of Mathematics, University of Michigan, Ann Arbor, MI 48109, USA} \email{\href{mailto:junaida@umich.edu}{junaida@umich.edu}} 
\urladdr{\href{https://junaid-aftab.github.io/}                 {https://junaid-aftab.github.io/}}

\date{\today}

\begin{abstract}
The Clifford hierarchy plays a central role in fault-tolerant quantum computation, but its higher levels remain only partially understood. Recent work on quantum higher order Fourier analysis characterizes membership in the Clifford hierarchy using quantum derivatives, but does not provide an explicit description of the individual levels. In this work, we address this problem by showing that the collection of quantum derivatives of a unitary operator forms a non-abelian \(1\)-cocycle, and conversely that every such \(1\)-cocycle arises from a unitary operator. This yields a recursive cohomological characterization of the Clifford hierarchy. We then specialize this framework to the third level and decompose the cocycle data into symplectic, affine, and phase components, each admitting a cohomological interpretation. As applications, we give a new proof that every two-qudit gate in the third level of the Clifford hierarchy is semi-Clifford, and prove for the first time that every three-qudit gate in the third level is semi-Clifford.
\end{abstract}

\maketitle

\setcounter{tocdepth}{2}
\tableofcontents

\section{Introduction}
Gottesman and Chuang~\cite{gottesman1999demonstrating} introduced the Clifford hierarchy in the context of teleportation-based quantum gate implementation. Since then, it has become a fundamental concept in fault-tolerant quantum computation. Mathematically, the Clifford hierarchy is a sequence of subsets of the unitary group on $n$ qudits. The first two levels are groups, whereas the higher levels are not groups. This lack of structure has motivated the development of several algebraic approaches aimed at understanding the higher levels of the hierarchy.~\cite{cui2017diagonal,rengaswamy2019unifying,pllaha2020weyl,hu2021climbingdiagonalcliffordhierarchy,deSilva2021efficient,anderson2024groups,chen2024characterising,feng2024groups,he2025characterizationpermutationgateslevel,anderson2025controlledgatescliffordhierarchy,anderson2025affineequivalencecliffordhierarchy,de2025clifford,bastioni2026climbingcliffordhierarchy,xu2026controlledjumpcliffordhierarchy}.

In contrast to these algebraic approaches, Bu, Gu, and Jaffe~\cite{bu2025quantum} recently developed an analytic approach to the study of the Clifford hierarchy based on quantum higher order Fourier analysis. Their framework extends classical higher order Fourier analysis~\cite{tao2012higher} from the commutative setting of functions on finite abelian groups to the noncommutative setting of unitary operators. A key construction in their approach is the notion of a quantum derivative, which generalizes finite-difference operators to this noncommutative setting. Using quantum derivatives, they define a family of quantum uniformity measures, $\|\cdot\|_{Q^{k+1}}$ for $k \geq 0$, which generalize the classical Gowers norms~\cite{gowers2001new}. Their main result establishes that $\|U\|_{Q^{k+1}} \leq 1$ for every unitary operator $U$ and every $k \geq 0$. Moreover, for $k \geq 1$, they show that
\begin{equation}\label{membership}
U \in \mathcal C^{(k)} \iff \|U\|_{Q^{k+1}} = 1,
\end{equation}
where $C^{(k)}$ is the $k$-th level of the Clifford hierarchy. Therefore, \cref{membership} provides a membership criterion for determining whether a quantum gate belongs to a particular level of the Clifford hierarchy. However, their framework does not use quantum derivatives to investigate the internal structure of individual levels of the hierarchy. This raises the question of whether quantum derivatives can be used to characterize these levels explicitly.

In this work, we begin to address this question by exploiting the observation that the collection of first order quantum derivatives of a unitary operator forms a non-abelian \(1\)-cocycle. Conversely, we prove that every such non-abelian \(1\)-cocycle arises as the collection of first order quantum derivatives of a unitary operator, which is unique up to a global phase. Consequently, we show that \(\mathcal{C}^{(k)}\), modulo multiplication by \(U(1)\) phases, is classified by non-abelian \(1\)-cocycles whose images lie in \(\mathcal{C}^{(k-1)}\). This gives a recursive cohomological characterization of the Clifford hierarchy. We use this perspective to parameterize \(\mathcal{C}^{(3)}\) in terms of symplectic, affine, and phase data that admit a cohomological interpretation. In particular, we show that the affine data are classified by a first cohomology group, while compatible phase data exist precisely when an obstruction in a second cohomology group vanishes. We further show that this obstruction is given by a quadratic expression involving the symplectic form and the cup product. For fixed symplectic data, this yields a cohomological moduli description of \(\mathcal{C}^{(3)}\). We use this framework to provide a new proof that every two-qudit gate in \(\mathcal{C}^{(3)}\) is semi-Clifford for odd prime dimension. We also prove, to the best of our knowledge for the first time, that every three-qudit gate in \(\mathcal{C}^{(3)}\) is semi-Clifford for arbitrary odd prime dimension.

\subsection{Summary and Comparison of Results}
We summarize the main results of this work in \cref{sum-rel}. Since some results recover known results from the literature within our framework, we provide a comparative analysis in \cref{com-rel} to place our contributions in context.

\subsubsection{Summary of Results}\label{sum-rel}
Let $n \geq 1$, let $p$ be a prime, and let $\mathbb{F}_p$ denote the finite field with $p$ elements. Set $V_n=\mathbb{F}_p^{2n}$ and $\mathcal{H}_n=\mathbb{C}^{p^n}$, and let $U\in\mathcal{U}(\mathcal{H}_n)$ be an $n$-qudit unitary operator. We observe that the collection of first order quantum derivatives $\{\partial_a U\}_{a\in V_n}$ defines a non-abelian $1$-cocycle. Conversely, we prove that the non-abelian $1$-cocycle condition is sufficient for a map to arise as the collection of first order quantum derivatives of a unitary operator.

\begin{result}[Informal Version of \cref{thm:integrability}]
\label{informal-integ}
Let $A : V_n \to \mathcal{U}(\mathcal{H}_n)$ be a non-abelian $1$-cocycle. There exists a $U \in \mathcal{U}(\mathcal{H}_n)$, unique up to a $U(1)$ phase, such that $A = A_U$. 
\end{result}

\cref{informal-integ} is an application of the Stone--von Neumann theorem~\cite{rosenberg2004selective,hall2013quantum}. Our observation is that the condition \(A(a)=A_U(a)\) is equivalent to the existence of a unitary \(U\in\mathcal{U}(\mathcal{H}_n)\) such that
\begin{equation}\label{U-furnish}
A(a)^*w(a)=Uw(a)U^*
\end{equation}
for every \(a\in V_n\), where \(w(a)\) denotes the \(n\)-qudit Weyl operator associated with \(a\in V_n\). The maps \(a\mapsto A(a)^*w(a)\) and \(a\mapsto w(a)\) define projective representations of \(V_n\) with the same multiplier and hence lift to irreducible unitary representations of the Heisenberg group \(H(V_n)\). The Stone--von Neumann theorem, which states that irreducible unitary representations of \(H(V_n)\) with the same central character are unique up to unitary equivalence, therefore guarantees the existence of \(U\in\mathcal{U}(\mathcal{H}_n)\) satisfying \cref{U-furnish}. Combining \cref{informal-integ} with the quantum derivative characterization of Bu, Gu, and Jaffe~\cite{bu2025quantum}, namely that \(U\in\mathcal{C}^{(k)}\) if and only if \(\partial_aU\in\mathcal{C}^{(k-1)}\) for every \(a\in V_n\), yields a recursive classification of the Clifford hierarchy.

\begin{result}[Informal Version of \cref{thm:abstract-classification}]
\label{informal-abstract}
Let \(Z_{\tau}^{1}(V_n,\mathcal{U}(\mathcal{H}_n))\) denote the set of non-abelian \(1\)-cocycles for the action of \(V_n\) on \(\mathcal{U}(\mathcal{H}_n)\) by conjugation, \(\tau_a(T)=w(a)Tw(a)^*\), and let \(Z_{\tau}^{1,k}(V_n)\) denote the subset of such $1$-non-abelian cocycles taking values in \(\mathcal{C}^{(k-1)}\). We have
\begin{equation}
\mathcal{C}^{(k)}/U(1)\cong Z_{\tau}^{1,k}(V_n),
\end{equation}
where \(\mathcal{C}^{(k)}/U(1)\) is the \(k\)-th level of the Clifford hierarchy modulo multiplication by $U(1)$ phases. 
\end{result}

\cref{informal-abstract} gives a cohomological description of the Clifford hierarchy by identifying \(\mathcal{C}^{(k)}/U(1)\) with the space of \(\mathcal{C}^{(k-1)}\)-valued non-abelian \(1\)-cocycles. For a gate in \(\mathcal{C}^{(3)}\), each non-abelian \(1\)-cocycle lies in \(\mathcal{C}^{(2)}\). Using the description of \(\mathcal{C}^{(2)}\), we show that the cocycle data separates into three layers: a homomorphism into the symplectic group, an affine \(1\)-cocycle, and a phase obstruction. This yields the following cohomological parametrization of \(\mathcal{C}^{(3)}\).

\begin{result}[Informal Version of \cref{third-level-classification}, \cref{affine-part-H1} and \cref{phase-cohomology}]
\label{informal-third-cohom}
Let $p$ be an odd prime. Gates in \(\mathcal{C}^{(3)}/U(1)\) are in bijection with triples \((S,q,\alpha)\) defining three maps
\begin{equation}
S:V_n\longrightarrow \operatorname{Sp}(V_n), \qquad
q:V_n\longrightarrow V_n,
\qquad
\alpha:V_n\longrightarrow U(1),
\end{equation}
where \(S\) is a group homomorphism, \(q\) is a \(1\)-cocycle, and \(\alpha\) is a phase function.  
\begin{enumerate}
    \item For fixed \(S\), the possible choices of \(q\), modulo translation, are parametrized by \(H^1(V_n,V_{n,S})\), where $V_{n,S}$ is the group $V_n$ with the group action of $V_n$ defined by \(a\cdot v:=S(a)v\). 
    \item For a fixed pair \((q,S)\), a compatible phase function \(\alpha\) exists if and only if the associated \(U(1)\)-valued \(2\)-cocycle \(R_{q,S}\), defined in \cref{rqs}, is cohomologically trivial in \(H^2(V_n,U(1))\).
\end{enumerate}
\end{result}

We further strengthen \cref{informal-third-cohom} by showing that the phase obstruction depends only on the cohomology class \([q]\) and is determined by the symplectic form and the cup product. This allows us to obtain a cohomological moduli description of \(\mathcal{C}^{(3)}\).

\begin{result}[Informal Version of \cref{cor:third-level-cup-square-obstruction} and \cref{prop:third-level-cohomological-moduli}]
\label{informal-third-moduli}
Let \(p\) be an odd prime and fix a homomorphism \(S:V_n\to\operatorname{Sp}(V_n)\). The obstruction class \([R_{q,S}]\in H^2(V_n,U(1))\) depends only on the cohomology class \([q]\in H^1(V_n,V_{n,S})\), giving a well-defined map
\begin{align}
\operatorname{ob}_S:H^1(V_n,V_{n,S})&\longrightarrow H^2(V_n,U(1)),\\
[q]&\longmapsto[R_{q,S}].
\end{align}
defined in terms of the symplectic form and the cup product. Consequently, modulo left and right multiplication by Weyl operators, third-level gates with fixed \(S\) are classified by the cohomology classes \([q]\in H^1(V_n,V_{n,S})\) satisfying \(\operatorname{ob}_S([q])=0\).
\end{result}
 
We next apply \cref{informal-third-cohom} to derive several results concerning semi-Clifford gates in \(\mathcal{C}^{(3)}\).

\begin{result}[Informal Version of \cref{semi-clifford-kernel-criterion}, \cref{witt-index-consequences}, \cref{two-qudit-c3-semi-clifford} and \cref{three-qudit-c3-semi-clifford}]
\label{informal-semi-clifford}
Let $p$ be an odd prime and \(U\in\mathcal{C}^{(3)}\). \(U\) is semi-Clifford if and only if \(K_U:=\ker S_U\) contains a Lagrangian subspace of \(V_n\). Equivalently,
    \begin{equation}\label{witt-informal}
    \dim K_U+\dim (K_U\cap K_U^\perp )\geq 2n.
    \end{equation}
    or $K_U^\perp \subseteq K_U$. Every two and three-qudit gate in \(\mathcal{C}^{(3)}\) is semi-Clifford.
\end{result}

We prove the result for two and three-qudits by analyzing the possible codimensions of \(K_U\). Our techniques involve using the cocycle identities to constrain \(q_U(K_U)\), realizing \(E_U:=\operatorname{im}S_U\) as an abelian subgroup of a symplectic group, and applying Wong's results~\cite{wong1982abelian} on abelian subgroups of Sylow \(p\)-subgroups of symplectic groups to restrict the possible form of \(E_U\). 

\subsubsection{Comparison of Results}\label{com-rel} 
We now compare \cref{informal-integ}--\cref{informal-semi-clifford} with the existing literature. De Silva~\cite{deSilva2021efficient} described unitaries by conjugate tuples and used the Stone--von Neumann theorem, together with constraints such as \(k\)-closedness, to obtain a recursive enumeration of the Clifford hierarchy. Thus, \cref{informal-integ,informal-abstract} are not the first applications of the Stone--von Neumann theorem to the Clifford hierarchy. Their contribution here is instead to organize the complete first order quantum derivative field as a non-abelian \(1\)-cocycle. To the best of our knowledge, this global cocycle formulation, together with its cohomological interpretation, has not previously appeared in the literature. Specializing this formalism to \(\mathcal{C}^{(3)}\), we decompose the cocycle data into symplectic, affine, and phase components. To the best of our knowledge, this yields the first explicit cohomological parametrization of \(\mathcal{C}^{(3)}\), in which the affine data are classified by \(H^1\), while the obstruction to compatible phase data is captured by \(H^2\). We further show that this obstruction depends only on the cohomology class \([q]\) of the affine data and is given by a quadratic expression involving the symplectic form and the cup product. For fixed symplectic data, this yields a moduli-theoretic description of third-level gates, which, to the best of our knowledge, provides a new perspective on \(\mathcal{C}^{(3)}\). We use our formalism to recover several known results and derive new results on semi-Clifford gates in \(\mathcal{C}^{(3)}\). A detailed comparison with the existing literature is provided below:

\begin{enumerate}
\item De Silva~\cite{deSilva2021efficient} characterized semi-Clifford gates via the existence of a Lagrangian semibasis satisfying certain conjugation conditions. In our formalism, this becomes the kernel criterion in \cref{semi-clifford-kernel-criterion}, which, to the best of our knowledge, gives a basis independent characterization of semi-Clifford gates in \(\mathcal{C}^{(3)}\).

\item  \cref{witt-informal} reduces the verification of the semi-Clifford condition to a dimension-theoretic condition on \(K_U := \ker S_U\) and its radical. To the best of our knowledge, this criterion has not previously appeared in the literature.

\item It is known that every two-qudit gate in \(\mathcal{C}^{(3)}\) is semi-Clifford in odd prime dimension. However, our formalism yields a new, short, and structural proof based on our cocycle formalism, \cref{witt-informal} and the classification of elementary abelian \(p\)-subgroups of order \(p^3\) in \(\operatorname{Sp}(\mathbb F_p^4)\). Our approach avoids both the exhaustive computation used for the $p=3$ case~\cite{deSilva2021efficient} and the algebraic-geometric methods used in the $p>3$ case~\cite{chen2024characterising}.  
\item The result that every three-qudit gate in \(\mathcal{C}^{(3)}\) is semi-Clifford for \(p>2\) is new.
\end{enumerate}

\subsection{Related Works}\label{sec:related-work}
Zeng, Chen, and Chuang~\cite{zeng2008semi} introduced
semi-Clifford gates and proved that every one and two qubit gate in \(\mathcal{C}^{(k)}\) is semi-Clifford and that every three-qubit gate in \(\mathcal{C}^{(3)}\) is semi-Clifford. 
De Silva~\cite{deSilva2021efficient} extended the analysis to prime-dimensional qudits and proved that every two qutrit gate in \(\mathcal{C}^{(3)}\) is semi-Clifford. Chen and de Silva~\cite{chen2024characterising} subsequently proved that every two-qudit gate in \(\mathcal{C}^{(3)}\) is semi-Clifford. De Silva and Lautsch~\cite{de2025clifford} later gave a complete classification of the Clifford hierarchy for a single prime-dimensional qudit and showed, in particular, that every one-qudit gate in the Clifford hierarchy is semi-Clifford.  Anderson and Connelly~\cite{anderson2025affineequivalencecliffordhierarchy} proved that every four-qubit gate in \(\mathcal{C}^{(3)}\) is semi-Clifford. 

A different line of work has sought explicit algebraic descriptions of distinguished subclasses of gates in the Clifford hierarchy. Cui, Gottesman, and Krishna~\cite{cui2017diagonal} gave a complete classification of diagonal gates in the Clifford hierarchy for prime-dimensional qudits, describing them in terms of certain polynomial phase functions.  Rengaswamy, Calderbank, and Pfister~\cite{rengaswamy2019unifying} developed a complementary description of diagonal qubit hierarchy gates using symmetric matrices over finite rings. More recent work has considered qubit permutation gates in the third~\cite{he2025characterizationpermutationgateslevel} and fourth~\cite{anderson2025affineequivalencecliffordhierarchy} levels of the Clifford hierarchy, as well as controlled gates~\cite{bastioni2026climbingcliffordhierarchy,xu2026controlledjumpcliffordhierarchy}.

A more recent line of work, which directly motivates the present paper, is quantum higher order Fourier analysis. Bu, Gu, and Jaffe~\cite{bu2025quantum} developed an analytic theory of quantum uniformity measures, establishing positivity and monotonicity properties, norm inequalities, and connections with Weyl--Fourier coefficients and quantum convolution. They also formulated associated inverse-theorem and property-testing problems. Subsequent work has begun to address these questions. Hinsche, Bao, van Dordrecht, Eisert, Bri\"et, and Helsen~\cite{hinsche2026clifford} proved a polynomial inverse theorem for the third noncommutative uniformity norm and obtained an efficient tester for the Clifford group. Bao, Bri\"et, Castro-Silva, van Dordrecht, and Helsen~\cite{zongbo2026cliffordhierarchytestingnearextremizers} established a ``\(99\%\) inverse theorem'' for the fourth noncommutative uniformity norm and used it to obtain an efficient tester for \(\mathcal{C}^{(3)}\). In a complementary direction, Jamneshan~\cite{jamneshan2026quantumuniformitynormspullbacks} showed that the quantum uniformity norms of Bu, Gu, and Jaffe arise as pullbacks of the matrix-valued uniformity norms of Gowers and Hatami~\cite{gowers2016inversestabilitytheoremsapproximate}.

\subsection{Organization}
The remainder of the paper is organized as follows. In \cref{sec:preliminaries}, we introduce preliminary material used throughout the paper. In \cref{formal-framework}, we develop the cocycle classification of the Clifford hierarchy. In \cref{third-clifford}, we specialize this framework to \(\mathcal{C}^{(3)}\) and derive the corresponding cohomological interpretation. In \cref{semi-third}, we apply this description to semi-Clifford gates in \(\mathcal{C}^{(3)}\). 

\subsection{Acknowledgments}
This work was initiated while the author was a graduate student at the University of Maryland, College Park. The author gratefully acknowledges support from the National Science Foundation under Grant DMS-2231533. He also thanks the MathQuantum RTG for organizing the Research Interaction Team on Quantum Information at the University of Maryland, College Park in Spring 2026, which focused on the Clifford hierarchy. The discussions that led to the present work originated through this activity. This work was completed during the author's current affiliation with the University of Michigan, Ann Arbor. 
\section{Preliminaries}\label{sec:preliminaries}
We collect preliminary results used throughout the paper. In \cref{notation}, we discuss the notation used throughout this paper. In \cref{weyl-group}, we discuss the \(n\)-qudit Weyl group. In \cref{clifford-hierarchy}, we discuss aspects of the Clifford hierarchy. Finally, \cref{sec:derivative-characterization}  discusses the quantum derivative. 

\subsection{Notation}\label{notation}
We discuss the main notation used throughout this work.

\subsubsection{Standard Notation}
Throughout, \(n \geq 1\) is an integer and \(p\) is a prime. 
Let \(\mathbb F\) denote an arbitrary field. For a prime \(p\), let \(\omega:=e^{2\pi i/p}\) be a primitive \(p\)-th root of unity, and let \(\mathbb{F}_p\) denote the finite field with \(p\) elements. Let \(U(1)\) denote the circle group, and let \(\mu_m\subseteq U(1)\) denote the subgroup of \(m\)-th roots of unity. Let \(\Gamma_p := \mu_p\) if \(p\) is odd, and \(\Gamma_p := \mu_4\) if \(p=2\). 

\subsubsection{Linear Algebra Notation}\label{la-not}
Let \(V\) be a finite-dimensional vector space over \(\mathbb F\). We denote by \(\operatorname{End}(V)\) the algebra of linear operators on \(V\). If \(V\) is a complex inner product space, then \(\operatorname{End}(V)\) is equipped with the normalized Hilbert--Schmidt inner product
$\langle A,B\rangle_{\mathrm{HS}}:=\frac{\operatorname{Tr}(A^*B)}{\dim V}$
for \(A,B\in\operatorname{End}(V)\). \(I_V\in\operatorname{End}(V)\) denotes the identity operator on \(V\). \(\operatorname{GL}(V)\) denotes the subgroup of \(\operatorname{End}(V)\) consisting of invertible linear operators. If \((V,\langle\, , \,\rangle)\) is an inner product space over \(\mathbb C\), then \(\mathcal U(V)\) denotes the subgroup of \(\operatorname{GL}(V)\) consisting of unitary operators
\begin{equation}
\mathcal U(V):=\{A\in\operatorname{GL}(V):\langle Av,Aw\rangle=\langle v,w\rangle \text{ for all }v,w\in V\}.
\end{equation}

A symplectic vector space is a pair \((V,[\, , \,])\), where \(V\) is a vector space over \(\mathbb F\) and \([\, , \,]:V\times V\to\mathbb F\) is a non-degenerate alternating bilinear form, called a symplectic form. Every finite-dimensional symplectic vector space has even dimension. \(\operatorname{Sp}(V)\) denotes the symplectic group of linear operators on \(V\) that preserve the symplectic form
\begin{equation}
\operatorname{Sp}(V)=\{A\in\operatorname{GL}(V):[Av,Aw]=[v,w]\text{ for all }v,w\in V\}.
\end{equation}
If \(U\subseteq V\) is a subspace of \((V,[\,\cdot,\cdot\,])\), then its symplectic orthogonal complement is defined by
\begin{equation}
L^\perp:=\{v\in V:[v,\ell]=0\text{ for every }\ell\in L\}.
\end{equation}
A subspace \(L\subseteq V\) is called isotropic if \(L\subseteq L^\perp\), coisotropic if \(L^\perp\subseteq L\)  and Lagrangian if \(L=L^\perp\). Equivalently, if \(\dim_{\mathbb F}V=2m\), then \(L\) is Lagrangian if and only if \(L\) is isotropic and \(\dim_{\mathbb F}L=m\). If \(L\subseteq V\) is Lagrangian, there exists a Lagrangian subspace \(L'\subseteq V\) such that \(V=L\oplus L'\) and the symplectic pairing induces an isomorphism \(L' \cong L^*\), where $L^*$ is the dual space of $L$. The reader is referred to \cite{EckhardSGnotes,taylor1992geometry} for further details on symplectic vector spaces.

\subsubsection{Group Theory Notation}\label{gp-not}
For groups \(G\) and \(H\), we write \(\operatorname{Hom}(G,H)\) for the set of group homomorphisms from \(G\) to \(H\). Elements of \(\operatorname{Hom}(G,U(1))\) are called characters. For \(g,h\in G\), the commutator of \(g\) and \(h\) is $ghg^{-1}h^{-1}$. We write \(Z(G)\) for the center of \(G\), consisting of those elements \(g\in G\) such that \([g,h]\) is the identity element for every \(h\in G\). If \(G\) acts on a set \(X\) on the left, we write \(g\cdot x\) for the action of \(g\in G\) on \(x\in X\). If \(G\) acts on a group \(N\) by automorphisms, the corresponding semidirect product is denoted by \(N\rtimes G\), with multiplication
\begin{equation}
(n,g)(n',g')=(n(g\cdot n'),gg').
\end{equation}
If \(E\) is a \(p\)-group of the form $(\mathbb Z/p\Z)^q$ for some $q \geq 0$, then \(E\) has a canonical structure of an \(\mathbb F_p\)-vector space. We have $\dim (\mathbb Z/p\Z)^q = q$ and
 \(|E|=p^{q}\).
 Any group homomorphism between elementary abelian \(p\)-groups is automatically \(\mathbb F_p\)-linear with respect to these vector-space structures.  
A short exact sequence of groups is a sequence of group homomorphisms
\begin{equation}\label{ses}
1\longrightarrow N\xrightarrow{\hspace{0.2cm}i\hspace{0.2cm}}E\xrightarrow{\hspace{0.2cm}\pi\hspace{0.2cm}}Q\longrightarrow1
\end{equation}
such that \(i\) is injective, \(\pi\) is surjective, and \(\operatorname{im}(i)=\ker(\pi)\). The sequence is called central if \(i(N)\) is contained in \(Z(E)\). If \(f:H\to Q\) is a group homomorphism, the pullback of \cref{ses} along \(f\) is the group
\begin{equation}
f^*E:=\{(h,e)\in H\times E:f(h)=\pi(e)\},
\end{equation}
which fits into the short exact sequence
\begin{equation}\label{pullback}
1\longrightarrow N\longrightarrow f^*E\longrightarrow H\longrightarrow1.
\end{equation}
If \cref{ses} is a central extension, then \cref{pullback} is also a central extension.

\subsubsection{Group Cohomology Notation}\label{gp-cohom}
Let \(G\) act on an abelian group \(M\). We write \(C^r(G,M)\), \(Z^r(G,M)\), and \(B^r(G,M)\) for the groups of \(r\)-cochains, \(r\)-cocycles, and \(r\)-coboundaries, respectively, and define
\begin{equation}
H^r(G,M):=Z^r(G,M)/B^r(G,M).
\end{equation}
If \(\eta\in Z^r(G,M)\), we write \([\eta]\in H^r(G,M)\) for its cohomology class. When \(M\) is written multiplicatively, a \(1\)-cochain is a map \(\phi:G\to M\), and it is a \(1\)-cocycle if
\begin{equation}
\phi(gh)=\phi(g)(g\cdot\phi(h))
\end{equation}
for all \(g,h\in G\). A \(2\)-cochain is a map \(\eta:G\times G\to M\), and it is a \(2\)-cocycle if
\begin{equation}
(g\cdot\eta(h,k))\eta(g,hk)=\eta(gh,k)\eta(g,h)
\end{equation}
for all \(g,h,k\in G\).  A \(2\)-cochain is a \(2\)-coboundary if it is of the form \(\delta\phi\), where $\phi$ is a \(1\)-cochain and 
\begin{equation}
(\delta\phi)(g,h):=\frac{\phi(g)(g\cdot\phi(h))}{\phi(gh)}.
\end{equation}
If \(G\) acts on abelian groups \(M\) and \(N\) and \(\psi:M\to N\) is a compatible map, then \(\psi\) induces
\begin{equation}\label{eqn:cohomology-coefficient-map}
\psi_*:H^r(G,M)\longrightarrow H^r(G,N).
\end{equation}
Equivalence classes of central extensions
are classified by \(H^2(G,A)\), where \(G\) acts trivially on \(A\).  If \(f:H\to G\) is a group homomorphism, its defining \(2\)-cocycle is the pullback
\begin{equation}
(f^*\eta)(h_1,h_2):=\eta(f(h_1),f(h_2)),
\end{equation}
and hence its extension class is $f^*[\eta]\in H^2(H,A)$. Suppose that \(G\) acts on abelian groups \(M\) and \(N\).  For \(f\in C^1(G,M)\) and \(g\in C^1(G,N)\), their degree-\((1,1)\) cup product is the map
\begin{align}
\smile : H^1(G,M)\times H^1(G,N) &\longrightarrow H^2(G,M\otimes N),\\
(f\smile g)(a,b) &\longmapsto (f\smile g)(a,b)=f(a)\otimes a\cdot g(b).
\end{align}
If \(G\) also acts on an abelian group \(P\) and \(M\otimes N\to P\) is a linear map compatible with the \(G\)-actions, then composing with the cup product gives a pairing with values in \(H^2(G,P)\).

\subsubsection{Representation Theory Notation}
A representation of a group \(G\) on a complex vector space \(W\) is a homomorphism \(\rho:G\to\operatorname{GL}(W)\). If \(W\) has an inner product and \(\rho(G)\subseteq\mathcal U(W)\), then \(\rho\) is called a unitary representation. A projective representation of \(G\) on \(W\) is a homomorphism \(\bar\rho:G\to\operatorname{GL}(W)/\mathbb C^\times\). Equivalently, it is given by lifts \(\rho(g)\in\operatorname{GL}(W)\) satisfying
\begin{equation}
\rho(g)\rho(h)=\eta(g,h)\rho(gh)
\end{equation}
for some function \(\eta:G\times G\to\mathbb C^\times\), called the multiplier. Associativity implies that
\(\eta \in Z^2(G,\mathbb C^\times)\), where \(G\) acts trivially on
\(\mathbb C^\times\).  If $\rho'$ is another projective representation such that $\rho'(g)=\varphi(g)\rho(g)$ for all $g\in G$, where $\varphi:G\to\mathbb{C}^{\times}$, then $\eta'$ changes by a coboundary
\begin{equation}
\eta'(g,h)
=
(\delta\varphi)(g,h)\eta(g,h)
=
\frac{\varphi(g)\varphi(h)}{\varphi(gh)}\eta(g,h).
\end{equation}
Thus,  \([\eta]\in H^2(G,\mathbb C^\times)\) depends only on the projective representation. In the unitary case, \([\eta]\in H^2(G, U(1))\). 

\subsection{Weyl Group}\label{weyl-group} 
An \(n\)-qudit quantum system with local dimension \(p\) is described by \(\mathcal H_n:=(\mathbb C^p)^{\otimes n}\cong\mathbb C^{p^n}\) with computational basis \(\{\ket{x}:x=(x_1,\dots,x_n)\in\mathbb F_p^n\}\). For \((s,t)\in V_n\), define
\begin{equation}
X^s\ket{x}=\ket{x+s}, \qquad Z^t\ket{x}=\omega^{t\cdot x}\ket{x}.
\end{equation}
Here addition and the dot product are taken modulo \(p\). The operator \(Z^t\) is diagonal in the computational basis and encodes the additive character \(x\mapsto\omega^{t\cdot x}\), while \(X^s\) acts by translation on computational basis states. 
These operators satisfy the relations
\begin{equation}\label{eqn:XZ-relations}
X^sX^{s'}=X^{s+s'}, \qquad Z^tZ^{t'}=Z^{t+t'}, \qquad
X^sZ^t=\omega^{-s\cdot t}Z^tX^s,
\end{equation}

\begin{defn}\label{weyl-def}
The \(n\)-qudit Weyl group is the subgroup of \(\mathcal U(\mathcal H_n)\) generated by the operators \(X^s\), \(Z^t\), and the scalar phases in \(\Gamma_p\). Equivalently,
\begin{equation}
\mathcal W_n=\{\lambda Z^tX^s:\lambda\in\Gamma_p,\ s,t\in\mathbb F_p^n\}\subseteq\mathcal U(\mathcal H_n).
\end{equation}
\end{defn}

The \(n\)-qudit Weyl group is non-abelian. 
If \(V_n:=\mathbb F_p^n\oplus\mathbb F_p^n\cong\mathbb F_p^{2n}\) is equipped with the symplectic form
\begin{equation}
[a,b]=[(s,t),(s',t')]
:=t\cdot s'-s\cdot t',
\end{equation}
for \(s,t,s',t'\in\mathbb F_p^n\), then
\begin{equation}
Z^tX^sZ^{t'}X^{s'}=\omega^{[a,b]}Z^{t'}X^{s'}Z^tX^s.
\end{equation}
Hence, \(Z^tX^s\) and \(Z^{t'}X^{s'}\) commute if and only if \([a,b]=0\). It is convenient to modify the operators \(Z^tX^s\) by scalar phases. This is done in \cref{modify-ops}.

\begin{defn}\label{modify-ops}
For \(a=(s,t)\in V_n\), define the normalized Weyl operator
\begin{equation}
w(a)=w(s,t):=
\begin{cases}
\omega^{-2^{-1}s\cdot t}Z^tX^s, & p \text{ odd},\\
i^{-s\cdot t}Z^tX^s, & p=2.
\end{cases}
\end{equation}
\end{defn}

The normalized Weyl operators form an orthonormal basis of \(\operatorname{End}(\mathcal H_n)\) with respect to the normalized Hilbert--Schmidt inner product. Thus, every \(B\in\operatorname{End}(\mathcal H_n)\) has a unique expansion \(B=\sum_{a\in V_n}\Xi_B(a)w(a)\), where \(\Xi_B(a)=p^{-n}\operatorname{Tr}(w(a)^*B)\). Moreover,
\begin{equation}\label{weyl-inner}
\frac{1}{p^{n}}\operatorname{Tr}(w(a)^*w(b))=\begin{cases}1, & a=b,\\0, & a\neq b.\end{cases}
\end{equation}
for all \(a,b\in V_n\). The normalized Weyl operators define a projective representation of \(V_n\). Since we will use this fact repeatedly, we prove it explicitly in \cref{weyl-mul-law}.

\begin{lem}\label{weyl-mul-law}
Let \(p\) be an odd prime. The map
\begin{align}
w : V_n & \longrightarrow \mathcal{U}(\mathcal{H}_n) \\
a & \longmapsto w(a)
\end{align}
is a projective representation of the additive group \(V_n\). More explicitly, \(w(a)w(b)=\omega^{\frac12[a,b]}w(a+b)\) for all \(a,b\in V_n\). Consequently,
\begin{equation}\label{comm-relation-odd prime}
w(a)w(b)=\omega^{[a,b]}w(b)w(a)
\end{equation}
for all \(a,b\in V_n\).
\end{lem}

\begin{proof}
Let \(a=(s,t)\) and \(b=(s',t')\) be elements of \(V_n\). By definition, \(w(a)=\omega^{-2^{-1}s\cdot t}Z^tX^s\) and \(w(b)=\omega^{-2^{-1}s'\cdot t'}Z^{t'}X^{s'}\). Using \(X^sZ^{t'}=\omega^{-s\cdot t'}Z^{t'}X^s\), we obtain
\begin{align}
w(a)w(b)=\omega^{-2^{-1}s\cdot t}\omega^{-2^{-1}s'\cdot t'}Z^tX^sZ^{t'}X^{s'} = \omega^{-2^{-1}s\cdot t-2^{-1}s'\cdot t'-s\cdot t'} Z^{t+t'}X^{s+s'}.
\end{align}
On the other hand, \(w(a+b)=\omega^{-2^{-1}(s+s')\cdot(t+t')}Z^{t+t'}X^{s+s'}\). Therefore,
\begin{align}
w(a)w(b) &= \omega^{-2^{-1}s\cdot t-2^{-1}s'\cdot t'-s\cdot t' +2^{-1}(s+s')\cdot(t+t')} w(a+b) \\ &= \omega^{2^{-1}(t\cdot s'-s\cdot t')} w(a+b) = \omega^{\frac12[a,b]}w(a+b).
\end{align}
\cref{comm-relation-odd prime} follows by comparing \(w(a)w(b)\) and \(w(b)w(a)\) and using \([b,a]=-[a,b]\). This completes the proof.
\end{proof}

\begin{remark}\label{qubit-mult-comm}
Let \(p=2\). For \(a=(s,t)\in V_n\), define $\vartheta(a):=s\cdot t,$ where the dot product is evaluated in \(\mathbb Z\) using the representatives \(0,1\). If \(a=(s,t)\) and \(b=(s',t')\), then
\begin{equation}\label{p=2}
w(a)w(b)=i^{\mu(a,b)}w(a+b),
\end{equation}
where
$\mu(a,b):=\vartheta(a+b)-\vartheta(a)-\vartheta(b)+2s\cdot t'
\pmod 4.$
It can be checked that the analog of \cref{comm-relation-odd prime} is
\begin{equation}\label{comm-weyl-qubit}
w(a)w(b)=(-1)^{[a,b]}w(b)w(a),
\end{equation}
where $[a,b]=t\cdot s'+s\cdot t'\in\mathbb F_2.$
Thus, although the commutation relation has the same symplectic form as in the odd prime case, the multiplication multiplier is different.  In what follows, we prove statements for odd \(p\) and indicate the modifications required when \(p=2\). 
\end{remark}

\begin{remark}
\label{symplectic-invariance-weyl-multiplier}
For \(p=2\), the commutator determined by the symplectic form remains
invariant under \(\operatorname{Sp}(V_n)\), but the fourth-root multiplier
\(i^{\mu(a,b)}\) from \cref{p=2} need not be symplectically invariant.
This failure occurs already for one qubit. 
Consider the linear map
\begin{align}
S\colon \mathbb F_2^2 &\longrightarrow \mathbb F_2^2,\\
(s,t) &\longmapsto (s,t+s).
\end{align}
It can be checked that $S \in \operatorname{Sp}(\mathbb F^2_2)$. Let $a=(1,0)$ and $b=(0,1)$.  Then
$\mu(a,b)=3\pmod 4$, $Sa=(1,1)$, $Sb=(0,1)$ and $\mu(Sa,Sb)=1\pmod 4$. Therefore,
$i^{\mu(Sa,Sb)} = i \neq -i = i^{\mu(a,b)}.$
Thus, the multiplier in \cref{p=2} is not invariant under the symplectic group. This distinction will be relevant in \cref{formal-framework}.
\end{remark}

\subsection{Clifford Hierarchy}\label{clifford-hierarchy}
The Clifford hierarchy plays a central role in fault-tolerant quantum computation. By the Gottesman--Knill theorem~\cite{gottesman1998heisenberg}, quantum circuits constructed from stabilizer-state preparation, gates in the first two levels of the Clifford hierarchy and Pauli measurements are classically simulable. In contrast, the higher levels contain the non-Clifford gates required for universal quantum computation.

\begin{defn}\label{clifford-def}
Let \(\mathcal C^{(0)}:=U(1)I\). The Clifford hierarchy is the family \(\{\mathcal C^{(k)}\}_{k\geq 0}\) of subsets of \(\mathcal U(\mathcal H_n)\) defined recursively as follows. The first level is the Weyl group (up to phase):
\begin{equation}
\mathcal C^{(1)}:=\{\lambda w(a):\lambda\in U(1),\ a\in V_n\}.
\end{equation}
For \(k\geq2\),
\begin{equation}
\mathcal C^{(k)}
:=
\{U\in\mathcal U(\mathcal H_n):U\mathcal C^{(1)}U^*\subseteq\mathcal C^{(k-1)}\}.
\end{equation}
\end{defn}

The second level \(\mathcal C^{(2)}\) is precisely the Clifford group, namely the normalizer of the Weyl group up to phase. For \(k\geq 1\), we write \(\overline{\mathcal C}^{(k)}:=\overline{\mathcal C}^{(k)}\) for the \(k\)-th level of the Clifford hierarchy modulo multiplication by \(U(1)\) phases, and denote the image of \(U\in\mathcal C^{(k)}\) in \(\overline{\mathcal C}^{(k)}\) by \([U]\). We record, without proof, some basic properties of the Clifford hierarchy.

\begin{fact}\label{clifford-props-no-proof}
The Clifford hierarchy satisfies the following properties:
\begin{enumerate}
    \item For each \(k\geq0\), \(\mathcal C^{(k)}\subseteq\mathcal C^{(k+1)}\).
    \item For \(k\geq3\), \(\mathcal C^{(k)}\) is generally not a group.
    \item For each \(k\geq1\), if \(U\in\mathcal C^{(k)}\) then \(w(a)Uw(b)\in\mathcal C^{(k)}\) for all \(a,b\in V_n\).
    \item For each \(k\geq1\), if \(U\in\mathcal C^{(k)}\) then \(CUC^*\in\mathcal C^{(k)}\) for all \(C\in\mathcal C^{(2)}\).
    \item For each \(k\geq2\), if \(U\in\mathcal C^{(k)}\) then $C_1UC_2\in\mathcal C^{(k)}$ for all \(C_1,C_2\in\mathcal C^{(2)}\). 
\end{enumerate}
\end{fact}

\begin{remark}
The second statement of \cref{clifford-props-no-proof} follows from Anderson's example~\cite{anderson2024groups}: the single-qubit gate \(HT\) lies in \(\mathcal C^{(3)}\), whereas \((HT)^2\) does not belong to \(\mathcal C^{(k)}\) for any \(k\geq1\).
\end{remark}

Diagonal gates in the Clifford hierarchy typically admit simpler circuit decompositions than arbitrary gates at the same level. This motivates the study of semi-Clifford gates, which are defined in \cref{def:semi-clifford}. 

\begin{defn}\label{def:semi-clifford}
An \(n\)-qudit unitary operator \(U\) is called semi-Clifford if there exist \(C_1,C_2\in\mathcal C^{(2)}\) and an \(n\)-qudit diagonal unitary \(D\) such that \(U=C_1DC_2\).
\end{defn}

\cref{def:semi-clifford} gives a convenient normal-form definition of semi-Clifford gates as those that are Clifford-equivalent to diagonal unitaries. However, Zeng et al.~\cite{zeng2008semi} originally defined qubit semi-Clifford gates as operators that map, by conjugation, at least one maximal abelian subgroup of \(\mathcal C^{(1)}\) to another. For completeness, we prove the equivalence of these definitions in \cref{lem:semi-clifford-diagonal-equivalence}, since the more difficult implication is given only in compressed form in~\cite{zeng2008semi}. Moreover, our proof extends directly to the qudit setting.  The proof is relegated to \cref{appendix:semi-clifford-proof}.

\begin{lem}[Proposition 1 in \cite{zeng2008semi}]
\label{lem:semi-clifford-diagonal-equivalence}
An \(n\)-qudit unitary operator \(U\) is semi-Clifford if and only if there exist maximal abelian subgroups \(A,A'\subseteq \mathcal C^{(1)}\) such that \(UAU^*=A'\).
\end{lem}




\subsection{Quantum Derivatives}\label{sec:derivative-characterization}

Bu, Gu, and Jaffe~\cite{bu2025quantum} introduced quantum derivatives of unitary operators and showed that the Clifford hierarchy can be characterized in terms of them. We briefly review the relevant details below. The arguments here do not require the full machinery of quantum higher order Fourier analysis. We begin by defining the action of \(V_n\) on \(\operatorname{End}(\mathcal H_n)\), which plays the role of a noncommutative translation action.

\begin{defn}\label{left-action}
Define the action
\begin{align}
\tau:V_n\times \operatorname{End}(\mathcal H_n)
&\longrightarrow \operatorname{End}(\mathcal H_n),
\label{eqn:tau-def}\\
(a,T)&\longmapsto w(a)T w(a)^* .
\end{align}
We refer to \(\tau\) as the conjugation action of \(V_n\) on \(\operatorname{End}(\mathcal H_n)\) and write \(\tau_a(T):=w(a)T w(a)^*\).
\end{defn}
 
Since \(a\mapsto w(a)\) is a projective representation of \(V_n\) by \cref{weyl-mul-law} and \cref{qubit-mult-comm}, the scalar phase cancels under conjugation. Hence, \(\tau_{a+b}=\tau_a\circ\tau_b\) for all \(a,b\in V_n\), and \(\tau_0=\operatorname{Id}_{\operatorname{End}(\mathcal H_n)}\). Thus, \(\tau\) defines a left action of \(V_n\) on \(\operatorname{End}(\mathcal H_n)\). Since \(V_n\) is abelian, the same family defines a right action via \(T\cdot a=\tau_a(T)\). Indeed,
\begin{equation}
(T\cdot a)\cdot b=\tau_b(\tau_a(T))=\tau_{b+a}(T)=\tau_{a+b}(T)=T\cdot(a+b).
\end{equation}
Using the commutation relations in \cref{comm-relation-odd prime} and \cref{comm-weyl-qubit}, we obtain
\begin{align}\label{tau-on-Weyl}
\tau_a(w(b))=w(a)w(b)w(a)^*=\omega^{[a,b]}w(b)
\end{align}
for all \(a,b\in V_n\). The quantum derivative is defined in \cref{def:derivative}.

\begin{defn}\label{def:derivative}
Let \(U\in\mathcal U(\mathcal H_n)\) and \(a\in V_n\). The quantum derivative of \(U\) in the direction \(a\) is
$\partial_a U:=\tau_a(U)U^*=w(a)Uw(a)^*U^*.$
\end{defn}

 \cref{derivative-characterization} states the characterization of the Clifford hierarchy due to Bu, Gu, and Jaffe.

\begin{fact}[Proposition 22 in \cite{bu2025quantum}]
\label{derivative-characterization}
For every \(k\geq 1\) and every unitary \(U\in\mathcal U(\mathcal H_n)\), we have \(U\in\mathcal C^{(k)}\) if and only if \(\partial_aU\in\mathcal C^{(k-1)}\) for every \(a\in V_n\).
\end{fact}




\section{Formal Framework}\label{formal-framework}
We begin by reinterpreting first order quantum derivatives as non-abelian \(1\)-cocycles, which yields a recursive cohomological description of the Clifford hierarchy. In \cref{weyl-der}, we show that quantum derivatives satisfy a non-abelian \(1\)-cocycle condition and use this observation to define \emph{flat Weyl derivative fields} abstractly. In \cref{sec:integrability}, we prove the \emph{integrability theorem}, which states that every flat Weyl derivative field arises as the quantum derivative of a unitary operator. In \cref{abstract-class}, we impose the Clifford-level constraint on flat Weyl derivative fields, thereby obtaining a recursive classification of the Clifford hierarchy. Finally, in \cref{sec:low-order}, we validate the framework by recovering the standard descriptions of the $\mathcal{C}^{(1)}$ and $\mathcal{C}^{(2)}$. 

\subsection{Weyl Derivative Fields}\label{weyl-der}
We begin by reinterpreting quantum derivatives as a $1$-cocycle. 

\begin{defn}\label{def:derivative-field}
For \(U\in \mathcal U(\mathcal H_n)\), define its \emph{Weyl derivative field} to be the map
\begin{align}
A_U : V_n &\longrightarrow \mathcal U(\mathcal H_n), \label{eqn:derivative-field-def}\\
a &\longmapsto \partial_a U.
\end{align}
Equivalently, \(A_U\) records the first order quantum derivative of \(U\) in every direction \(a\in V_n\).
\end{defn}

The key observation is that \(A_U\) satisfies a \(1\)-cocycle condition. This is proven in \cref{lem:derivative-field-cocycle-identity}. 

\begin{lem}\label{lem:derivative-field-cocycle-identity}
For every \(U\in \mathcal U(\mathcal H_n)\), the derivative field \(A_U\)
satisfies
\begin{equation}\label{eqn:derivative-field-cocycle-identity}
\tau_b(A_U(a))A_U(b)=A_U(a+b)=\tau_a(A_U(b))A_U(a)
\end{equation}
for all \(a,b\in V_n\). That is,  $A_U$ is a $1$-cocycle.
\end{lem}

\begin{proof}
We have
$A_U(a+b)=\tau_{a+b}(U)U^*=\tau_a(\tau_b(U))U^*.$
Inserting the identity operator in the form $\tau_a(U^*)\tau_a(U) = I$ between the two factors, we get
\begin{align}
A_U(a+b)
&=\tau_a(\tau_b(U))\tau_a(U^*)\tau_a(U)U^*  \\
&=\tau_a(\tau_b(U)U^*)\tau_a(U)U^*=\tau_a(A_U(b))A_U(a).
\end{align}
Since \(V_n\) is abelian, we may interchange \(a\) and \(b\) in the preceding
identity. Hence, $A_U(a+b)=A_U(b+a)=\tau_b(A_U(a))A_U(b)$. This completes the proof. 
\end{proof}

With respect to the right action of \(V_n\) on \(\mathcal U(\mathcal H_n)\) given by \(T\cdot a:=\tau_a(T)\), \cref{lem:derivative-field-cocycle-identity} gives
\begin{equation}
A_U(a+b)=(A_U(a)\cdot b)A_U(b)
\end{equation}
for all \(a,b\in V_n\). Thus, \(A_U\) is a \(1\)-cocycle for this right action. This is the right-action analog of the usual \(1\)-cocycle condition in group cohomology~\cite{brown2012cohomology}. Motivated by this observation, we consider arbitrary maps \(A:V_n\to\mathcal U(\mathcal H_n)\), which need not satisfy the \(1\)-cocycle condition.

\begin{defn}\label{def:flat-field}
Let \(A:V_n\to\mathcal U(\mathcal H_n)\) be a map, called a \emph{derivative field}. \(A\) is \emph{flat} if
\begin{equation}
A(a+b)=(A(a)\cdot b)A(b)
\end{equation}
for all \(a,b\in V_n\), where \(T\cdot b:=\tau_b(T)\). We denote the set of \emph{flat derivative fields} by $Z^1_\tau(V_n,\mathcal U(\mathcal H_n))$. 
\end{defn}

Note that every \(A\in Z^1_\tau(V_n,\mathcal U(\mathcal H_n))\) satisfies \(A(0)=I\). Indeed, setting \(a=b=0\) gives
$A(0)=\tau_0(A(0))A(0)=A(0)^2,$ and multiplying by \(A(0)^*\) yields \(A(0)=I\).

\begin{example}
For every \(U\in\mathcal U(\mathcal H_n)\), \cref{lem:derivative-field-cocycle-identity} implies that \(A_U\) is a flat derivative field.
\end{example}

Since the elements of \(Z^1_\tau(V_n,\mathcal U(\mathcal H_n))\) take values in the non-abelian group \(\mathcal U(\mathcal H_n)\), it is natural to consider the corresponding non-abelian cohomology. Following the standard framework of non-abelian cohomology~\cite{serre1997galois}, adapted to our right-action convention, we define the associated equivalence relation on \(Z^1_\tau(V_n,\mathcal U(\mathcal H_n))\). We refer to this equivalence relation as \emph{gauge equivalence}, and to equivalent cocycles as being related by a \emph{gauge transformation}.

\begin{defn}\label{def:gauge-equivalence}
Let \(A,B\in Z^1_\tau(V_n,\mathcal U(\mathcal H_n))\). We say that \(A\) and \(B\) are gauge equivalent, and write \(A\sim B\), if there exists \(G\in \mathcal U(\mathcal H_n)\) such that $B(a)=\tau_a(G)A(a)G^*$ for all \(a\in V_n\).
\end{defn}

\begin{example}\label{ex:weyl-field-gauge}
Let \(U,G\in \mathcal U(\mathcal H_n)\). Then the Weyl derivative fields
\(A_U\) and \(A_{GU}\) are gauge equivalent. Indeed, for every \(a\in V_n\) we have
\begin{align}
A_{GU}(a)
&=\tau_a(GU)(GU)^*\\
&=\tau_a(G)\tau_a(U)U^*G^*\\
&=\tau_a(G)A_U(a)G^* .
\end{align}
Thus \(A_{GU}\sim A_U\). In particular, gauge equivalence of Weyl derivative
fields records the effect of left multiplication of the underlying unitary.
\end{example}

\begin{lem}\label{lem:gauge-equivalence-preserves-flatness}
Let \(A:V_n\to\mathcal U(\mathcal H_n)\) be a flat derivative field, and let \(G\in\mathcal U(\mathcal H_n)\). Define \(B:V_n\to\mathcal U(\mathcal H_n)\) by $B(a):=\tau_a(G)A(a)G^*$.  Then \(B\) is also a flat derivative field.
\end{lem}

\begin{proof}
Note that $B(0)=\tau_0(G)A(0)G^*=GG^*=I$. Using the flatness of \(A\), we have
\begin{align}
B(a+b)
&=\tau_{a+b}(G)A(a+b)G^* \\
&=\tau_b(\tau_a(G))\tau_b(A(a))A(b)G^* \\
&=\tau_b(\tau_a(G)A(a)G^*)\tau_b(G)A(b)G^*=\tau_b(B(a))B(b)=(B(a)\cdot b)B(b).
\end{align}
for \(a,b\in V_n\). Hence, \(B\in Z^1_\tau(V_n,\mathcal U(\mathcal H_n))\).
\end{proof}

By \cref{lem:gauge-equivalence-preserves-flatness}, gauge transformations preserve flatness. Therefore, gauge equivalence defines an equivalence relation on \(Z^1_\tau(V_n,\mathcal{U}(\mathcal{H}_n))\). This leads to \cref{def:non-abelian-H1}.

\begin{defn}\label{def:non-abelian-H1}
The first non-abelian cohomology set of \(V_n\) with coefficients in \(\mathcal U(\mathcal H_n)\) for the right conjugation action \(T\cdot a:=\tau_a(T)\) is $H^1_\tau(V_n,\mathcal U(\mathcal H_n)):= Z^1_\tau(V_n,\mathcal U(\mathcal H_n))/\sim$, where \(\sim\) is the relation in \cref{def:gauge-equivalence}.
\end{defn}

We show in \cref{thm:integrability} that every \(A\in Z^1_{\tau}(V_n,\mathcal U(\mathcal H_n))\) is of the form \(A=A_U\) for some \(U\in\mathcal U(\mathcal H_n)\). Consequently, by \cref{ex:weyl-field-gauge}, every flat derivative field is gauge equivalent to the trivial derivative field \(A_I\), and hence \(H^1_\tau(V_n,\mathcal U(\mathcal H_n))\) is trivial.
 
\subsection{Integrability Theorem}\label{sec:integrability}
We now prove \cref{thm:integrability}, which states that every flat derivative field arises as the Weyl derivative field of a unitary operator. Thus, no information is lost in passing from unitary operators to flat derivative fields. The proof relies on the Stone--von Neumann theorem for the Heisenberg group associated with a finite symplectic vector space~\cite{gurevich2010notes, lysenko2022}. 

\begin{fact}[Finite Stone--von Neumann theorem]
\label{fact:finite-stone-von-neumann}
Let $p$ be an odd prime and \((V,\beta)\) be a finite-dimensional symplectic vector space over \(\mathbb F_p\). Let \(H(V)\) be a finite group fitting into an extension
\begin{equation}
\label{eq:finite-heisenberg-extension-full}
1 \longrightarrow \mathbb F_p \overset{\iota}{\longrightarrow}
H(V) \overset{\pi}{\longrightarrow} V \longrightarrow 1
\end{equation}
such that \(\iota(\mathbb F_p)=Z(H(V))\).  Assume that the commutator pairing of \(H(V)\) recovers the symplectic form \(\beta\). Namely, for every \(v,w\in V\), and for any lifts \(\widetilde v,\widetilde w\in H(V)\) satisfying $\pi(\widetilde v)=v$ and  $\pi(\widetilde w)=w$ one has $\widetilde v\widetilde w\widetilde v^{-1}\widetilde w^{-1} = \iota(\beta(v,w))$. Fix a non-trivial character \(\chi:\mathbb F_p\to \mathbb C^\times\) given by \(\chi(z)=\omega^z\). There exists, up to unitary equivalence, a unique irreducible unitary representation \(\rho_\chi:H(V)\to \mathcal U(\mathcal H_\chi)\) such that $\rho_\chi(\iota(z))=\chi(z)I_{\mathcal H_\chi}$ for every $z\in\mathbb F_p$. 
\end{fact}


\begin{thm}\label{thm:integrability}
Let \(p\) be an odd prime. If \(A \in Z_\tau^1(V_n,\mathcal U(\mathcal H_n))\), there exists a unitary \(U \in \mathcal U(\mathcal H_n)\), unique up to global phase, such that \(A = A_U\). Equivalently, every flat derivative field is the Weyl derivative field of a unitary, and the unitary is unique modulo a \(U(1)\)-phase.
\end{thm}

\begin{proof} 
Let $H(V_n):=\mathbb F_p \times V_n \cong \mathbb F_p \times \mathbb F^{2n}_p \cong \mathbb F^{2n+1}_p$ with multiplication law
\begin{equation}\label{eqn:explicit-heisenberg-product}
(z,a)(z',b) := \left(z+z'+\frac12[a,b], a+b\right).
\end{equation}
We first verify that \(H(V_n)\) is a central extension of \(V_n\) by \(\mathbb{F}_p\), with central subgroup \(\mathbb{F}_p\times\{0\}\), and that its commutator pairing recovers the symplectic form \([\,\cdot,\cdot\,]\) on \(V_n\). Let \(\pi:H(V_n)\to V_n\) denote the canonical projection. Then \(\pi\) is a group homomorphism with \(\ker\pi=\mathbb{F}_p\times\{0\}\). The inclusion \(\iota:\mathbb{F}_p\to H(V_n)\), defined by \(\iota(z)=(z,0)\), identifies \(\mathbb{F}_p\) with \(\ker\pi\). Moreover, this subgroup is central since
\begin{equation}
(z,0)(z',a)=(z+z',a)=(z',a)(z,0)
\end{equation}
for all \(z,z'\in\mathbb{F}_p\) and \(a\in V_n\).  In fact, this central subgroup is the full center of \(H(V_n)\). Indeed, suppose \((z,a)\in Z(H(V_n))\). Then \((z,a)\) commutes with every \((0,b)\in H(V_n)\). We have
\begin{align}
(z,a)(0,b)(z,a)^{-1}(0,b)^{-1}
&=(z,a)(0,b)(-z,-a)(0,-b)\\
&=\left(z+\frac{1}{2}[a,b],a+b\right)(-z,-a)(0,-b)\\&=\left([a,b],b\right)(0,-b)=\left([a,b],0\right).
\end{align}
Therefore, \([a,b]=0\) for every \(b\in V_n\). Since the symplectic form  on \(V_n\) is non-degenerate, this implies \(a=0\). Hence, every central  element of \(H(V_n)\) lies in \(\mathbb F_p\times\{0\}\). Therefore, $Z(H(V_n))=\mathbb F_p\times\{0\}$. We have that \(H(V_n)\) fits into the central extension of groups
\begin{equation}
1\longrightarrow \mathbb{F}_p
\overset{\iota}{\longrightarrow}
H(V_n)
\overset{\pi}{\longrightarrow}
V_n
\longrightarrow 1.
\end{equation}
It remains to compute the commutator pairing. Let $a,b \in V_n$ and let \((z,a),(z',b)\in H(V_n)\) be arbitrary lifts of \(a\) and \(b\), respectively. Since $(z,a)^{-1}=(-z,-a)$, we compute
\begin{align}
(z,a)(z',b)(z,a)^{-1}(z',b)^{-1}
&=(z,a)(z',b)(-z,-a)(-z',-b) \\
&=\left(z+z'+\frac12[a,b],a+b\right)(-z,-a)(-z',-b) \\
&=\left(z'+\frac12[a,b]+\frac12[a+b,-a],b\right)(-z',-b) \\
&=\left(z'+[a,b],b\right)(-z',-b)=\left([a,b]+\frac12[b,-b],0\right)=([a,b],0).
\end{align}
Hence, the commutator is independent of the lifts and recovers the symplectic form. That is,
\begin{equation}\label{eqn:heisenberg-commutator-recovers-symplectic-form}
[(z,a),(z',b)]_{H(V_n)}=([a,b],0)=\iota([a,b]).
\end{equation}
By \cref{weyl-mul-law}, the map \(a \mapsto w(a)\) defines a projective representation of \(V_n\) with multiplier \(\alpha(a,b)=\omega^{\frac{1}{2}[a,b]}\). Define the map
\begin{align}
B_A : V_n & \longrightarrow \mathcal{U}(\mathcal{H}_n) \\
a & \longmapsto A(a)^*w(a).
\end{align}
Since \(A\) is a flat derivative field, it satisfies \(A(a+b)=\tau_a(A(b))A(a)\). Taking adjoints gives \(A(a+b)^*=A(a)^*\tau_a(A(b)^*)\). Thus, for \(a,b\in V_n\), we have
\begin{align}
B_A(a)B_A(b) &=A(a)^*w(a)A(b)^*w(b) \\ &=A(a)^*\tau_a(A(b)^*)w(a)w(b)\\&=\omega^{\frac12[a,b]}A(a)^*\tau_a(A(b)^*)w(a+b)\\&=\omega^{\frac12[a,b]}A(a+b)^*w(a+b)=\omega^{\frac12[a,b]}B_A(a+b).
\end{align}
Hence, \(B_A\) is also a projective representation of \(V_n\) with the same multiplier \(\alpha\).  We now lift these projective representations to irreducible unitary representations of \(H(V_n)\). Define
\begin{align}
\rho_w:H(V_n)&\longrightarrow \mathcal{U}(\mathcal{H}_n), \label{eqn:rho-w-def}\\
(z,a)&\longmapsto \omega^z w(a).
\end{align}
For every \((z,a)\in H(V_n)\) the operator \(w(a)\) is unitary and \(\omega^z\in U(1)\), so \(\rho_w(z,a)\in \mathcal U(\mathcal H_n)\). Moreover, we have
\begin{align}
\rho_w(z,a)\rho_w(z',b)&=\omega^{z+z'}w(a)w(b) \\
&=\omega^{z+z'+\frac12[a,b]}w(a+b) \\
&=\rho_w\left(z+z'+\frac12[a,b],a+b\right)=\rho_w((z,a)(z',b)).
\end{align}
Hence, $\rho_w$ is a unitary representation of $H(V_n)$. 
\(\rho_w\) is irreducible. Indeed, if \(T\in\operatorname{End}(\mathcal H_n)\) commutes with \(\rho_w(z,a)\) for all \((z,a)\in H(V_n)\), then \(T\) commutes with \(w(a)\) for every \(a\in V_n\). Equivalently, \(\tau_a(T)=T\) for every \(a\in V_n\). 
By \cref{cor:tau-invariants-are-scalars}, we have \(T=\lambda I\) for $\lambda \in U(1)$. Hence, \(\rho_w\) is irreducible. Similarly, define
\begin{align}
\rho_A:H(V_n)&\longrightarrow \mathcal{U}(\mathcal{H}_n), \label{eqn:rho-A-def}\\
(z,a)&\longmapsto \omega^z B_A(a).
\end{align}
As above, \(\rho_A\) is a unitary representation of \(H(V_n)\). We now show that \(\rho_A\) is irreducible. Since \(H(V_n)\) is finite, Maschke's theorem~\cite{serre1977linear,etingof2011introduction} implies that \(\rho_A\) decomposes as a direct sum of irreducible unitary representations. Let \(\mathcal K\subseteq \mathcal H_n\) be any irreducible summand. We have 
$\rho_A(z,0)=\omega^z I$ for \(z\in\mathbb F_p\). Hence, it acts on \(\mathcal K\) by the same scalar character \(z\mapsto \omega^z\). Therefore \(\mathcal K\) is an irreducible representation of \(H(V_n)\) with the central character \(\chi(z)=\omega^z\). By \cref{fact:finite-stone-von-neumann}, we must have \(\dim \mathcal K=p^n\). But \(\dim \mathcal H_n=p^n\). Thus \(\mathcal K=\mathcal H_n\), and so \(\rho_A\) is irreducible. Putting everything together, \(\rho_w\) and \(\rho_A\) are irreducible unitary representations of \(H(V_n)\) with the same central character \(\chi(z)=\omega^z\). By \cref{fact:finite-stone-von-neumann}, \(\rho_w\) and \(\rho_A\) are unitarily equivalent. Hence, there exists \(U\in \mathcal U(\mathcal H_n)\) such that $\rho_A(z,a)=U\rho_w(z,a)U^*$ for all $(z,a)\in H(V_n)$. Setting \(z=0\), we obtain
\begin{equation}
A(a)^*w(a)=B_A(a)=Uw(a)U^* \quad \Longleftrightarrow \quad w(a)Uw(a)^*=A(a)U
\end{equation}
for all $a\in V_n$. Therefore, $\tau_a(U)=A(a)U$. Multiplying on the right by \(U^*\), we obtain
\begin{equation}\label{eqn:integrability-proof-11}
A(a)=\tau_a(U)U^*=\partial_a(U)=A_U(a).
\end{equation}
This proves existence. We now prove uniqueness. Suppose \(U_1,U_2\in \mathcal U(\mathcal H_n)\) both satisfy  $A(a)=\tau_a(U_j)U_j^*$ for all $a\in V_n$ and $j=1,2$. Then $\tau_a(U_2^*U_1)=U_2^*U_1$ for all $a\in V_n$. By \cref{cor:tau-invariants-are-scalars}, we have \(U_2^*U_1=\lambda I\) for $\lambda \in U(1)$. Hence, \(U_1\) and \(U_2\) differ by a $U(1)$ phase. This completes the proof.
\end{proof}

\begin{remark}\label{rem:integrability-qubit-case}
When \(p=2\), the normalized Weyl operators satisfy \cref{p=2}. The corresponding Heisenberg group is \(H_2(V_n):=\mathbb Z_4\times V_n\) with multiplication
\begin{equation}
(r,a)(s,b):=(r+s+\mu(a,b),a+b),
\end{equation}
where \(\mu(a,b)\) is the multiplier from \cref{p=2}. Its center is \(\mathbb Z_4\times\{0\}\), and the central character is \(z\mapsto i^z\). Thus, for \(p=2\), the \(\mathbb F_p\)-central Heisenberg group used above is replaced by this \(\mathbb Z_4\)-central extension. The $p=2$ version of the finite Stone--von Neumann theorem~\cite{gurevich-hadani-char2} states that there is, up to unitary equivalence, a unique irreducible unitary representation of $H_2(V_n)$ with this central character. Applying this result to $a\mapsto A(a)^*w(a)$ gives the same conclusion as \cref{thm:integrability}.  
\end{remark}

\begin{cor}\label{cor:trivial-H1}
Every flat derivative field is gauge equivalent to the trivial field. Equivalently, \(H^1_\tau(V_n,\mathcal U(\mathcal H_n))=\{[A_I]\}\), where \(A_I:V_n\to \mathcal U(\mathcal H_n)\) is defined by $A_I(a)=I$ for all \(a\in V_n\).
\end{cor}

\begin{proof}
By \cref{thm:integrability} and \cref{rem:integrability-qubit-case}, every flat derivative field \(A\) has the form $A(a)=\tau_a(U)U^*$ for some \(U\in \mathcal U(\mathcal H_n)\). Since \(A_I(a)=I\), this can be
rewritten as $A(a)=\tau_a(U)A_I(a)U^*$ for all \(a\in V_n\). By \cref{ex:weyl-field-gauge}, \(A \sim A_I\). Hence, $H^1_\tau(V_n,\mathcal U(\mathcal H_n))=\{[A_I]\}$.
\end{proof}

Note that \cref{cor:trivial-H1} does not trivialize the classification problem for the Clifford hierarchy. At the \(k\)-th level, one requires not only flatness but also the pointwise condition \(A(a)\in\mathcal C^{(k-1)}\) for all \(a\in V_n\). Thus, the classification of the Clifford hierarchy is not an unrestricted problem in non-abelian cohomology, but rather the problem of classifying constrained non-abelian $1$-cocycles. We develop this viewpoint in \cref{abstract-class}.

\subsection{Classification Theorem}\label{abstract-class}
\cref{thm:integrability} shows that flat derivative fields can be identified with first order quantum derivatives, up to global phase. We now impose a constraint on the image of these maps to obtain an abstract classification of the Clifford hierarchy.

\begin{defn}\label{def:constrained-cocycle-space}
For $k\ge 1$, define
\begin{equation}\label{eqn:Z1k-def}
Z^{1,k}_\tau(V_n) :=
\{A\in Z^1_\tau(V_n,\mathcal U(\mathcal H_n)): A(a)\in \mathcal C^{(k-1)}\ \text{for every } a\in V_n \}.
\end{equation}
\end{defn}

\cref{thm:abstract-classification} establishes our abstract classification of the Clifford hierarchy. It is based on the characterization of membership in the Clifford hierarchy using quantum derivatives in \cref{derivative-characterization}.

\begin{thm}\label{thm:abstract-classification}
For every \(k\ge 1\), the map 
\begin{align}
\Phi_k:\mathcal{C}^{(k)}&\longrightarrow Z^{1,k}_\tau(V_n), \label{eqn:abstract-classification-map}\\
U&\longmapsto A_U.
\end{align}
induces a bijection $\overline{\mathcal C}^{(k)}\cong Z^{1,k}_\tau(V_n)$. 
\end{thm}

\begin{proof}
We first check that $\Phi_k$ is well-defined. Let \(U\in \mathcal C^{(k)}\). By \cref{derivative-characterization}, we have \(A_U(a)=\partial_aU\in \mathcal C^{(k-1)}\) for every \(a\in V_n\). By \cref{lem:derivative-field-cocycle-identity}, \(A_U\) is flat. Therefore, $A_U\in Z^{1,k}_\tau(V_n)$. Moreover, if \(U'=\lambda U\) for some \(\lambda\in U(1)\), then $A_{U'}(a)=\tau_a(\lambda U)(\lambda U)^*=A_U(a)$. Thus \(\Phi_k\) is constant modulo global phases and descends to a map
\begin{align}
\overline{\Phi}_k:\overline{\mathcal C}^{(k)}&\longrightarrow Z^{1,k}_\tau(V_n) \\
[U] &\longmapsto A_U
\end{align}
We first prove injectivity. Suppose \(\overline{\Phi}_k([U])=\overline{\Phi}_k([V])\). Then \(A_U=A_V\).  By uniqueness in \cref{thm:integrability} and \cref{rem:integrability-qubit-case}, \(U\) and \(V\) differ by a global phase. Hence \([U]=[V]\) in \(\overline{\mathcal C}^{(k)}\). We now prove surjectivity. Let \(A\in Z^{1,k}_\tau(V_n)\). By \cref{thm:integrability} and \cref{rem:integrability-qubit-case}, there exists a unitary \(U\in \mathcal U(\mathcal H_n)\), unique up to global phase, such that
$A(a)=\tau_a(U)U^*=A_U(a)$
for every \(a\in V_n\). Since \(A\in Z^{1,k}_\tau(V_n)\), we have \(A(a)\in \mathcal C^{(k-1)}\) for every \(a\in V_n\). Therefore, \(A_U(a)=\partial_aU\in \mathcal C^{(k-1)}\) for every \(a\in V_n\). By \cref{derivative-characterization}, \(U\in \mathcal C^{(k)}\). Hence, \(A=\Phi_k(U)\).
\end{proof}

The proof of \cref{thm:abstract-classification} crucially uses \cref{thm:integrability} (and \cref{rem:integrability-qubit-case}), which in turn relies on the fact that the map \(B_A\), defined by \(B_A(a)=A(a)^*w(a)\), is a projective representation unitarily equivalent to \(a\mapsto w(a)\). This allows us to reformulate \cref{thm:abstract-classification} in terms of \(B_A\), replacing the cocycle condition on \(A\) by a projective representation condition. For an odd prime \(p\), define
\begin{equation}
Y^{k}_{\omega}(V_n):=
\{
B:V_n\mapsto \mathcal C^{(k-1)} : B(a)B(b)=\omega^{\frac12[a,b]}B(a+b)
\ \text{for all }a,b\in V_n
\}.
\end{equation}
We call an element of \(Y^{k}_{\omega}(V_n)\) a \(\mathcal C^{(k-1)}\)-valued projective Weyl field.

\begin{cor}\label{cor:equiv-abstract}
Assume that \(p\) is an odd prime number. For every \(k\geq 2\), the assignment 
\begin{align} \overline \Theta_k:\overline{\mathcal C}^{(k)} &\longrightarrow Y^{k}_\omega(V_n), \\ [U]&\longmapsto B_U, \end{align}
where $B_U(a):=Uw(a)U^*$ induces a bijection $\overline{\mathcal C}^{(k)} \cong Y^{k}_\omega(V_n)$.  Equivalently, \(\overline{\mathcal C}^{(k)}\) is identified with the set of all \(\mathcal C^{(k-1)}\)-valued projective Weyl fields.
\end{cor}

\begin{proof}
For \(A\in Z^{1,k}_{\tau}(V_n)\), we define \(B_A\) as above. Conversely, for \(B\in Y^{k}_\omega(V_n)\), define \(A_B:V_n\to\mathcal U(\mathcal H_n)\) by \(A_B(a)=w(a)B(a)^*\) for all \(a\in V_n\). We show that the maps
\begin{equation}
\begin{aligned}
f : Z^{1,k}_{\tau}(V_n) &\longrightarrow Y^{k}_\omega(V_n) \\
A &\longmapsto B_A
\end{aligned}
\qquad\qquad
\begin{aligned}
g : Y^{k}_\omega(V_n) &\longrightarrow Z^{1,k}_{\tau}(V_n) \\
B &\longmapsto A_B
\end{aligned}
\end{equation}
are inverse to each other. We first show that $f$ is well-defined. Let $A \in Z^{1,k}_{\tau}(V_n)$. By \cref{thm:integrability}, $B_A$ is indeed a projective representation with multiplier $\omega^{\frac{1}{2}[a,b]}$. Putting $b=-a$ in the cocycle identity, we obtain $A(a)^*=\tau_a(A(-a))$ for all $a \in V_n$. Therefore, we have
\begin{align}
B_A(a)=A(a)^*w(a)=\tau_a(A(-a))w(a)&=w(a)A(-a).
\end{align}
Since $A(-a)\in \mathcal C^{(k-1)}$ and $\mathcal C^{(k-1)}$ is stable under left multiplication by Weyl operators by \cref{clifford-props-no-proof}, it follows that $B_A(a)\in \mathcal C^{(k-1)}$ for every $a\in V_n$. Hence, $f$ is well-defined. We now show that $g$ is well-defined.
Let $B \in Y^{k}_\omega(V_n)$. We have $B(0)=I$ and $B(a)B(-a)=I$ for all $a \in V_n$. Hence, $B(a)^*=B(-a)$ for all $a \in V_n$. Therefore, we have
\begin{align}
A_B(a)=w(a)B(a)^*=w(a)B(-a).
\end{align}
Since $B(-a)\in \mathcal C^{(k-1)}$ and $\mathcal C^{(k-1)}$ is stable under left multiplication by Weyl operators by \cref{clifford-props-no-proof}, it follows that $A_B(a)\in \mathcal C^{(k-1)}$ for every $a\in V_n$. We now check flatness. Taking adjoints in the defining relation for the projective representation, we obtain
\begin{equation}\label{eqn:B-projective-adjoint}
B(a+b)^*=\omega^{\frac12[a,b]}B(b)^*B(a)^*.
\end{equation}
for all $a,b \in V_n$. Therefore, we have
\begin{align}
A_B(a+b)
&=w(a+b)B(a+b)^* \\
&=\omega^{\frac12[a,b]}w(a+b)B(b)^*B(a)^* =w(a)w(b)B(b)^*B(a)^*=\tau_a(A_B(b))A_B(a).
\end{align}
Since $V_n$ is abelian, this is equivalent to the cocycle identity. Therefore, $A_B\in Z^{1,k}_{\tau}(V_n)$. Hence, $g$ is well-defined. Note that $g \circ f = \operatorname{Id}_{Z^{1,k}_{\tau}(V_n)}$. Indeed, for \(A\in Z^{1,k}_{\tau}(V_n)\) we have
\begin{align}
A_{B_A}(a)&=w(a)B_A(a)^*\\&=w(a)(A(a)^*w(a))^*=w(a)w(a)^*A(a)=A(a).
\end{align}
for all $a \in V_n$. Hence, $g \circ f = \operatorname{Id}_{Z^{1,k}_{\tau}(V_n)}$. A similar argument shows that $f \circ g = \operatorname{Id}_{Y^{k}_\omega(V_n)}$. Hence, we have $Z^{1,k}_{\tau}(V_n) \cong Y^{k}_\omega(V_n)$. By \cref{thm:abstract-classification}, we have $\mathcal{C}^{(k)}/U(1) \cong Y^{k}_\omega(V_n)$. Explicitly, the bijection is given by the map
\begin{align}
\overline \Theta_k:\overline{\mathcal C}^{(k)}&\longrightarrow Y^{k}_\omega(V_n), \\
[U]&\longmapsto B_{A_U}
\end{align}
such that
\begin{align}
B_{A_U}(a)
&=A_U(a)^*w(a) \\
&=(w(a)Uw(a)^*U^*)^*w(a) \\
&=Uw(a)U^*.
\end{align}
for all $a \in V_n$. This completes the proof.
\end{proof}

\begin{remark}\label{rem:equiv-abstract-qubit}
When \(p=2\), one replaces the multiplier \(\omega^{\frac12[a,b]}\) by the multiplier \(i^{\mu(a,b)}\) from \cref{qubit-mult-comm}. Moreover, for \(k\geq2\) define \(Y^{k}_2(V_n)\) to be the set of maps \(B:V_n\to\mathcal C^{(k-1)}\) satisfying $B(a)B(b)=i^{\mu(a,b)}B(a+b)$ for all \(a,b\in V_n\). The proof of \cref{cor:equiv-abstract} is otherwise unchanged.
\end{remark}

\subsection{Low Order Consequences}\label{sec:low-order}
We now recover the standard descriptions of \(\mathcal C^{(1)}\) and \(\mathcal C^{(2)}\) from \cref{thm:abstract-classification}. This provides a consistency check for the formalism we applying it to the third level in \cref{third-clifford}.

\subsubsection{First Clifford Level} 
By \cref{clifford-def}, we have $\mathcal{C}^{(1)}/U(1) \cong V_n$. We now recover this definition as a consequence of our classification theorem by explicitly computing $Z^{1,1}_\tau(V_n)$.

\begin{prop}\label{level-one-computation}
The map
\begin{align}
f :V_n&\longrightarrow Z^{1,1}_\tau(V_n), \label{eqn:level-one-theta-map}\\
v&\longmapsto A^{(v)}
\end{align}
where $A^{(v)}(a):=\omega^{[a,v]}I$ for all $a \in V_n$.
Consequently, we have $\overline{\mathcal C}^{(1)}\cong V_n$. 
\end{prop}

\begin{proof}
Let \(A\in Z^{1,1}_\tau(V_n)\). By definition, \(A(a)\in \mathcal C^{(0)}=U(1)I\) for every \(a\in V_n\). Hence, there exists a function \(\chi:V_n\to U(1)\) such that
$A(a)=\chi(a)I$. Since \(A\) is flat, we have $A(a+b)=\tau_a(A(b))A(a)$. 
Substituting \(A(a)=\chi(a)I\) gives 
$\chi(a+b)=\chi(a)\chi(b)$
for all $a,b\in V_n$. 
Thus, \(\chi\) is a character of the additive group \(V_n\). Since \(V_n\) is an \(\mathbb F_p\)-vector space, every such character is of the form $\chi(a)=\omega^{\ell(a)}$ for a unique linear functional \(\ell:V_n\to \mathbb F_p\). 
By non-degeneracy of the symplectic form, the map
\begin{align}
V_n &\longrightarrow V_n^*, \\
v &\longmapsto (a \longmapsto [a,v])
\end{align}
is an isomorphism. Therefore there exists a unique \(v\in V_n\) such that $\ell(a) = [a,v]$ for all $a\in V_n$. Therefore, \(A=A^{(v)}\) so \(f\) is surjective. Suppose \(A^{(v)}=A^{(v')}\). Then $\omega^{[a,v]}=\omega^{[a,v']}$ for all $a\in V_n$. Equivalently, $[a,v-v']=0$ for all $a\in V_n$. By non-degeneracy of the symplectic form, \(v-v'=0\). Hence, \(v=v'\) so \(f\) is injective. Thus,  \(Z^{1,1}_\tau(V_n)\cong V_n\). By \cref{thm:abstract-classification}, we have $\overline{\mathcal C}^{(1)}\cong V_n$. 
\end{proof}

\subsubsection{Second Clifford Level}\label{subsec:second-clifford-level}
We now specialize to the second level of the Clifford hierarchy. In this case, the derivative field is Weyl-valued up to phase, so the flatness condition reduces to explicit algebraic constraints on a Weyl-label map and a phase function.
This is made precise in \cref{Pauli-cocycle-equations}.

\begin{lem}\label{Pauli-cocycle-equations}
Let \(A:V_n\to \mathcal C^{(1)}\) be a map. Write $A(a)=\beta(a)w(\ell(a))$ for $\beta(a)\in U(1)$ and  $\ell(a)\in V_n$ for all $a \in V_n$. Then \(A\in Z^{1,2}_\tau(V_n)\) if and only if the following conditions hold:
\begin{enumerate}
    \item $\ell(a+b)=\ell(a)+\ell(b)$ for all $a,b\in V_n$. 
    \item If \(p\) is odd, then
    \begin{equation}\label{eqn:beta-cocycle-odd}
    \beta(a+b)=\beta(a)\beta(b)\omega^{[a,\ell(b)]+\frac12[\ell(b),\ell(a)]}
    \end{equation}
    for all $a,b\in V_n$. If \(p=2\), then
    \begin{equation}\label{eqn:beta-cocycle-qubit}
    \beta(a+b)=\beta(a)\beta(b)(-1)^{[a,\ell(b)]}i^{\mu(\ell(b),\ell(a))}
    \end{equation}
    for all $a,b\in V_n$. 
\end{enumerate}
\end{lem}

\begin{proof}
If $p$ is odd, using \cref{tau-on-Weyl} and \cref{weyl-mul-law} we have
\begin{align}
\beta(a+b) w(\ell(a+b))
&=: A(a+b) \\
&= \tau_a(A(b)) A(a) \\
&= \beta(a)\beta(b)\omega^{[a,\ell(b)]} w(\ell(b))w(\ell(a))\\
&= \beta(a)\beta(b)
   \omega^{[a,\ell(b)] + \frac{1}{2}[\ell(b),\ell(a)]}
   w(\ell(b)+\ell(a)).
\end{align}
Similarly, if \(p=2\) then \cref{qubit-mult-comm} gives
\begin{align}
A(a+b)
&:= \beta(a+b)w(\ell(a+b)) \\
&= \tau_a(A(b))A(a) = \beta(a)\beta(b)(-1)^{[a,\ell(b)]}i^{\mu(\ell(b),\ell(a))}w(\ell(b)+\ell(a)).
\end{align}
Invoking \cref{lem:Weyl-label-phase-uniqueness} shows that $\ell$ is additive and proves \cref{eqn:beta-cocycle-odd} and \cref{eqn:beta-cocycle-qubit}. Conversely, suppose that \(\ell\) and \(\beta\) satisfy (1) and (2). We show that \(A\) is flat. If \(p\) is odd, we have
\begin{align}
\tau_a(A(b))A(a)
&= \beta(a)\beta(b)\,
   \omega^{[a,\ell(b)]+\frac{1}{2}[\ell(b),\ell(a)]}
   w(\ell(b)+\ell(a)) \\& 
   = \beta(a+b)w(\ell(a+b))= A(a+b).
\end{align}
Thus, \(A(a+b)=\tau_a(A(b))A(a)\) for all \(a,b\in V_n\).  A similar argument works for the case \(p=2\). Therefore, \(A\in Z^{1,2}_\tau(V_n)\). This completes the proof.
\end{proof}

\cref{Pauli-cocycle-equations} is the first point at which \cref{thm:abstract-classification} reduces the classification problem to a concrete system of algebraic equations.  \cref{lem:clifford-symplectic-consequence} and \cref{second-level-symplectic-quotient} show that these conditions recover the standard description of \(\mathcal C^{(2)}\).

\begin{lem}\label{lem:clifford-symplectic-consequence}
Let \(C\in \mathcal C^{(2)}\), and write $A_C(a)=\beta(a)w(\ell(a))$. 
Then \(\ell:V_n\to V_n\) is \(\mathbb F_p\)-linear. Moreover, let \(S_C:V_n\to V_n\) be the Weyl-label map determined by
\begin{equation}\label{eqn:UwaUstar-S}
Cw(a)C^*=\gamma(a)w(S_Ca).
\end{equation}
Then \(S_C\) is \(\mathbb F_p\)-linear and symplectic, and
$\ell=\operatorname{Id}_{V_n} - S_C$.
\end{lem}

\begin{proof}
Since \(C\in \mathcal C^{(2)}\), \cref{derivative-characterization} implies that \(A_C(a)\in \mathcal C^{(1)}\) for every \(a\in V_n\). Therefore, \(A_C\in Z^{1,2}_\tau(V_n)\). By \cref{Pauli-cocycle-equations}, \(\ell\) is additive. Since \(V_n\) is a vector space over \(\mathbb F_p\), every additive map \(V_n\to V_n\) is \(\mathbb F_p\)-linear. Hence, \(\ell\) is \(\mathbb F_p\)-linear. 
Define \(S_C\) by
\begin{equation}\label{eqn:S-label-definition}
Cw(a)C^*=\gamma(a)w(S_Ca), \quad \gamma(a) \in U(1).
\end{equation}
Such an \(S_C\) exists because \(C\in \mathcal C^{(2)}\). Since $Cw(a)w(b)C^* = (Cw(a)C^*)(Cw(b)C^*)$, comparison of both sides gives $S_C(a+b)=S_Ca+S_Cb$. Thus, \(S_C\) is \(\mathbb F_p\)-linear. Next, we show that \(S_C\) preserves the symplectic form. Since $w(a)w(b)w(a)^*w(b)^*=\omega^{[a,b]}I$ for all $a,b \in V_n$, we have
\begin{equation}\label{eqn:S-symplectic-proof-1}
Cw(a)w(b)w(a)^*w(b)^*C^*=\omega^{[a,b]}I.
\end{equation}
Using \(Cw(a)C^*=\gamma(a)w(S_Ca)\) and \(Cw(b)C^*=\gamma(b)w(S_Cb)\) the left-hand side becomes the commutator of \(\gamma(a)w(S_Ca)\) and \(\gamma(b)w(S_Cb)\). The scalar phases cancel in the commutator, so
\begin{align}
Cw(a)w(b)w(a)^*w(b)^*C^*
&= w(S_Ca)w(S_Cb)w(S_Ca)^*w(S_Cb)^*= \omega^{[S_Ca,S_Cb]} I.
\end{align}
Hence, $\omega^{[S_Ca,S_Cb]}=\omega^{[a,b]}$ for all $a,b\in V_n$. Since \(\omega\) is a primitive \(p\)-th root of unity, it follows that $[S_Ca,S_Cb]=[a,b]$ for all \(a,b\in V_n\). Thus, $S_C\in \operatorname{Sp}(V_n)$. Finally, note that
\begin{align}
\label{eqn:ell-S-proof-1}
A_C(a)= w(a)Cw(a)^*C^* \propto w(a)w(-S_Ca)\propto w(a-S_Ca),
\end{align}
where the proportionality factor is a $U(1)$ phase depending on whether $p=2$ or $p$ is odd. Uniqueness of Weyl labels gives $\ell(a)=a-S_Ca$. Hence,  $\ell=\operatorname{Id}_{V_n}-S_C$.
\end{proof}

\begin{prop}\label{second-level-symplectic-quotient}
Let \(p\) be an odd prime. The following is a short exact sequence of groups
\begin{equation}\label{eqn:clifford-symplectic-exact-sequence}
1
\longrightarrow
\mathcal C^{(1)}
\longrightarrow
\mathcal C^{(2)}
\longrightarrow
\operatorname{Sp}(V_n)
\longrightarrow
1.
\end{equation}
Consequently, $\mathcal C^{(2)}/\mathcal C^{(1)} \cong \operatorname{Sp}(V_n).$
\end{prop}

\begin{proof}
Define $\Psi:\mathcal C^{(2)} \rightarrow \operatorname{Sp}(V_n)$ by \(\Psi(C)=S_C\), where $Cw(a)C^*=\gamma_C(a)w(S_Ca).$ By \cref{lem:clifford-symplectic-consequence}, \(S_C\in\operatorname{Sp}(V_n)\). Let \(C,D\in\mathcal C^{(2)}\) and \(a\in V_n\). We have
\begin{align}
\gamma_{CD}(a)w(S_{CD}a)
&=(CD)w(a)(CD)^*\\
&=C(Dw(a)D^*)C^*\\
&=C(\gamma_D(a)w(S_Da))C^*\\
&=\gamma_D(a)\gamma_C(S_Da)w(S_C(S_Da)).
\end{align}
By uniqueness, we have $S_{CD}(a)=S_C(S_Da)$ for all \(a\in V_n\). Therefore, $S_{CD}=S_C\circ S_D,$ so \(\Psi\) is a group homomorphism. We now identify \(\ker\Psi\). Suppose \(C\in\ker\Psi\). Then \(S_C=\operatorname{Id}_{V_n}\). Hence, we have $Cw(a)C^*=\gamma_C(a)w(a)$ for every \(a\in V_n\). Therefore,
\begin{equation}\label{eqn:kernel-derivative-scalar}
\partial_aC = w(a)Cw(a)^*C^* \in U(1)I
\end{equation}
for every \(a\in V_n\). By \cref{derivative-characterization} with \(k=1\), we have \(C\in\mathcal C^{(1)}\). Conversely, suppose that \(C\in\mathcal C^{(1)}\). Then $C=\lambda w(v)$ for some \(\lambda\in U(1)\) and \(v\in V_n\). By \cref{weyl-mul-law}, conjugation by \(w(v)\) only multiplies each \(w(a)\) by a scalar phase. Hence, $S_C=\operatorname{Id}_{V_n}.$ Therefore, $\ker\Psi=\mathcal C^{(1)}.$ It remains to prove surjectivity. Let \(S\in\operatorname{Sp}(V_n)\) and define $w_S(a):=w(Sa).$ Since \(p\) is odd and \(S\) is symplectic, it follows from
\cref{weyl-mul-law} that
\begin{align}
w_S(a)w_S(b)
&=w(Sa)w(Sb)\\
&=\omega^{\frac12[Sa,Sb]}w(Sa+Sb)\\
&=\omega^{\frac12[a,b]}w(S(a+b))=\omega^{\frac12[a,b]}w_S(a+b).
\end{align}
Thus, \(a\mapsto w_S(a)\) is a projective representation of \(V_n\) with the same multiplier as \(a\mapsto w(a)\). Hence, it lifts to an ordinary representation of \(H(V_n)\) with the same central character as the standard Weyl representation. By \cref{fact:finite-stone-von-neumann}, there exists
\(C_S\in\mathcal U(\mathcal H_n)\) such that
$C_Sw(a)C_S^*=w(Sa)$
for every \(a\in V_n\). Hence, \(C_S\in\mathcal C^{(2)}\) and $\Psi(C_S)=S.$ Hence, \(\Psi\) is surjective. 
\end{proof}

\begin{remark}
When \(p=2\), \cref{eqn:clifford-symplectic-exact-sequence} continues to hold after adapting it to the qubit case.
\end{remark}

Assume for the remainder of \cref{subsec:second-clifford-level} that \(p\) is odd. In \cref{second-level-symplectic-quotient}, we showed that the maps \(a\mapsto w(a)\) and \(a\mapsto w(Sa)\) are projective representations with the same multiplier. For \(S\in \operatorname{Sp}(V_n)\), the Stone--von Neumann theorem implies   there exists \(M_S\in\mathcal C^{(2)}\) such that
\begin{equation}
M_Sw(x)M_S^*=w(Sx)
\end{equation}
for every \(x\in V_n\). We fix such a choice \(S\mapsto M_S\), with \(M_{\operatorname{Id}_{V_n}}=I\). Let \(S,T\in\operatorname{Sp}(V_n)\). Since \(\Psi\) is a group homomorphism, we have $\Psi(M_SM_T)=ST=\Psi(M_{ST}).$ Hence,
\begin{equation}
M_SM_TM_{ST}^*\in\ker\Psi=\mathcal C^{(1)}.
\end{equation}
Therefore, there exist \(\kappa(S,T)\in U(1)\) and \(c(S,T)\in V_n\) such that
\begin{equation}
\label{eqn:odd-pre-kappa-product}
M_SM_T = \kappa(S,T)w(c(S,T))M_{ST}.
\end{equation}
In fact, the term \(c(S,T)\) necessarily vanishes. Indeed, we have \(M_SM_Tw(x)M_T^*M_S^*=w(STx)\) for every \(x\in V_n\).  On the other hand, using \cref{eqn:odd-pre-kappa-product}, we obtain
\begin{align}
M_SM_Tw(x)M_T^*M_S^*
&=
w(c(S,T))M_{ST}w(x)M_{ST}^*w(c(S,T))^*\\
&=
w(c(S,T))w(STx)w(c(S,T))^*\\
&=
\omega^{[c(S,T),STx]}w(STx).
\end{align}
Comparing the two expressions yields $\omega^{[c(S,T),STx]}=1$ for every \(x\in V_n\). Since \(\omega\) is a primitive \(p\)-th root of unity, this implies $[c(S,T),STx]=0$ for every \(x\in V_n\). As \(ST\) is invertible, \(STx\) ranges over all of \(V_n\), and hence $[c(S,T),y]=0$ for every \(y\in V_n\). By non-degeneracy of the symplectic form, $c(S,T)=0.$ Consequently,
\begin{equation}
\label{eqn:kappa-third-level-definition}
M_SM_T=\kappa(S,T)M_{ST}.
\end{equation}
Associativity implies that
\(\kappa:\operatorname{Sp}(V_n)\times\operatorname{Sp}(V_n)\to U(1)\) is a \(2\)-cocycle. Since \(M_{\operatorname{Id}_{V_n}}=I\), it is normalized. Although the cocycle \(\kappa\) depends on the chosen representatives \(S\mapsto M_S\), its cohomology class does not. This is made precise in \cref{lem:kappa-cohomology-class}.

\begin{lem}\label{lem:kappa-cohomology-class}
Let \(p\) be an odd prime. The normalized \(2\)-cocycle
\begin{equation}
\kappa:\operatorname{Sp}(V_n)\times\operatorname{Sp}(V_n)\longrightarrow U(1)
\end{equation}
determines a cohomology class $[\kappa]\in H^2(\operatorname{Sp}(V_n),U(1))$ that is independent of the choice of representatives \(S\mapsto M_S\) 
for every \(S\in\operatorname{Sp}(V_n)\). 
\end{lem}

\begin{proof}
Let \(S\mapsto M_S\) and \(S\mapsto M'_S\) be two normalized choices such that \(M_{\operatorname{Id}_{V_n}}=M'_{\operatorname{Id}_{V_n}}=I\). For every \(S\in\operatorname{Sp}(V_n)\), we have
\begin{align}
M'_SM_S^*w(x)M_SM_S^{\prime *}
&=M'_Sw(S^{-1}x)M_S^{\prime *}=w(x)
\end{align}
for every \(x\in V_n\). Hence, \(M'_SM_S^*\) commutes with every Weyl operator. By \cref{cor:tau-invariants-are-scalars}, there exists \(\eta(S)\in U(1)\) such that $M'_S=\eta(S)M_S$.  Since both choices are normalized, \(\eta(\operatorname{Id}_{V_n})=1\). Let \(\kappa\) and \(\kappa'\) denote the corresponding normalized \(2\)-cocycles. We obtain
\begin{align}
\kappa'(S,T)\eta(ST)M_{ST}
&=\kappa'(S,T)M'_{ST}\\
&=M'_SM'_T\\
&=\eta(S)\eta(T)M_SM_T=\eta(S)\eta(T)\kappa(S,T)M_{ST}.
\end{align}
Therefore,
\begin{equation}\label{eqn:kappa-coboundary-change}
\kappa'(S,T)=\frac{\eta(S)\eta(T)}{\eta(ST)}\kappa(S,T).
\end{equation}
Thus, \(\kappa'=\kappa\cdot\delta\eta\), so \(\kappa\) and \(\kappa'\) determine the same class in \(H^2(\operatorname{Sp}(V_n),U(1))\).
\end{proof}

The class \([\kappa]\) is the multiplier class of the projective symplectic representation \(S\mapsto[M_S]\). Indeed, \cref{eqn:kappa-third-level-definition} implies that $[M_S][M_T]=[M_{ST}]$
in the projective unitary group.  

\begin{remark}\label{p2-clifford-remark}
When \(p=2\), the choice of representatives \(M_S\) requires a modification. The reason is that, for \(S\in\operatorname{Sp}(V_n)\), the projective representations \(a\mapsto w(a)\) and \(a\mapsto w(Sa)\) need not have the same multiplier. The details are discussed in \cref{appendix:characteristic-two-third-level}. As in the odd prime case, \cref{eqn:odd-pre-kappa-product} continues to hold when \(p=2\). However, in the odd prime case we showed that the representatives may be chosen so that \(c(S,T)=0\), whereas this need not hold a priori when \(p=2\). This is why the discussion above and the parametrization in \cref{third-clifford} are developed under the assumption that \(p\) is odd. The discussion of \(\kappa\) remains unchanged.
\end{remark}


\begin{lem}\label{clifford-normal-form}
Let $p$ be an odd prime. Every \(C\in\mathcal C^{(2)}\) determines a unique triple \((\alpha_C,q_C,S_C)\in U(1)\times V_n\times\operatorname{Sp}(V_n)\) such that
\begin{equation}
C=\alpha_Cw(q_C)M_{S_C},
\end{equation}
for the fixed choice \(S\mapsto M_S\) above.
\end{lem}

\begin{proof}
Let \(\Psi:\mathcal C^{(2)}\to\operatorname{Sp}(V_n)\) be the map from \cref{second-level-symplectic-quotient}. For \(C\in\mathcal C^{(2)}\), define \(S_C:=\Psi(C)\). By the choice of \(M_{S_C}\), we have \(\Psi(M_{S_C})=S_C\). Since \(\Psi\) is a group
homomorphism and \(M_{S_C}^*=M_{S_C}^{-1}\), we have
\begin{equation}
\Psi(CM_{S_C}^*) = \Psi(C)\Psi(M_{S_C}^*) = S_CS_C^{-1} = \operatorname{Id}_{V_n}.
\end{equation}
Hence, \(CM_{S_C}^*\in\ker\Psi=\mathcal C^{(1)}\). Therefore, there exist \(\alpha_C\in U(1)\) and \(q_C\in V_n\) such that $CM_{S_C}^*=\alpha_Cw(q_C).$ Multiplying on the right by \(M_{S_C}\), we obtain $C=\alpha_Cw(q_C)M_{S_C}$. It remains to prove uniqueness. Suppose that \(C=\alpha w(q)M_S\) for some \(\alpha\in U(1)\), \(q\in V_n\), and \(S\in\operatorname{Sp}(V_n)\). Applying \(\Psi\) gives \(\Psi(C)=S\), since \(\alpha I\) and \(w(q)\) lie in \(\mathcal C^{(1)}=\ker\Psi\). Thus, \(S=S_C\). Therefore, we have $\alpha w(q)M_{S_C} = \alpha_Cw(q_C)M_{S_C},$ which implies that \(\alpha w(q)=\alpha_Cw(q_C)\). By \cref{lem:Weyl-label-phase-uniqueness}, we have \(q=q_C\) and \(\alpha=\alpha_C\).
\end{proof}

\begin{remark}
\cref{clifford-normal-form} continues to hold for \(p=2\), with essentially the same argument. The only difference lies in the multiplication law for the chosen representatives \(M_S\).
\end{remark}

 The normal form in \cref{clifford-normal-form} identifies \(\overline{\mathcal C}^{(2)}\) with the affine symplectic group \(\operatorname{ASp}(V_n):=V_n\rtimes\operatorname{Sp}(V_n)\), whose multiplication is given by $(q,S)(r,T):=(q+Sr,ST).$  

\begin{cor}\label{cor:projective-clifford-affine-symplectic}
Let \(p\) be an odd prime. The map
\begin{align}
\Xi:\operatorname{ASp}(V_n)&\longrightarrow\overline{\mathcal C}^{(2)},\label{eqn:projective-clifford-asp-map}\\
(q,S)&\longmapsto[w(q)M_S]
\end{align}
is a group isomorphism. In particular, there is a split short exact sequence
\begin{equation}\label{eqn:projective-clifford-asp-exact-sequence}
1\longrightarrow V_n\longrightarrow\overline{\mathcal C}^{(2)}\longrightarrow\operatorname{Sp}(V_n)\longrightarrow1.
\end{equation}
\end{cor}

\begin{proof}
We first show that \(\Xi\) is a homomorphism. Let \((q,S),(r,T)\in\operatorname{ASp}(V_n)\). We obtain
\begin{align}
[w(q)M_S][w(r)M_T]
&=[w(q)M_Sw(r)M_T]\\
&=[w(q)w(Sr)M_SM_T]\\
&=[\kappa(S,T)\omega^{\frac12[q,Sr]}w(q+Sr)M_{ST}]\\
&=[w(q+Sr)M_{ST}]\\
&=\Xi(q+Sr,ST).
\end{align}
This is precisely $\Xi((q,S)(r,T))=\Xi(q,S)\Xi(r,T).$
Thus, \(\Xi\) is a group homomorphism. We next prove surjectivity. Let \([C]\in\overline{\mathcal C}^{(2)}\). By \cref{clifford-normal-form}, there exist unique \(\alpha_C\in U(1)\), \(q_C\in V_n\), and \(S_C\in\operatorname{Sp}(V_n)\) such that  
$C=\alpha_Cw(q_C)M_{S_C}.$
Passing to projective classes gives
\begin{equation}
[C]=[w(q_C)M_{S_C}]=\Xi(q_C,S_C),
\end{equation}
so \(\Xi\) is surjective. Finally, suppose that $\Xi(q,S)=\Xi(r,T).$ Then there exists \(\lambda\in U(1)\) such that
$w(q)M_S=\lambda w(r)M_T.$
By \cref{clifford-normal-form}, we have \(S=T\) and \(q=r\). Hence, \(\Xi\) is injective and therefore an isomorphism. Under this identification, the inclusion \(V_n\to\overline{\mathcal C}^{(2)}\) is given by $q\mapsto[w(q)],$ and the projection \(\overline{\mathcal C}^{(2)}\to\operatorname{Sp}(V_n)\) is given by $[w(q)M_S]\mapsto S.$ The map $S\mapsto[M_S]$ is a section of this projection because \([M_SM_T]=[M_{ST}]\). Hence, \cref{eqn:projective-clifford-asp-exact-sequence} is split.
\end{proof}

\section{Third Clifford Level}\label{third-clifford}
We now apply the formalism developed in \cref{formal-framework} to the third level of the Clifford hierarchy. In \cref{third-level-linear-algebraic}, we describe a gate in $\mathcal C^{(3)}$ in terms of symplectic, affine, and phase data $(S,q,\alpha)$. In \cref{cohom-inter}, we give the $H^1$ and $H^2$ interpretation of these data, and in \cref{generator-level-criteria-third}, we translate the resulting conditions into a finite number of equations. We then return to the cohomological arguments in \cref{cohom-lifting-third} and obtain a cohomological moduli description of $\mathcal C^{(3)}$. Throughout this section, we assume that $p$ is an odd prime. We discuss the case $p=2$ case in \cref{appendix:characteristic-two-third-level}.

\subsection{Affine, Symplectic and Phase Data}\label{third-level-linear-algebraic}
Applying \cref{cor:equiv-abstract} with \(k=3\), we have
\begin{equation}\label{eqn:third-level-projective-weyl-fields}
\overline{\mathcal C}^{(3)} 
\cong
\{B:V_n \rightarrow \mathcal C^{(2)} : B(a)B(b)=\omega^{\frac12[a,b]}B(a+b)\ \text{for every } a,b \in V_n \} := Y^3_{\omega}(V_n). 
\end{equation}
Thus, it suffices to describe the \(\mathcal C^{(2)}\)-valued projective Weyl fields.  We do this in \cref{third-level-classification} by applying \cref{clifford-normal-form} pointwise.

\begin{prop}\label{third-level-classification}
Let \(p\) be an odd prime.
For each \(a\in V_n\), write \(B(a)\) as
\begin{equation}\label{eqn:third-level-B-normal-form}
B(a)=\alpha_B(a)w(q_B(a))M_{S_B(a)},
\end{equation}
where $\alpha_B(a)\in U(1)$, $q_B(a)\in V_n$ and $S_B(a)\in \operatorname{Sp}(V_n)$. Then \(B\in Y^3_\omega(V_n)\) if and only if for all \(a,b\in V_n\) we have
\begin{align}
S_B(a+b)&=S_B(a)S_B(b), \label{eqn:third-level-S-equation}\\
q_B(a+b)&=q_B(a)+S_B(a)q_B(b), \label{eqn:third-level-q-equation}\\
\alpha_B(a+b)
&=\alpha_B(a)\alpha_B(b)\kappa(S_B(a),S_B(b))
\omega^{\frac12\left([q_B(a),S_B(a)q_B(b)]-[a,b]\right)}.
\label{eqn:third-level-alpha-equation}
\end{align}
\end{prop}

\begin{proof}
Fix a lifting map $S \mapsto M_S$ appearing in \cref{clifford-normal-form}. Then every element of \(\mathcal C^{(2)}\) has a unique expression of the form $\alpha w(q)M_S$ for some $\alpha\in U(1),\ q\in V_n,\ S\in \operatorname{Sp}(V_n)$. Consequently, for every map \(B:V_n\to \mathcal C^{(2)}\), the functions
\begin{equation}
\alpha_B:V_n\longrightarrow U(1),\qquad
q_B:V_n\longrightarrow V_n,\qquad
S_B:V_n\longrightarrow \operatorname{Sp}(V_n)
\end{equation}
defined by $B(a)=\alpha_B(a)w(q_B(a))M_{S_B(a)}$ are uniquely determined by \(B\). Using \cref{weyl-mul-law} and
\cref{eqn:kappa-third-level-definition}, we have
\begin{align}
B(a)B(b)
&=\alpha_B(a)\alpha_B(b)w(q_B(a))M_{S_B(a)}w(q_B(b))M_{S_B(b)} \\
&=\alpha_B(a)\alpha_B(b)w(q_B(a))w(S_B(a)q_B(b))M_{S_B(a)}M_{S_B(b)} \\
&=\alpha_B(a)\alpha_B(b)\kappa(S_B(a),S_B(b))
\omega^{\frac12[q_B(a),S_B(a)q_B(b)]}
w(q_B(a)+S_B(a)q_B(b))M_{S_B(a)S_B(b)} \notag,
\end{align}
where \(\kappa\) is defined at the end of \cref{sec:low-order}.
On the other hand, we have
\begin{align}
B(a)B(b)
&=\omega^{\frac12[a,b]}B(a+b)=\omega^{\frac12[a,b]}\alpha_B(a+b)w(q_B(a+b))M_{S_B(a+b)}.
\end{align}
Using the uniqueness in \cref{clifford-normal-form}, we obtain
\begin{align}
S_B(a+b)&=S_B(a)S_B(b),\\
q_B(a+b)&=q_B(a)+S_B(a)q_B(b), \\
\alpha_B(a)\alpha_B(b)\kappa(S_B(a),S_B(b))
\omega^{\frac12[q_B(a),S_B(a)q_B(b)]}
&= \omega^{\frac12 [a,b]} \alpha_B(a+b).
\end{align}
This completes the proof.
\end{proof}

\cref{third-level-classification} implies that \(\overline{\mathcal C}^{(3)} \) is in bijection with triples $(\alpha,q,S)$
where $\alpha:V_n\to U(1)$, $q:V_n\to V_n$ and $S:V_n\to \operatorname{Sp}(V_n)$ are maps satisfying
\begin{align}
S(a+b)&=S(a)S(b), \label{eqn:third-level-triple-S}\\
q(a+b)&=q(a)+S(a)q(b), \label{eqn:third-level-triple-q}\\
\alpha(a+b)
&=\alpha(a)\alpha(b)\kappa(S(a),S(b))
\omega^{\frac12\left([q(a),S(a)q(b)]-[a,b]\right)} \label{eqn:third-level-triple-alpha}
\end{align}
for all \(a,b\in V_n\). In fact, the tuple \((q,S)\) can be described in terms of the affine symplectic group \(\operatorname{ASp}(V_n):=V_n\rtimes \operatorname{Sp}(V_n)\). This is made precise in \cref{cor:third-level-affine-symplectic-homomorphism}. 

\begin{cor}\label{cor:third-level-affine-symplectic-homomorphism}
Let \(p\) be an odd prime.
Let \(S:V_n\to \operatorname{Sp}(V_n)\) and \(q:V_n\to V_n\) be maps. Then \cref{eqn:third-level-triple-S} and \cref{eqn:third-level-triple-q} hold if and only if the map
\begin{align}
\Phi_{q,S}: V_n &\longrightarrow \operatorname{ASp}(V_n), \\
a &\longmapsto  (q(a),S(a) )
\end{align}
is a group homomorphism.
\end{cor}

\begin{proof}
Note that 
\begin{equation}
\Phi_{q,S}(a)\Phi_{q,S}(b)=(q(a)+S(a)q(b),S(a)S(b)).
\end{equation}
Therefore, \(\Phi_{q,S}(a+b)=\Phi_{q,S}(a)\Phi_{q,S}(b)\) is equivalent to $S(a+b)=S(a)S(b)$ and $q(a+b)=q(a)+S(a)q(b)$ for all $a,b \in V_n$.
\end{proof}

\cref{cor:third-level-affine-symplectic-homomorphism} packages the symplectic and affine data \((S,q)\) into a single homomorphism to the affine symplectic group. The remaining phase data \(\alpha\) can be incorporated by passing to a central extension of this group, so that all three equations defining \((\alpha,q,S)\) admit a unified formulation. 
This is made precise in \cref{cor:third-level-central-extension-form}.

\begin{cor} \label{cor:third-level-central-extension-form}
Let \(p\) be an odd prime. Define $\widetilde{\operatorname{ASp}}(V_n):=U(1)\times V_n\times \operatorname{Sp}(V_n)$ with multiplication
\begin{equation}
\label{eqn:central-extension-asp-multiplication}
(\alpha,q,S)(\beta,r,T)= (\alpha\beta\,\kappa(S,T)\omega^{\frac12[q,Sr]},q+Sr,ST ).
\end{equation}
Then \(\widetilde{\operatorname{ASp}}(V_n)\) is a central extension
\begin{equation}
1\longrightarrow U(1)\longrightarrow \widetilde{\operatorname{ASp}}(V_n)\longrightarrow \operatorname{ASp}(V_n)\longrightarrow 1,
\end{equation}
where the inclusion \(U(1)\to \widetilde{\operatorname{ASp}}(V_n)\) is given by $\lambda \mapsto (\lambda,0,\operatorname{Id}_{V_n})$ and the projection \(\widetilde{\operatorname{ASp}}(V_n)\to \operatorname{ASp}(V_n)\) is given by $(\alpha,q,S) \mapsto (q,S)$. Moreover, let \(B:V_n\to \mathcal C^{(2)}\) be written in the normal form as in \cref{eqn:third-level-B-normal-form}. Then \(B\) belongs to \(Y^3_\omega(V_n)\) if and only if the associated map
\begin{align}
\widetilde{\Phi}_B:V_n&\longrightarrow \widetilde{\operatorname{ASp}}(V_n), \\
a&\longmapsto  (\alpha_B(a),q_B(a),S_B(a) )
\end{align}
satisfies $\widetilde{\Phi}_B(a)\widetilde{\Phi}_B(b)= (\omega^{\frac12[a,b]},0,\operatorname{Id}_{V_n} )\widetilde{\Phi}_B(a+b)$ for all \(a,b\in V_n\).
\end{cor}

\begin{proof}
Observe that the multiplication law on \(\widetilde{\operatorname{ASp}}(V_n)\) is induced by the normal form gates in $\mathcal{C}^{(2)}$. Indeed, using the multiplication law and the relation $M_SM_T=\kappa(S,T)M_{ST}$, we have 
\begin{align}
(\alpha w(q)M_S) \cdot (\beta w(r)M_T) &=\alpha\beta w(q)M_Sw(r)M_T \\ &=\alpha\beta w(q)w(Sr)M_SM_T=\alpha\beta\kappa(S,T)\omega^{\frac12[q,Sr]}w(q+Sr)M_{ST}.
\end{align}
Thus,  the product in \cref{eqn:central-extension-asp-multiplication} is associative, because it is transported from operator multiplication. The projection $(\alpha,q,S)\mapsto (q,S)$ is a group homomorphism, since the multiplication law on \(\operatorname{ASp}(V_n)\) is $(q,S)(r,T)=(q+Sr,ST)$. Its kernel is $\{(\lambda,0,\operatorname{Id}):\lambda\in U(1)\}$,  which is naturally identified with \(U(1)\). This subgroup is central because
\begin{equation}
(\lambda,0,\operatorname{Id})(\alpha,q,S)=(\lambda\alpha,q,S)=(\alpha,q,S)(\lambda,0,\operatorname{Id}).
\end{equation}
This defines the claimed central extension. Note that we have
\begin{align}
\widetilde{\Phi}_B(a)\widetilde{\Phi}_B(b) &= ( \alpha_B(a)\alpha_B(b)\kappa(S_B(a),S_B(b))\omega^{\frac12[q_B(a),S_B(a)q_B(b)]}, \\
& q_B(a)+S_B(a)q_B(b), S_B(a)S_B(b) ).
\end{align}
The right-hand side corresponds to
\begin{align}
(\omega^{\frac12[a,b]},0,\operatorname{Id})\widetilde{\Phi}_B(a+b) &= ( \omega^{\frac12[a,b]}\alpha_B(a+b), q_B(a+b), S_B(a+b) ).
\end{align}
Therefore, $B(a)B(b)=\omega^{\frac12[a,b]}B(a+b)$ is equivalent to $\widetilde{\Phi}_B(a)\widetilde{\Phi}_B(b)= (\omega^{\frac12[a,b]},0,\operatorname{Id})\widetilde{\Phi}_B(a+b)$. 
\end{proof}

\subsection{Cohomological Interpretation}\label{cohom-inter} 
We now show that the triple \((q,S,\alpha)\) admits a cohomological interpretation. Fix a homomorphism \(S:V_n\to \operatorname{Sp}(V_n)\). Let \(V_{n,S}\) be the group \(V_n\) with group action of $V_n$ defined by \(a\cdot v:=S(a)v\). The \(1\)-cocycles determined by the pair \((V_n,V_{n,S})\) are
\begin{equation}
Z^1(V_n,V_{n,S}):=\{q:V_n\to V_n:q(a+b)=q(a)+S(a)q(b)\text{ for all }a,b\in V_n\}.
\end{equation}
The subgroup of \(1\)-coboundaries is
\begin{equation}
B^1(V_n,V_{n,S}):=\{q_v:v\in V_n,\ q_v(a)=S(a)v-v\}.
\end{equation}
The first cohomology group is then $H^1(V_n,V_{n,S}):=Z^1(V_n,V_{n,S})/B^1(V_n,V_{n,S})$.

\begin{prop}\label{affine-part-H1}
Let \(p\) be an odd prime and \(S:V_n\to \operatorname{Sp}(V_n)\) be a homomorphism. The map $V_n \to V_n$ in \cref{eqn:third-level-triple-q} is an element of \(Z^1(V_n,V_{n,S})\). Replacing \(q\) by \(q+q_v\), where \(q_v(a)=S(a)v-v\), corresponds to conjugating the affine homomorphism $\Phi : V_n \to \operatorname{ASp}(V_n)$ by the translation \((-v,\operatorname{Id}_{V_n})\in \operatorname{ASp}(V_n)\). Consequently, affine symplectic homomorphisms with fixed symplectic part \(S\), up to conjugation by translations, are parametrized by $H^1(V_n,V_{n,S})$. 
\end{prop}

\begin{proof}
The first statement is exactly \cref{eqn:third-level-triple-q}, which is a \(1\)-cocycle from $V_n$ to \(V_{n,S}\). For the second statement, let \(v\in V_n\). In \(\operatorname{ASp}(V_n)\), we have
\begin{align}
(-v,\operatorname{Id}_{V_n})(q(a),S(a))(v,\operatorname{Id}_{V_n})
&=(-v+q(a)+S(a)v,S(a))\\
&=(q(a)+S(a)v-v,S(a)).
\end{align}
Thus, conjugation by the translation \((-v,\operatorname{Id}_{V_n})\) replaces
\(q\) by \(q+q_v\). This proves the claim.
\end{proof}

\cref{second-char} and \cref{phase-cohomology} show that \cref{eqn:third-level-triple-alpha} can be rephrased in terms of $H^2$.

\begin{lem}\label{second-char}
Let \(p\) be an odd prime and \(R_{q,S}:V_n\times V_n\to U(1)\) such that
\begin{equation}\label{rqs}
R_{q,S}(a,b):= \kappa(S(a),S(b)) \omega^{\frac12\left([q(a),S(a)q(b)]-[a,b]\right)}.
\end{equation}
\(R_{q,S}\) is a \(U(1)\)-valued \(2\)-cocycle on \(V_n\).
\end{lem}

\begin{proof}
Define $T_{q,S}(a):=w(q(a))M_{S(a)}.$ Using the multiplication law for Weyl operators and the cocycle \(\kappa\), we obtain
\begin{align}
T_{q,S}(a)T_{q,S}(b) &=
w(q(a))M_{S(a)}w(q(b))M_{S(b)} \\ &= w(q(a))w(S(a)q(b))M_{S(a)}M_{S(b)} \\ &= \kappa(S(a),S(b)) \omega^{\frac12[q(a),S(a)q(b)]} w(q(a)+S(a)q(b))M_{S(a)S(b)} \\ & =  \kappa(S(a),S(b)) \omega^{\frac12[q(a),S(a)q(b)]} T_{q,S}(a+b).
\end{align}
Thus, \(T_{q,S}\) is a projective representation of \(V_n\) and its multiplier is
\begin{equation}\label{multiplier}
F_{q,S}(a,b)
:=
\kappa(S(a),S(b))
\omega^{\frac12[q(a),S(a)q(b)]}.
\end{equation}
Associativity implies that \(F_{q,S}\) is a \(2\)-cocycle. The multiplier $\omega^{\frac12[a,b]}$ is also a \(2\)-cocycle. Therefore, $R_{q,S}(a,b) = F_{q,S}(a,b)\omega^{-\frac12[a,b]}$ is a \(2\)-cocycle since a ratio of two \(2\)-cocycles is a \(2\)-cocycle. 
\end{proof}

\begin{prop}\label{phase-cohomology}
Let \(p\) be an odd prime, Let \(q : V_n\to V_n\) and \(S : V_n\to \operatorname{Sp}(V_n)\) satisfy \cref{eqn:third-level-triple-S} and \cref{eqn:third-level-triple-q}, respectively. There exists a phase function \(\alpha:V_n\to U(1)\) such that
$B(a):=\alpha(a)w(q(a))M_{S(a)}$
is in \(Y^3_\omega(V_n)\) if and only if $[R_{q,S}]=0$ in $H^2(V_n,U(1))$, where \(H^2(V_n,U(1))\) denotes group cohomology of \(V_n\) with coefficients in \(U(1)\)  and \(V_n\) acts trivially on \(U(1)\).
\end{prop}

\begin{proof}
By \cref{third-level-classification}, the map $B(a)=\alpha(a)w(q(a))M_{S(a)}$ belongs to \(Y^3_\omega(V_n)\) if and only if
\begin{equation}
\frac{\alpha(a+b)}{\alpha(a)\alpha(b)} = R_{q,S}(a,b)
\end{equation}
for all \(a,b\in V_n\). Recall that for a \(1\)-cochain \(\phi:V_n\to U(1)\), its \(2\)-coboundary is defined by
\begin{equation}
(\delta\phi)(a,b) = \frac{\phi(a)\phi(b)}{\phi(a+b)}.
\end{equation}
Thus, we have $R_{q,S}(a,b) = (\delta(\alpha^{-1}))(a,b)$. Hence, \(\alpha\) exists if and only if \(R_{q,S}\) is a \(2\)-coboundary. Equivalently, $[R_{q,S}]=0$ in \(H^2(V_n,U(1))\).
\end{proof}

A natural question is whether the cohomology class in \cref{phase-cohomology} depends on the choice of lifts \(S\mapsto M_S\).  \cref{no-choice-kappa} shows that the resulting cohomology class is independent of this choice.

\begin{cor}\label{no-choice-kappa}
Let \(p\) be an odd prime.
Let \(S\mapsto M_S\) and \(S\mapsto M'_S\) be two choices of Clifford lifts such that $M'_S=\eta(S)M_S$ for some function \(\eta:\operatorname{Sp}(V_n)\to U(1)\). Let \(\kappa\) and \(\kappa'\) denote the corresponding \(2\)-cocycles, and let \(R_{q,S}\) and \(R'_{q,S}\) denote the corresponding phase-obstruction cocycles. Then the associated cohomology classes agree. That is, $[R'_{q,S}]=[R_{q,S}]$ in $H^2(V_n,U(1))$. 
\end{cor}

\begin{proof}
By \cref{lem:kappa-cohomology-class}, we have $\kappa'(S,T)=\frac{\eta(S)\eta(T)}{\eta(ST)}\kappa(S,T)$. It follows that
\begin{align}
R'_{q,S}(a,b)
&=\kappa'(S(a),S(b))\omega^{\frac12\left([q(a),S(a)q(b)]-[a,b]\right)} \\
&=\frac{\eta(S(a))\eta(S(b))}{\eta(S(a)S(b))}R_{q,S}(a,b)=\frac{\eta(S(a))\eta(S(b))}{\eta(S(a+b))}R_{q,S}(a,b).
\end{align}
If we define $\eta_S(a):=\eta(S(a))$, then $R'_{q,S}=R_{q,S}\cdot\delta\eta_S$.  Thus, \(R'_{q,S}\) and \(R_{q,S}\) differ by a $2$-coboundary, and therefore determine the same cohomology class in \(H^2(V_n,U(1))\).
\end{proof}

\subsection{Generator Level Criteria}\label{generator-level-criteria-third}
\cref{third-level-linear-algebraic} gives an explicit description of \(\mathcal C^{(3)}\) in terms of the triple \((\alpha,q,S)\). We now make this description algorithmic. We first show that the affine symplectic data \((q,S)\) are determined by their values on a basis of \(V_n\), subject only to the relations imposed by the additive group structure of \(V_n\).

\begin{cor}\label{cor:third-level-generator-equations}
Let \(p\) be an odd prime.
Let \(e_1,\ldots,e_{2n}\) be an \(\mathbb F_p\)-basis of \(V_n\). A pair of maps $S:V_n\to \operatorname{Sp}(V_n)$ and $q:V_n\to V_n$ satisfies \cref{eqn:third-level-triple-S} and \cref{eqn:third-level-triple-q} if and only if the data $S_i:=S(e_i)$ and $q_i:=q(e_i)$ satisfy
\begin{align}
S_iS_j&=S_jS_i, \label{eqn:third-level-generator-S-commute}\\ S_i^p&=\operatorname{Id}_{V_n}, \label{eqn:third-level-generator-S-order}\\ q_i+S_iq_j&=q_j+S_jq_i, \label{eqn:third-level-generator-q-commute}\\
q_i+S_iq_i+\cdots+S_i^{p-1}q_i&=0 \label{eqn:third-level-generator-q-order}
\end{align}
for all \(1\leq i,j\leq 2n\).
\end{cor}

\begin{proof}
By \cref{cor:third-level-affine-symplectic-homomorphism}, specifying the pair \((q,S)\) is equivalent to specifying a group homomorphism \(\Phi:V_n\to \operatorname{ASp}(V_n)\) such that \(\Phi(a)=(q(a),S(a))\). The additive group \(V_n\) is generated by \(e_1,\ldots,e_{2n}\), subject to the relations $e_i+e_j=e_j+e_i$ and $pe_i=0$ for \(1\leq i,j\leq 2n\). Therefore, \(\Phi\) is a homomorphism if and only if the elements $\Phi(e_i)=(q_i,S_i)$ commute pairwise and have order dividing \(p\). Using the multiplication law in \(\operatorname{ASp}(V_n)\), the commutativity condition $(q_i,S_i)(q_j,S_j)=(q_j,S_j)(q_i,S_i)$ becomes equivalent to
\begin{align}
S_iS_j=S_jS_i, \qquad
q_i+S_iq_j=q_j+S_jq_i.
\end{align}
Similarly, since \((q_i,S_i)^p= (q_i+S_iq_i+\cdots+S_i^{p-1}q_i, S_i^p )\), the condition \((q_i,S_i)^p=(0,\operatorname{Id}_{V_n})\) holds if and only if $S_i^p=\operatorname{Id}_{V_n}$ and 
\begin{align}
q_i+S_iq_i+\cdots+S_i^{p-1}q_i=0.
\end{align}
This completes the proof.
\end{proof}

Once the tuple \((q,S)\) has been fixed, the remaining obstruction class \([R_{q,S}]\) can be tested explicitly on generators through the commutator bicharacter associated with \(R_{q,S}\). The proof
of  \cref{cor-cohom} is relegated to \cref{appendix:cohom-facts}. All relevant cohomological facts and results used here are discussed there.  

\begin{prop}\label{cor-cohom}
Let \(p\) be an odd prime.
Let \(q:V_n\to V_n\) and \(S:V_n\to \operatorname{Sp}(V_n)\) 
satisfy \cref{eqn:third-level-triple-S} and \cref{eqn:third-level-triple-q}, respectively. The obstruction class \([R_{q,S}]\in H^2(V_n,U(1))\) vanishes if and only if
\begin{equation}
\frac{R_{q,S}(a,b)}{R_{q,S}(b,a)}=1 
\iff 
\frac{\kappa(S(a),S(b))}{\kappa(S(b),S(a))}\,\omega^{\frac12\left([q(a),S(a)q(b)]-[q(b),S(b)q(a)]-2[a,b]\right)}=1
\end{equation}
In particular, if \(e_1,\ldots,e_{2n}\) is an \(\mathbb{F}_p\)-basis of \(V_n\), then it suffices to check that
\begin{equation}
\frac{\kappa(S(e_i),S(e_j))}{\kappa(S(e_j),S(e_i))}\,\omega^{\frac12\left([q(e_i),S(e_i)q(e_j)]-[q(e_j),S(e_j)q(e_i)]-2[e_i,e_j]\right)}=1
\end{equation}
for all \(1\leq i,j\leq 2n\).
\end{prop}

The discussion above gives a procedure for constructing \(\mathcal{C}^{(3)}/U(1)\). Fix a basis \(e_1,\ldots,e_{2n}\) of \(V_n\), and write \(S_i:=S(e_i)\) and \(q_i:=q(e_i)\). The relations in \cref{cor:third-level-generator-equations} determine when these values extend to an affine symplectic component \((q,S)\), and the obstruction criterion in \cref{cor-cohom} determines whether \((q,S)\) admits a compatible phase function \(\alpha\).

\subsection{Cohomological Moduli Description}\label{cohom-lifting-third}
\cref{cohom-inter} identifies the affine data with a first cohomology group and the phase condition with the vanishing of an obstruction class in a second cohomology group.
We now refine this description to obtain a moduli-theoretic description of \(\mathcal{C}^{(3)}\).
First, we give an extension-theoretic reformulation of the phase obstruction in \cref{phase-cohomology}.  Let
\begin{equation}\label{eqn:third-level-central-extension-cocycle}
\sigma((q,S),(r,T)):=\kappa(S,T)\omega^{\frac12[q,Sr]}
\end{equation}
for $(q,S),(r,T)\in\operatorname{ASp}(V_n)$. The multiplication law in \cref{eqn:central-extension-asp-multiplication} shows that $\sigma$ is a normalized $U(1)$-valued $2$-cocycle representing the central extension
\begin{equation}\label{eqn:third-level-central-extension-repeat}
1\longrightarrow U(1)\longrightarrow\widetilde{\operatorname{ASp}}(V_n)\longrightarrow\operatorname{ASp}(V_n)\longrightarrow 1.
\end{equation}
The homomorphism $\Phi_{q,S}:V_n\to\operatorname{ASp}(V_n)$ pulls back \cref{eqn:third-level-central-extension-repeat} to a central extension
\begin{equation}\label{eqn:third-level-pulled-back-central-extension}
1\longrightarrow U(1)\longrightarrow\Phi_{q,S}^*\widetilde{\operatorname{ASp}}(V_n)\longrightarrow V_n\longrightarrow 1,
\end{equation}
whose extension class is $\Phi_{q,S}^*[\sigma]\in H^2(V_n,U(1))$. On the other hand, the Weyl operators satisfy
\begin{equation}
w(a)w(b)=m_W(a,b)w(a+b),
\end{equation}
where $m_W(a,b)=\omega^{\frac12[a,b]}$. Since $m_W$ is a normalized $U(1)$-valued $2$-cocycle, it defines a central extension $\widetilde V_{n,W}:=U(1)\times V_n$ with multiplication
\begin{equation}
(\lambda,a)(\mu,b):=(\lambda\mu\,m_W(a,b),a+b).
\end{equation}
Thus, $\widetilde V_{n,W}$ fits into the short exact sequence
\begin{equation}\label{eqn:third-level-weyl-central-extension}
1\longrightarrow U(1)\longrightarrow\widetilde V_{n,W}\longrightarrow V_n\longrightarrow 1,
\end{equation}
whose extension class is $h_W=[m_W]\in H^2(V_n,U(1))$. Hence, \cref{eqn:third-level-weyl-central-extension} is the central extension associated with the projective Weyl representation $a\mapsto w(a)$. \cref{cor:third-level-phase-lifting} compares the two central extensions in \cref{eqn:third-level-pulled-back-central-extension,eqn:third-level-weyl-central-extension} and shows that a compatible phase function exists precisely when the two central extensions determine the same class in \(H^2(V_n,U(1))\), by identifying the difference between these extension classes with the obstruction class \([R_{q,S}]\).


\begin{lem}\label{cor:third-level-phase-lifting}
Let $p$ be an odd prime. Let $S:V_n\to\operatorname{Sp}(V_n)$ be a homomorphism, let $q\in Z^1(V_n,V_{n,S})$, and let $\Phi_{q,S}$ be as in \cref{cor:third-level-affine-symplectic-homomorphism}. Then $[R_{q,S}]=\Phi_{q,S}^*[\sigma]\,h_W^{-1}$ in $H^2(V_n,U(1))$. Consequently, the following statements are equivalent:
\begin{enumerate}
\item there exists a phase function $\alpha:V_n\to U(1)$ such that $B(a):=\alpha(a)w(q(a))M_{S(a)}$ is in $Y^3_\omega(V_n)$;
\item $\Phi_{q,S}^*[\sigma]=h_W$ in $H^2(V_n,U(1))$;
\item $[R_{q,S}]=0$ in $H^2(V_n,U(1))$.
\end{enumerate}
\end{lem}

\begin{proof}
Pulling back the cocycle $\sigma$ in \cref{eqn:third-level-central-extension-cocycle} along $\Phi_{q,S}$ gives
\begin{equation}\label{eqn:third-level-pulled-back-extension-cocycle}
(\Phi_{q,S}^*\sigma)(a,b)=\kappa(S(a),S(b))\omega^{\frac12[q(a),S(a)q(b)]}.
\end{equation}
Dividing \cref{eqn:third-level-pulled-back-extension-cocycle} by the Weyl multiplier $m_W(a,b)=\omega^{\frac12[a,b]}$ gives
\begin{equation}
(\Phi_{q,S}^*\sigma)(a,b)m_W(a,b)^{-1}=\kappa(S(a),S(b))\omega^{\frac12\left([q(a),S(a)q(b)]-[a,b]\right)}=R_{q,S}(a,b).
\end{equation}
Thus, $R_{q,S}=(\Phi_{q,S}^*\sigma)m_W^{-1}.$ Passing to cohomology and using $h_W=[m_W]$ gives
\begin{equation}
[R_{q,S}]=\Phi_{q,S}^*[\sigma]\,h_W^{-1}.
\end{equation}
Therefore, $\Phi_{q,S}^*[\sigma]=h_W$ if and only if $[R_{q,S}]=0$. By \cref{phase-cohomology}, the latter condition is equivalent to the existence of a phase function $\alpha$ for which $B(a)=\alpha(a)w(q(a))M_{S(a)}$ belongs to $Y^3_\omega(V_n)$.
\end{proof}

We next show in \cref{prop:third-level-obstruction-H1-invariant} that the obstruction in second cohomology depends only on the cohomology class \([q]\in H^1(V_n,V_{n,S})\) of the affine data. 

\begin{lem}\label{prop:third-level-obstruction-H1-invariant}
Let $p$ be an odd prime. Let $S:V_n\to\operatorname{Sp}(V_n)$ be a homomorphism. Suppose that $q,q'\in Z^1(V_n,V_{n,S})$ satisfy
$q'=q+q_v,$
where $q_v(a)=S(a)v-v$ for some $v\in V_n$. Then
$[R_{q',S}]=[R_{q,S}]$
in $H^2(V_n,U(1))$.
\end{lem}

\begin{proof}
Define $T_{q,S}$ and $T_{q',S}$ as in \cref{cor:third-level-phase-lifting}. Let $\widetilde{T}(a):=w(-v)T_{q,S}(a)w(v)$. We have
\begin{align}
\widetilde{T}(a)&=w(-v)T_{q,S}(a)w(v)\\
&=w(-v)w(q(a))M_{S(a)}w(v)=w(-v)w(q(a))w(S(a)v)M_{S(a)}=\lambda_v(a)T_{q',S}(a),
\end{align}
where
\begin{equation}\label{eqn:third-level-lambda-v}
\lambda_v(a):=\omega^{\frac12\left([-v,q(a)]+[q(a)-v,S(a)v]\right)}.
\end{equation}
Let \(F_{q,S}\) be the multiplier of \(T_{q,S}\), as in \cref{multiplier}. Since conjugation by the fixed unitary \(w(-v)\) does not change the multiplier $F_{q,s}$, we have
\begin{equation}
\widetilde{T}(a)\widetilde{T}(b)=F_{q,S}(a,b)\widetilde{T}(a+b).
\end{equation}
Hence, we have
\begin{align}
\lambda_v(a)\lambda_v(b)F_{q',S}(a,b)T_{q',S}(a+b)&=\widetilde{T}(a)\widetilde{T}(b)=F_{q,S}(a,b)\lambda_v(a+b)T_{q',S}(a+b).
\end{align}
This implies we have
\begin{equation}
F_{q',S}(a,b)=F_{q,S}(a,b)\frac{\lambda_v(a+b)}{\lambda_v(a)\lambda_v(b)}=F_{q,S}(a,b)\bigl(\delta(\lambda_v^{-1})\bigr)(a,b).
\end{equation}
Hence,
$F_{q',S}=F_{q,S}\cdot\delta(\lambda_v^{-1}).$
The multiplier \(m_W(a,b)=\omega^{\frac12[a,b]}\) depends only on \(a\) and \(b\), and is therefore unchanged when \(q\) is replaced by \(q'\). Since \(R_{q,S}=F_{q,S}m_W^{-1}\), it follows that
$R_{q',S}=R_{q,S}\cdot\delta(\lambda_v^{-1}).$
Thus, \(R_{q',S}\) and \(R_{q,S}\) differ by a \(2\)-coboundary and therefore determine the same cohomology class in \(H^2(V_n,U(1))\).
\end{proof}

Let \(S:V_n\to\operatorname{Sp}(V_n)\) be a fixed homomorphism. Combining \cref{prop:third-level-obstruction-H1-invariant} with \cref{no-choice-kappa} defines a map
\begin{align}
\operatorname{ob}_S:H^1(V_n,V_{n,S})&\longrightarrow H^2(V_n,U(1)),\label{eqn:third-level-obstruction-map}\\
[q]&\longmapsto[R_{q,S}],
\end{align}
which determines the phase obstruction governing the existence of compatible phase data. For fixed \(S:V_n\to\operatorname{Sp}(V_n)\), the dependence of \(R_{q,S}\) on \(q\) is contained entirely in the term \([q(a),S(a)q(b)]\). In \cref{lem:third-level-symplectic-cup-product}, we identify this term as a cohomological term obtained by pairing the group cohomological cup product with the symplectic form.

\begin{lem}\label{lem:third-level-symplectic-cup-product}
Let $p$ be an odd prime.
Let $S:V_n\to\operatorname{Sp}(V_n)$ be a homomorphism and let $q,r\in Z^1(V_n,V_{n,S})$. Define
\begin{equation}\label{eqn:third-level-symplectic-cup-cochain}
(q\smile_S r)(a,b):=[q(a),S(a)r(b)].
\end{equation}
Then $q\smile_S r$ is an $\mathbb F_p$-valued $2$-cocycle on $V_n$. Its cohomology class depends only on $[q]$ and $[r]$, and therefore determines a bilinear pairing
\begin{equation}\label{eqn:third-level-symplectic-cup-pairing}
\smile_S:H^1(V_n,V_{n,S})\times H^1(V_n,V_{n,S})\longrightarrow H^2(V_n,\mathbb F_p).
\end{equation}
\end{lem}

\begin{proof}
The symplectic form is an \(\mathbb F_p\)-bilinear map \(V_{n,S}\times V_{n,S}\to\mathbb F_p\). Hence, it induces an \(\mathbb F_p\)-linear map
\begin{align}
\beta:V_{n,S}\otimes_{\mathbb F_p}V_{n,S}&\longrightarrow\mathbb F_p,\\
x\otimes y&\longmapsto[x,y].
\end{align}
We equip \(V_{n,S}\otimes_{\mathbb F_p}V_{n,S}\) with the diagonal \(V_n\)-action. That is, \(a\cdot(x\otimes y)=S(a)x\otimes S(a)y\) for $a \in V_n$. Similarly, we equip \(\mathbb F_p\) with the trivial \(V_n\)-action. Since \(S(a)\in\operatorname{Sp}(V_n)\) for every \(a\in V_n\), we have
\begin{equation}
\beta\bigl(a\cdot(x\otimes y)\bigr)=\beta\bigl(S(a)x\otimes S(a)y\bigr)=[S(a)x,S(a)y]=[x,y]=a\cdot\beta(x\otimes y).
\end{equation}
Thus, \(\beta\) is \(V_n\)-equivariant and therefore induces a homomorphism
\begin{align}
\beta_*:H^2\left(V_n,V_{n,S}\otimes_{\mathbb F_p}V_{n,S}\right)&\longrightarrow H^2(V_n,\mathbb F_p),\\
[\eta]&\longmapsto[\beta\circ\eta].
\end{align}
For \(q,r\in Z^1(V_n,V_{n,S})\), their ordinary cup product is the \(2\)-cochain with values in \(V_{n,S}\otimes_{\mathbb F_p}V_{n,S}\) given by
\begin{equation}
(q\smile r)(a,b)=q(a)\otimes S(a)r(b).
\end{equation}
The cup product is compatible with the group-cohomology differential:
\begin{equation}
\delta(f\smile g)=(\delta f)\smile g+(-1)^{\deg f}f\smile(\delta g).
\end{equation}
Since \(q\) and \(r\) are \(1\)-cocycles, \(\delta q=\delta r=0\), and hence \(q\smile r\) is a \(2\)-cocycle. The same identity shows that replacing either \(q\) or \(r\) by a cohomologous \(1\)-cocycle changes \(q\smile r\) by a \(2\)-coboundary. Therefore, the cohomology class \([q\smile r]\) depends only on \([q]\) and \([r]\), and the ordinary cup product induces a bilinear map
\begin{align}
\smile:\;H^1(V_n,V_{n,S})\times H^1(V_n,V_{n,S})&\longrightarrow H^2\left(V_n,V_{n,S}\otimes_{\mathbb F_p}V_{n,S}\right),\\
([q],[r])&\longmapsto[q\smile r].
\end{align}
Composing the ordinary cup product with \(\beta_*\) gives a bilinear pairing
\begin{align}
H^1(V_n,V_{n,S})\times H^1(V_n,V_{n,S})&\xrightarrow{\ \smile\ }H^2\left(V_n,V_{n,S}\otimes_{\mathbb F_p}V_{n,S}\right)\xrightarrow{\ \beta_*\ }H^2(V_n,\mathbb F_p),\\
([q],[r])&\longmapsto[q\smile r]\longmapsto[\beta\circ(q\smile r)].
\end{align}
At the cochain level,
\begin{equation}
\bigl(\beta\circ(q\smile r)\bigr)(a,b)=\beta\bigl(q(a)\otimes S(a)r(b)\bigr)=[q(a),S(a)r(b)]=(q\smile_S r)(a,b).
\end{equation}
Therefore,
$[q]\smile_S[r]=\beta_*([q\smile r])\in H^2(V_n,\mathbb F_p),$ which is the pairing in \cref{eqn:third-level-symplectic-cup-pairing}.
\end{proof}

\cref{lem:third-level-symplectic-cup-product} identifies the term \([q(a),S(a)q(b)]\) appearing in the phase obstruction as the symplectically paired cup product of the affine cocycle \(q\). We now use this identification to express the entire obstruction map intrinsically in terms of the cohomology class \([q]\). In particular, we show in \cref{cor:third-level-cup-square-obstruction} that, for fixed symplectic data \(S\), all dependence of the phase obstruction on the affine class \([q]\) is captured by the symplectically paired cup square \([q]\smile_S[q]\).

\begin{cor}\label{cor:third-level-cup-square-obstruction}
Let $p$ be an odd prime and let $S:V_n\to\operatorname{Sp}(V_n)$ be a homomorphism. For every $[q]\in H^1(V_n,V_{n,S})$, the obstruction map satisfies
\begin{equation}\label{eqn:third-level-cup-square-obstruction}
\operatorname{ob}_S([q])=S^*[\kappa]\cdot\chi_*\left(\frac12([q]\smile_S[q])\right)\cdot h_W^{-1}.
\end{equation}
Here $\chi:\mathbb F_p\to U(1)$
is the standard additive character and
\begin{equation}
\chi_*:H^2(V_n,\mathbb F_p)\longrightarrow H^2(V_n,U(1))
\end{equation}
is the induced map on cohomology.  
\end{cor}

\begin{proof}
By \cref{cor:third-level-phase-lifting}, the map \(\operatorname{ob}_S\) satisfies
\begin{equation}\label{eqn:third-level-obstruction-via-extension}
\operatorname{ob}_S([q])=[R_{q,S}]=\Phi_{q,S}^*[\sigma]\,h_W^{-1}.
\end{equation}
The cocycle \(\kappa(S(a),S(b))\) represents the pullback class \(S^*[\kappa]\). By \cref{eqn:third-level-symplectic-cup-cochain}, the term \(\omega^{\frac12[q(a),S(a)q(b)]}\) represents
\(\chi_*\left(\frac12([q]\smile_S[q])\right)\). Hence,
\begin{equation}\label{pullback-identify}
\Phi_{q,S}^*[\sigma]
=
S^*[\kappa]\cdot\chi_*\left(\frac12([q]\smile_S[q])\right).
\end{equation}
Substituting \cref{pullback-identify} into \cref{eqn:third-level-obstruction-via-extension} gives \cref{eqn:third-level-cup-square-obstruction}.
\end{proof}

Using \cref{cor:third-level-cup-square-obstruction}, we can now record the quadratic structure implied by the cup-square formula. In what follows, we use multiplicative notation in \(H^2(V_n,U(1))\) and additive notation in \(H^1(V_n,V_{n,S})\) and \(H^2(V_n,\mathbb F_p)\). Define

\begin{equation}\label{eqn:third-level-QS}
Q_S(x):=\chi_*\left(\frac12(x\smile_S x)\right)
\end{equation}
for \(x\in H^1(V_n,V_{n,S})\). By bilinearity of the paired cup product,
\begin{equation}
(x+y)\smile_S(x+y)=x\smile_Sx+x\smile_Sy+y\smile_Sx+y\smile_Sy.
\end{equation}
Multiplying by \(\frac12\) in \(H^2(V_n,\mathbb F_p)\), applying the homomorphism \(\chi_*\), and using multiplicative notation in \(H^2(V_n,U(1))\) gives
\begin{equation}\label{eqn:third-level-QS-polarization}
Q_S(x+y)Q_S(x)^{-1}Q_S(y)^{-1}=\chi_*\left(\frac12(x\smile_S y+y\smile_S x)\right).
\end{equation}
The right-hand side is bilinear in \(x\) and \(y\), since \(\smile_S\) is bilinear and \(\chi_*\) is a homomorphism, and it is symmetric under interchange of \(x\) and \(y\). Thus, \(Q_S\) is a quadratic map. By \cref{cor:third-level-cup-square-obstruction}, we have
\begin{equation}
\operatorname{ob}_S(x)=S^*[\kappa]\,Q_S(x)\,h_W^{-1}.
\end{equation}
Hence, \(\operatorname{ob}_S\) is an affine quadratic map with constant term $\operatorname{ob}_S(0)=S^*[\kappa]\,h_W^{-1}$.  The cohomology classes of affine data compatible with third-level phase data therefore form the quadratic locus
\begin{equation}\label{eqn:third-level-admissible-locus}
\mathfrak M_S:=\operatorname{ob}_S^{-1}(0)\subseteq H^1(V_n,V_{n,S}).
\end{equation}
The cohomological moduli result that we wish to prove is formulated in terms of \(\mathfrak M_S\). We first show that the compatible phase functions form a torsor over \(H^1(V_n,U(1))\) and identify this ambiguity with right multiplication by Weyl operators. We also identify the coboundary ambiguity in the affine data with left multiplication by Weyl operators. This is proved in \cref{lem:third-level-phase-Weyl-actions}.

\begin{lem}\label{lem:third-level-phase-Weyl-actions}
Let \(p\) be an odd prime. Let \(S:V_n\to\operatorname{Sp}(V_n)\) be a homomorphism, and let \(q\in Z^1(V_n,V_{n,S})\) satisfy \(\operatorname{ob}_S([q])=0\). Then the following statements hold.
\begin{enumerate}[label=\textnormal{(\arabic*)}]
\item The set of phase functions \(\alpha:V_n\to U(1)\) for which $B(a)=\alpha(a)w(q(a))M_{S(a)}$ belongs to \(Y^3_\omega(V_n)\) is a torsor over
\begin{equation}\label{eqn:third-level-phase-H1}
H^1(V_n,U(1))=\operatorname{Hom}(V_n,U(1)) \cong V_n. 
\end{equation}

\item If \([U]\in\overline{\mathcal C}^{(3)}\) and  $B_U(a):=Uw(a)U^*=\alpha(a)w(q(a))M_{S(a)},$ then left multiplication by \(w(-v)\), for \(v\in V_n\), leaves \(S\) unchanged and replaces \(q\) by the cohomologous cocycle $q+q_v$, where $q_v(a):=S(a)v-v$. 

\item Under the same assumptions as in (2), right multiplication by \(w(v)\) leaves \(S\) and \(q\) unchanged and replaces \(\alpha\) by \(\chi_v\alpha\), where $\chi_v(a)=\omega^{[v,a]}$. 
\end{enumerate}
\end{lem}

\begin{proof}
Let \(\alpha\) and \(\alpha'\) be two compatible phase functions for the same \(q\) and \(S\). By \cref{third-level-classification}, both satisfy
\begin{equation}\label{ratio}
\frac{\alpha(a+b)}{\alpha(a)\alpha(b)}=R_{q,S}(a,b)=\frac{\alpha'(a+b)}{\alpha'(a)\alpha'(b)}.
\end{equation}
Define \(\chi(a):=\alpha'(a)\alpha(a)^{-1}\). \cref{ratio} implies 
$\chi(a+b)=\chi(a)\chi(b),$
so \(\chi\in\operatorname{Hom}(V_n,U(1))\). Conversely, let \(\chi\in\operatorname{Hom}(V_n,U(1))\) and let \(\alpha\) be a compatible phase function. Since \(\chi(a+b)=\chi(a)\chi(b)\), we have
\begin{align}
\frac{(\chi\alpha)(a+b)}{(\chi\alpha)(a)(\chi\alpha)(b)}
&=\frac{\chi(a+b)\alpha(a+b)}{\chi(a)\chi(b)\alpha(a)\alpha(b)}=\frac{\alpha(a+b)}{\alpha(a)\alpha(b)}=R_{q,S}(a,b).
\end{align}
Thus, \(\chi\alpha\) is again a compatible phase function. Hence, \(\operatorname{Hom}(V_n,U(1))\) acts on the set of compatible phase functions by pointwise multiplication. The calculation above shows that the action is transitive, since any two compatible phase functions differ by a character, and it is free, since \(\chi\alpha=\alpha\) implies \(\chi=1\). Therefore, the set of compatible phase functions is a torsor over
\begin{equation}
H^1(V_n,U(1))=\operatorname{Hom}(V_n,U(1)).
\end{equation}
This proves (1). (2) follows from the calculation in \cref{prop:third-level-obstruction-H1-invariant}, since for \(U_L:=w(-v)U\) we have $B_{U_L}(a)=w(-v)B_U(a)w(v)$. Finally, let \(U_R:=Uw(v)\). We have
\begin{align}
B_{U_R}(a)&=Uw(v)w(a)w(v)^*U^*=\omega^{[v,a]}Uw(a)U^*=\chi_v(a)B_U(a).
\end{align}
Therefore, \(S\) and \(q\) are unchanged, while \(\alpha\) is replaced by \(\chi_v\alpha\). This proves (3).
\end{proof}

In what follows, for a fixed homomorphism \(S:V_n\to\operatorname{Sp}(V_n)\), define
\begin{equation}\label{eqn:third-level-fixed-S-stratum}
\overline{\mathcal C}^{(3)}_S:=\{[U]\in\overline{\mathcal C}^{(3)}:S_U=S\}.
\end{equation}
We define an equivalence relation \(\sim_W\) on \(\overline{\mathcal C}^{(3)}_S\) by declaring $[U]\sim_W[U']$
if there exist \([W_L],[W_R]\in\overline{\mathcal C}^{(1)}\) such that
$[U']=[W_LUW_R].$
By \cref{lem:third-level-phase-Weyl-actions}, left and right Weyl multiplication leave the symplectic part \(S_U\) unchanged and therefore preserve \(\overline{\mathcal C}^{(3)}_S\). Thus, the quotient by \(\sim_W\) is parametrized by the admissible cohomology classes \([q]\in\mathfrak M_S\), yielding \cref{prop:third-level-cohomological-moduli}.

\begin{prop}\label{prop:third-level-cohomological-moduli}
Let \(p\) be an odd prime. For every homomorphism \(S:V_n\to\operatorname{Sp}(V_n)\), there is a bijection
\begin{equation}\label{eqn:third-level-fixed-S-moduli}
\overline{\mathcal C}^{(3)}_S/ \sim_W \cong \mathfrak M_S=\operatorname{ob}_S^{-1}(0).
\end{equation}
Consequently,
\begin{equation}\label{eqn:third-level-global-moduli}
\overline{\mathcal C}^{(3)}/\!\sim_W \cong \bigsqcup_{S\in\operatorname{Hom}(V_n,\operatorname{Sp}(V_n))}\operatorname{ob}_S^{-1}(0).
\end{equation}
Thus, modulo left and right Weyl multiplication, the third Clifford level is classified by the zero loci of the phase-obstruction maps on the groups \(H^1(V_n,V_{n,S})\).
\end{prop}

\begin{proof}
Let \([U]\in\overline{\mathcal C}^{(3)}_S\) and \(B_U(a)=\alpha(a)w(q(a))M_{S(a)}\). By \cref{affine-part-H1}, the affine data \(q\) determine a class \([q]\in H^1(V_n,V_{n,S})\). By \cref{phase-cohomology}, the existence of the phase function \(\alpha\) implies \(\operatorname{ob}_S([q])=0\). Hence, there is a map \begin{align}
\overline{\mathcal C}^{(3)}_S &\longrightarrow \mathfrak M_S,\\
[U] &\longmapsto [q].
\end{align} By \cref{lem:third-level-phase-Weyl-actions}, left Weyl multiplication changes \(q\) by a \(1\)-coboundary, while right Weyl multiplication leaves \(q\) unchanged. Therefore, the map is constant on \(\sim_W\)-equivalence classes and descends to a map
\begin{align}
\Theta_S:\overline{\mathcal C}^{(3)}_S/\!\sim_W&\longrightarrow\mathfrak M_S, \\
[U]_{\sim_W}&\longmapsto[q].
\end{align}
We show that \(\Theta_S\) is both injective and surjective.
\begin{enumerate} 
\item (Surjectivity)
Let \([q]\in\mathfrak M_S\), and choose a cocycle representative \(q\). Since \(\operatorname{ob}_S([q])=0\), \cref{phase-cohomology} gives a phase function \(\alpha\) such that \(B(a)=\alpha(a)w(q(a))M_{S(a)}\) belongs to \(Y^3_\omega(V_n)\). By \cref{cor:equiv-abstract}, there exists \([U]\in\overline{\mathcal C}^{(3)}\) with \(B=B_U\). Its symplectic component is \(S\), so \([U]\in\overline{\mathcal C}^{(3)}_S\), and its \(\sim_W\)-equivalence class maps to \([q]\). Thus, \(\Theta_S\) is surjective.

\item (Injectivity)
Suppose that \([U],[U']\in\overline{\mathcal C}^{(3)}_S\) determine the same class in \(\mathfrak M_S\), and write their data as \((\alpha,q,S)\) and \((\alpha',q',S)\). Since \([q']=[q]\), there exists \(v\in V_n\) such that \(q'=q+q_v\). By \cref{lem:third-level-phase-Weyl-actions}, left multiplication of \(U\) by \(w(-v)\) replaces \(q\) by \(q'\). We may therefore replace \(U\) by a \(\sim_W\)-equivalent gate whose affine cocycle is exactly \(q'\). After this replacement, both gates have the same affine and symplectic data \((q',S)\). By \cref{lem:third-level-phase-Weyl-actions}, their phase functions differ by a character of \(V_n\), which is of the form \(\chi_u\) for some \(u\in V_n\). Again by \cref{lem:third-level-phase-Weyl-actions}, right multiplication by \(w(u)\) realizes multiplication of the phase function by \(\chi_u\). Thus, after a further \(\sim_W\)-equivalent replacement, the two projective Weyl fields agree. By \cref{cor:equiv-abstract}, the resulting projective gates agree. Hence, \([U]\sim_W[U']\), so \(\Theta_S\) is injective.
\end{enumerate}
Therefore, \(\Theta_S\) is a bijection, establishing \cref{eqn:third-level-fixed-S-moduli}. Since left and right Weyl multiplication preserve the symplectic component \(S\), the equivalence relation \(\sim_W\) does not mix the strata \(\overline{\mathcal C}^{(3)}_S\). Taking the disjoint union over all homomorphisms \(S:V_n\to\operatorname{Sp}(V_n)\) gives \cref{eqn:third-level-global-moduli}.
\end{proof}

\section{Semi-Clifford Gates in the Third Level}\label{semi-third}
We now explain how the description of \(\mathcal C^{(3)}\) developed in \cref{third-clifford} can be used to determine whether a gate is semi-Clifford. In \cref{general-results}, we reformulate the semi-Clifford condition in terms of Lagrangian subspaces of \(V_n\) and establish general criteria for detecting semi-Clifford gates. In \cref{two}, we give a new proof that every two-qudit gate is semi-Clifford. Finally, in \cref{three}, we prove that every three-qudit gate is semi-Clifford.

\subsection{General Results}\label{general-results}
We begin by reformulating semi-Clifford condition in terms of Lagrangian subspaces of \(V_n\) by using the correspondence between maximal abelian subgroups of \(\mathcal C^{(1)}\) and Lagrangian subspaces of \(V_n\). For a Lagrangian subspace \(L\subseteq V_n\), let
\begin{equation}
W_L:=\{\lambda w(\ell):\lambda\in U(1),\ \ell\in L\}\subseteq\mathcal C^{(1)}
\end{equation}
denote the corresponding maximal abelian subgroup. \cref{lem:semi-clifford-lagrangian-reformulation} gives an equivalent characterization of semi-Clifford gates in terms of these subgroups.

\begin{lem}\label{lem:semi-clifford-lagrangian-reformulation}
Let \(p\) be a prime and \(U\in U(\mathcal H_n)\). \(U\) is semi-Clifford if and only if there exist Lagrangian subspaces \(L,L'\subseteq V_n\) such that $U W_L U^*=W_{L'}$. 
\end{lem}

\begin{proof}
This follows directly from \cref{lem:semi-clifford-diagonal-equivalence}. 
\end{proof}

Before considering semi-Clifford 
gates, we first give a criterion for when a gate in \(\mathcal C^{(3)}\) lies in \(\mathcal C^{(2)}\). Let \(U\in\mathcal C^{(3)}\). We have
\begin{equation}
\label{eqn:BU-third-data-semi}
Uw(a)U^*=B_U(a)=\alpha_U(a)w(q_U(a))M_{S_U(a)}
\end{equation}
for every \(a\in V_n\). The component \(S_U(a)\in\operatorname{Sp}(V_n)\) records the non-Pauli Clifford part of the conjugated Weyl operator. \cref{C3-in-C2-criterion} shows that \(U\in\mathcal C^{(2)}\) precisely when this component is trivial for every \(a\in V_n\).

\begin{lem}\label{C3-in-C2-criterion}
Let \(p\) be an odd prime.
Let \(U\in\mathcal C^{(3)}\). Then \(U\in\mathcal C^{(2)}\) if and only if \(S_U(a)=\operatorname{Id}_{V_n}\) for every \(a\in V_n\).
\end{lem}

\begin{proof}
If \(U\in\mathcal C^{(2)}\), then \(Uw(a)U^*\in\mathcal C^{(1)}\) for every \(a\in V_n\). By \cref{eqn:BU-third-data-semi}, we have
\begin{equation}
Uw(a)U^*=\alpha_U(a)w(q_U(a))M_{S_U(a)} \in \mathcal{C}^{(1)}
\end{equation}
for all $a \in V_n$. Let $\Psi$ be the map from \cref{second-level-symplectic-quotient}. Since $\ker \Psi = \mathcal C^{(1)}$, we have
\begin{align}
\operatorname{Id}_{V_n}
&= \Psi(\alpha_U(a)w(q_U(a))M_{S_U(a)}) \\
&= \Psi(\alpha_U(a)w(q_U(a)))\Psi(M_{S_U(a)})= \Psi(M_{S_U(a)})=S_U(a)
\end{align}
for all $a \in V_n$. Conversely, if \(S_U(a)=\operatorname{Id}_{V_n}\) for every \(a \in V_n\), then $Uw(a)U^*=\alpha_U(a)w(q_U(a))\in\mathcal C^{(1)}$ for every \(a\in V_n\). Hence, \(U\in\mathcal C^{(2)}\).
\end{proof}

\cref{C3-in-C2-criterion} characterizes the Clifford gates inside \(\mathcal C^{(3)}\) by requiring the symplectic component \(S_U(a)\) to be trivial for every \(a\in V_n\). For semi-Clifford gates, one only requires this condition on a Lagrangian subspace. Indeed, \(U\) is semi-Clifford precisely when there exists a Lagrangian subspace \(L\subseteq V_n\) such that \(Uw(\ell)U^*\) is a Weyl operator for every \(\ell\in L\). Equivalently, \(S_U(\ell)=\operatorname{Id}_{V_n}\) for every \(\ell\in L\). \cref{semi-clifford-kernel-criterion} shows that this condition is both necessary and sufficient. Before stating and proving \cref{semi-clifford-kernel-criterion}, we record in \cref{lem:kernel-cocycle-consequences} some basic facts about the cocycle data \(q_U\) that will be used repeatedly below.

\begin{lem}\label{lem:kernel-cocycle-consequences}
Let \(p\) be an odd prime. Let \(U\in\mathcal C^{(3)}\), and set \(K_U:=\ker S_U\) and \(E_U:=\operatorname{im}S_U\). Then \(q_U|_{K_U}\) is linear and injective. Moreover,
\(q_U|_{K_U}\) preserves the symplectic form. That is, 
\begin{equation}\label{eqn:kernel-cocycle-isometry}
[q_U(k),q_U(\ell)]=[k,\ell]
\end{equation}
for all \(k,\ell\in K_U\). Moreover, \(E_U\) fixes \(q_U(K_U)\) pointwise, and
\begin{equation}\label{eqn:kernel-cocycle-mixed-pairing}
[q_U(k),q_U(a)]=[k,a]
\end{equation}
for all \(k\in K_U\) and \(a\in V_n\).
\end{lem}

\begin{proof}
If \(k,\ell\in K_U\), then \(S_U(k)=S_U(\ell)=\operatorname{Id}_{V_n}\), so \cref{eqn:third-level-triple-q} gives $q_U(k+\ell)=q_U(k)+q_U(\ell).$ Thus,  \(q_U|_{K_U}\) is linear. If \(q_U(k)=0\), then
\begin{equation}
Uw(k)U^*=B_U(k)=\alpha_U(k)w(q_U(k))=\alpha_U(k)I.
\end{equation}
Hence, \(w(k)\in U(1)I\), so \(k=0\) by \cref{lem:weyl-scalar-label}. Thus, \(q_U|_{K_U}\) is injective.
For \(k,\ell\in K_U\), conjugating the Weyl commutator by \(U\) gives
\begin{equation}\label{co-1}
B_U(k)B_U(\ell)B_U(k)^*B_U(\ell)^*=Uw(k)w(\ell)w(k)^*w(\ell)^*U^*=\omega^{[k,\ell]}I.
\end{equation}
Since \(k,\ell\in K_U\), we have \(B_U(k)=\alpha_U(k)w(q_U(k))\) and \(B_U(\ell)=\alpha_U(\ell)w(q_U(\ell))\). Therefore, we have
\begin{equation}\label{co-2}
B_U(k)B_U(\ell)B_U(k)^*B_U(\ell)^*=\omega^{[q_U(k),q_U(\ell)]}I.
\end{equation}
Comparing \cref{co-1} and \cref{co-2} gives \cref{eqn:kernel-cocycle-isometry}. Now let \(k\in K_U\) and \(a\in V_n\). Since \(V_n\) is abelian, \cref{eqn:third-level-triple-q} gives
\begin{equation}
q_U(k)+q_U(a)=q_U(k+a)=q_U(a+k)=q_U(a)+S_U(a)q_U(k).
\end{equation}
Hence \(S_U(a)q_U(k)=q_U(k)\). Thus, \(E_U\) fixes \(q_U(K_U)\) pointwise. We have
\begin{equation}\label{co-3}
B_U(k)B_U(a)B_U(k)^*B_U(a)^*=\omega^{[k,a]}I.
\end{equation}
Since \(E_U\) fixes \(q_U(K_U)\) pointwise, we have
\begin{equation}
M_{S_U(a)}w(q_U(k))M_{S_U(a)}^*=w(S_U(a)q_U(k))=w(q_U(k)),
\end{equation}
so \(M_{S_U(a)}\) commutes with \(w(q_U(k))\). Since \(k\in K_U\), we have $B_U(k)=\alpha_U(k)w(q_U(k))$. Since  $B_U(a)=\alpha_U(a)w(q_U(a))M_{S_U(a)}$ and the scalar phases cancel in the commutator, we have
\begin{align}
B_U(k)B_U(a)B_U(k)^*B_U(a)^*
&=w(q_U(k))w(q_U(a))M_{S_U(a)}w(q_U(k))^*M_{S_U(a)}^*w(q_U(a))^* \\
&=w(q_U(k))w(q_U(a))w(q_U(k))^*w(q_U(a))^*\\
&=\omega^{[q_U(k),q_U(a)]}I. \label{co-4}
\end{align}
Comparing \cref{co-3} and \cref{co-4} gives \cref{eqn:kernel-cocycle-mixed-pairing}.
\end{proof}

We can now state and prove in \cref{semi-clifford-kernel-criterion} a criterion for determining whether a gate in \(\mathcal C^{(3)}\) is semi-Clifford in terms of Lagrangian subspaces of \(V_n\).

\begin{prop}\label{semi-clifford-kernel-criterion}
Let \(p\) be an odd prime.
Let \(U\in\mathcal C^{(3)}\) and $K_U := \ker S_U =\{a\in V_n:S_U(a)=\operatorname{Id}_{V_n}\}$. 
Then \(U\) is semi-Clifford if and only if \(K_U\) contains a Lagrangian subspace of \(V_n\).
\end{prop}

\begin{proof}
If \(U\) is semi-Clifford, then there exist Lagrangian subspaces \(L,L'\subseteq V_n\) such that \(UW_LU^*=W_{L'}\). Hence, \(Uw(\ell)U^*\) is a Weyl operator for every \(\ell\in L\). An argument analogous to that used in \cref{C3-in-C2-criterion} then shows that \(S_U(\ell)=\operatorname{Id}_{V_n}\) for every \(\ell\in L\). Thus, \(L\subseteq K_U\). Conversely, suppose that \(L\subseteq K_U\) is a Lagrangian subspace. We have 
\begin{equation}
Uw(\ell)U^*=B_U(\ell)=\alpha_U(\ell)w(q_U(\ell))
\end{equation}
for every \(\ell\in L\). Since \(L\) is isotropic, the operators \(w(\ell)\) commute for all \(\ell\in L\). This implies that $B_U(\ell)$ and $B_U(\ell')$ commute, which implies that $w(q_U(\ell))w(q_U(\ell'))= w(q_U(\ell'))w(q_U(\ell))$ for all $\ell,\ell' \in L$. This implies that $[q_U(\ell),q_U(\ell')]=0$ for all $\ell,\ell' \in L$. Let
\begin{equation}
L' := \{ q_U(\ell) \mid \ell \in L\} = q_U(L)
\end{equation}
Note that $q_U|_L$ is a linear map. This follows from \cref{lem:kernel-cocycle-consequences}. The argument presented above shows that $L' $ is an isotropic subspace of $V_n$. 
Moreover, the map \(q_U|_L\) is injective. This also follows from \cref{lem:kernel-cocycle-consequences}. 
Thus, \(q_U|_L : L \to L'\) is a linear isomorphism. We have $\dim  L = n$ since \(L\) is a Lagrangian subspace. Thus, \(L' := q_U(L)\) is an \(n\)-dimensional isotropic subspace, so it is also a Lagrangian subspace. Clearly,  $UW_LU^*=W_{L'}$ so \(U\) is semi-Clifford.
\end{proof}

\cref{semi-clifford-kernel-criterion} reduces the semi-Clifford condition for \(U\in\mathcal{C}^{(3)}\) to determining whether \(K_U:=\ker S_U\) contains a Lagrangian subspace of \(V_n\). 
\cref{lem:witt-index-subspace} gives a dimension-based criterion for determining whether a gate \(U\in\mathcal C^{(3)}\) is semi-Clifford. The result relies on a standard fact from symplectic linear algebra: for any subspace \(K\subseteq V_n\), the maximum dimension of an isotropic subspace contained in \(K\) is determined by \(\dim K\) and \(\dim(K\cap K^\perp)\).

\begin{lem}\label{lem:witt-index-subspace}
Let \(p\) be a prime and let \(K\subseteq V_n\) be a subspace. Define
\begin{equation}
\operatorname{rad}(K):=\{k\in K:[k,k']=0\text{ for all }k'\in K\}=K\cap K^\perp.
\end{equation}
Then there exists an isotropic subspace $L$ of $V_n$ such that \(L\subseteq K\) and
\begin{equation}
\dim L= \frac{\dim K+\dim \operatorname{rad}(K)}{2} = \frac{\dim K+\dim (K\cap K^\perp)}{2}. 
\end{equation}
Moreover, no isotropic subspace of $V_n$  contained in \(K\) has larger dimension. Hence, the maximum dimension of a subspace of \(K\) that is isotropic in \(V_n\) is
\begin{equation}\label{eqn:witt-index-formula}
\frac{\dim K+\dim \operatorname{rad}(K)}{2}
=
\frac{\dim K+\dim (K\cap K^\perp)}{2}.
\end{equation}
\end{lem}

\begin{proof}
The restriction of the symplectic form to \(K\) may be degenerate, and its degeneracy is measured precisely by \(\operatorname{rad}(K)\). Since \(\operatorname{rad}(K)\) is the kernel of the restricted alternating form on \(K\), the form descends to a non-degenerate alternating form on the quotient \(K/\operatorname{rad}(K)\). Hence, \(K/\operatorname{rad}(K)\) is a symplectic vector space, and therefore it has Lagrangian subspaces of dimension
\begin{equation}
\frac{1}{2}\dim (K/\operatorname{rad}(K))=
\frac{1}{2}(\dim K-\dim \operatorname{rad}(K)).
\end{equation}
Let \(\overline L\subseteq K/\operatorname{rad}(K)\) be such a Lagrangian subspace, and let \(L\subseteq K\) be its inverse image under the quotient map \(K\to K/\operatorname{rad}(K)\). Since \(\overline L\) is isotropic, the subspace \(L\) is isotropic in \(V_n\). We have
\begin{equation}
\dim L = \dim \operatorname{rad}(K)+\dim \overline L = \frac{\dim K+\dim \operatorname{rad}(K)}{2}.
\end{equation}
This proves the attainability of the dimension in \cref{eqn:witt-index-formula}. Conversely, let \(M\subseteq K\) be any subspace that is isotropic in \(V_n\). Then \(M+\operatorname{rad}(K)\subseteq K\) is also isotropic, because \(M\) is isotropic and \(\operatorname{rad}(K)\) pairs trivially with all of \(K\). The image of \(M+\operatorname{rad}(K)\) in \(K/\operatorname{rad}(K)\) is therefore isotropic in the non-degenerate symplectic vector space \(K/\operatorname{rad}(K)\). Hence, we have
\begin{equation}
\dim  ((M+\operatorname{rad}(K))/\operatorname{rad}(K) )
\leq
\frac{1}{2}\dim (K/\operatorname{rad}(K)).
\end{equation}
Therefore, we have
\begin{align}
\dim M&\leq \dim (M+\operatorname{rad}(K))\\ &= \dim \operatorname{rad}(K) + \dim  ((M+\operatorname{rad}(K))/\operatorname{rad}(K) )\\ &\leq \dim \operatorname{rad}(K) + \frac{1}{2}\dim (K/\operatorname{rad}(K))= \frac{\dim K+\dim \operatorname{rad}(K)}{2}.
\end{align}
Thus, no isotropic subspace of $V_n$ that is contained in \(K\) has dimension larger than the number in \cref{eqn:witt-index-formula}. This completes the proof.
\end{proof}

\cref{lem:witt-index-subspace} shows that the integer \(w(K)\) defined in \cref{eqn:witt-index-formula} is the maximum dimension of an isotropic subspace contained in \(K\). In particular, \(K\) contains a Lagrangian subspace of \(V_n\) if and only if \(w(K)\geq n\). We can use this observation to obtain an alternative characterization of semi-Clifford gates in \cref{witt-index-consequences}.

\begin{cor}\label{witt-index-consequences}
Let \(p\) be an odd prime. Let \(U \in \mathcal C^{(3)}\), and set \(K_U := \ker S_U\). Let \(w(K_U)\) denote the maximum dimension of a subspace of \(K_U\) that is isotropic in \(V_n\). Then the following statements hold:
\begin{enumerate}
\item \(U\) is semi-Clifford if and only if
\begin{equation}\label{eqn:witt-index-semi-clifford}
\dim K_U + \dim(K_U \cap K_U^\perp) \geq 2n.
\end{equation}
\item \(U\) is semi-Clifford if and only if \(K_U^\perp \subseteq K_U\). That is, $K_U$ is a coisotropic subspace of $V_n$. 
\end{enumerate}
\end{cor}

\begin{proof}
By \cref{semi-clifford-kernel-criterion}, \(U\) is semi-Clifford if and only if \(K_U\) contains a Lagrangian subspace of \(V_n\). Since a Lagrangian subspace has dimension \(n\), \cref{lem:witt-index-subspace} implies this is the condition
\begin{equation}
w(K_U)\geq n \iff  \dim K_U+\dim(K_U\cap K_U^\perp) \geq 2n.
\end{equation}
This proves (1). For (2), if \(K_U^\perp\subseteq K_U\), then $K_U\cap K_U^\perp=K_U^\perp$. Hence,
\begin{equation}
\dim K_U+\dim(K_U\cap K_U^\perp)=\dim K_U+\dim K_U^\perp=2n.
\end{equation}
Thus, \(U\) is semi-Clifford by (1). Conversely, suppose that \(U\) is semi-Clifford. Then \cref{eqn:witt-index-semi-clifford} holds. Since $\dim K_U+\dim K_U^\perp=2n$, we have $\dim(K_U\cap K_U^\perp)\geq \dim K_U^\perp.$ On the other hand, \(K_U\cap K_U^\perp\) is a subspace of \(K_U^\perp\), so $\dim(K_U\cap K_U^\perp)\leq \dim K_U^\perp.$ Hence, $\dim(K_U\cap K_U^\perp)=\dim K_U^\perp$ and therefore
$K_U\cap K_U^\perp=K_U^\perp.$
Thus, \(K_U^\perp\subseteq K_U\). 
\end{proof}

\cref{codim2-gen} puts the formalism to the test and proves that, in general, if \(\operatorname{codim}\ker S_U\leq 2\), then every \(U\in\mathcal C^{(3)}\) is semi-Clifford.

\begin{prop}\label{codim2-gen}
Let \(p\) be an odd prime. Let \(U \in \mathcal C^{(3)}\), and set \(K_U := \ker S_U\).  If \(\operatorname{codim} K_U \leq 2\), then \(U\) is semi-Clifford.
\end{prop}

\begin{proof}
The \(\operatorname{codim} K_U=0\) case is immediate, since then \(K_U=V_n\). If \(\operatorname{codim} K_U=1\), then \(K_U\) is a hyperplane and \(\dim K_U^\perp=1\). Moreover, $K_U^\perp \subseteq K_U$. Indeed, let \(0\neq v\in K_U^\perp\). If \(v\notin K_U\), then $V_n=K_U\oplus \langle v\rangle$ since \(K_U\) is a hyperplane. Since $[v,v]=0$ and $[v,k]=0$ for each $k \in K_U$, the decomposition $V_n=K_U\oplus \langle v\rangle$ implies that $[v,x]=0$ for each $x \in V_n$. This contradicts the non-degeneracy of the symplectic form. Hence, $K_U^\perp \subseteq K_U$. Thus, \(U\) is semi-Clifford by \cref{witt-index-consequences}. Let $\operatorname{codim}K_U=2$. 
\begin{enumerate}
    \item[(i)] Suppose first that \(K_U\cap K_U^\perp\neq 0\). Note that \(K_U\cap K_U^\perp\) is the kernel of the map
\begin{align}
\varphi_K \colon K_U &\longrightarrow K_U^* \label{eq-1}, \\
x &\longmapsto \varphi_K(x), \qquad \varphi_K(x)(y) := [x,y].
\end{align}
By rank-nullity, $\dim K_U = \operatorname{rank}\varphi_K+\dim(K_U\cap K_U^\perp)$. 
Hence, 
\begin{equation}
\operatorname{rank}\varphi_K
=
\dim K_U-\dim(K_U\cap K_U^\perp). \label{eq-2}
\end{equation}
The map \(\varphi_K\) is represented, after choosing a basis of \(K_U\), by a skew-symmetric matrix. Every skew-symmetric matrix has even rank. Therefore, $\dim K_U - \dim(K_U\cap K_U^\perp)$ is even. Since \(\dim K_U=2n-2\) is even, it follows that $\dim(K_U\cap K_U^\perp)$ is also even. Since \(K_U\cap K_U^\perp\neq 0\), we conclude that $\dim(K_U\cap K_U^\perp)\geq 2$. Therefore, we have 
\begin{equation}
\dim K_U+\dim(K_U\cap K_U^\perp)\geq (2n-2)+2=2n.
\end{equation}
Thus, \(U\) is semi-Clifford by \cref{witt-index-consequences}. 

\item[(ii)] Suppose now that \(K_U\cap K_U^\perp = 0\). That is, the symplectic form restricts to a non-degenerate alternating form on \(K_U\). We show that this case cannot occur. Assume, for contradiction, that the symplectic form restricts to a non-degenerate alternating form on \(K_U\). By \cref{lem:kernel-cocycle-consequences}, the restriction \(q_U|_{K_U}\) is a linear injective map. Since $q_U|_{K_U}$ is linear and bijective onto 
$W:=q_U(K_U)$, we have  $\dim q_U(K_U)=\dim K_U=2n-2$.  By \cref{lem:kernel-cocycle-consequences}, we also have
\begin{equation}\label{restrict-nondeg}
[q_U(a),q_U(b)]=[a,b]
\end{equation}
for all \(a,b\in K_U\). Since \(K_U\cap K_U^\perp = 0\), \cref{restrict-nondeg} implies that the symplectic form restricts to a non-degenerate alternating form on \(W\). Since \(W\) is non-degenerate, we have a symplectic direct sum $V_n=W\oplus W^\perp$ such that $\dim W^\perp = 2$\footnote{Indeed, \(W\cap W^\perp=\operatorname{rad}(W)=0\), and non-degeneracy of the ambient symplectic form gives \(\dim W+\dim W^\perp=\dim V_n\). Hence \(V_n=W\oplus W^\perp\). Since \(\dim W=2n-2\), we get \(\dim W^\perp=2\).}.
Set $E_U := \operatorname{im}S_U$. By \cref{lem:kernel-cocycle-consequences}, we have that every element of \(E_U \) fixes \(W\) pointwise. Hence, each \(T\in \operatorname{im}S_U\) preserves \(W^\perp\). Indeed, if \(x\in W^\perp\) and \(w\in W\) then \(Tw=w\) and
\begin{equation}
[Tx,w]=[Tx,Tw]=[x,w]=0.
\end{equation}
Thus, \(Tx\in W^\perp\). Hence, restriction to \(W^\perp\) defines a group homomorphism
\begin{align}
\rho \colon E_U  &\longrightarrow \operatorname{Sp}(W^\perp), \\
T &\longmapsto T|_{W^\perp}.
\end{align}
$\rho$ is injective. Indeed, if \(\rho(T)=\operatorname{Id}_{W^\perp}\), then \(T\) is the identity on \(W^\perp\). Since \(T\) also fixes \(W\) pointwise and \(V_n=W\oplus W^\perp\), it follows that $T=\operatorname{Id}_{V_n}$. Therefore, \(E_U \) embeds into \(\operatorname{Sp}(W^\perp) \cong \operatorname{Sp}(\mathbb F^2_{p})\). Since
\begin{equation}
|\operatorname{Sp}(\mathbb F^2_p)|=p(p^2-1),
\end{equation}
every \(p\)-subgroup of \(\operatorname{Sp}(\mathbb F_p^2)\) has order at most \(p\). However, \(E_U\cong V_n/K_U\), and hence
$|E_U|=p^2.$
This contradiction shows that the case \(K_U\cap K_U^\perp=0\) cannot occur.
\end{enumerate}
This completes the proof.
\end{proof}

The \(\operatorname{codim}K_U=2\) case in \cref{codim2-gen}, under the assumption \(K_U\cap K_U^\perp=0\), passed from \(K_U\) to \(W:=q_U(K_U)\) and used the cocycle identities to show that \(E_U\) fixes the non-degenerate subspace \(W\) pointwise. Consequently, \(E_U\) acts faithfully on the ``smaller'' symplectic space \(W^\perp\cong \mathbb F_p^2\), where abelian \(p\)-subgroups are too small to contain \(E_U\). Similar arguments are used below.

\subsection{Two Qudit Case}\label{two}
We now give a new proof of the known result that every two-qudit gate of odd prime dimension in \(\mathcal C^{(3)}\) is semi-Clifford.  De Silva proved the $p=3$ case~\cite{deSilva2021efficient} and Chen and de Silva later established the $p>3$ case~\cite{chen2024characterising} using techniques from algebraic geometry. Hence, \cref{two-qudit-c3-semi-clifford} provides an application of our formalism by giving a new and simpler proof of the two-qudit result based on different techniques. 
We restrict to the case \(p\geq 3\), since the \(p=2\) case is well known and the formalism developed in \cref{third-level-classification} is formulated primarily for odd \(p\), with the $p=2$ modifications treated separately in \cref{appendix:characteristic-two-third-level}. We will repeatedly use the following fact.

\begin{fact}[Wong {\cite[Theorem 7]{wong1982abelian}}]\label{wong}
Let $p$ be an odd prime and let \(V\) be a \(2m\)-dimensional symplectic vector space over $\mathbb F_p$. Let
\begin{equation}
0=W_0\subset W_1\subset \cdots \subset W_m
\end{equation}
be a complete isotropic flag, so that \(\dim W_i=i\), each \(W_i\) is isotropic, and \(W_m\) is Lagrangian. Let \(P=P(W_0,W_1,\ldots,W_m)\) be the Sylow \(p\)-subgroup of \(\operatorname{Sp}(V)\) consisting of elements fixing each quotient \(W_i/W_{i-1}\) for $i=1,\cdots,m$. For a subspace \(X\subseteq V\), write \(A(X)\) for the subgroup of \(\operatorname{Sp}(V)\) fixing \(X\) pointwise. Then \(A(W_m)\) is the unique abelian subgroup of largest order in \(P\).
\end{fact}

We can now state \cref{two-qudit-c3-semi-clifford}. 

\begin{thm}
\label{two-qudit-c3-semi-clifford}
Let \(p\) be an odd prime. Every two-qudit gate in \(\mathcal C^{(3)}\) is semi-Clifford.
\end{thm}

Since the proof of \cref{two-qudit-c3-semi-clifford} requires the analysis of several cases, we organize the argument into a sequence of steps. Let \(U\in\mathcal{C}^{(3)}\) and set \(K_U:=\ker S_U\). By \cref{codim2-gen}, \(U\) is semi-Clifford whenever \(\operatorname{codim} K_U\leq 2\). It therefore remains to consider the cases \(\operatorname{codim} K_U=3,4\). We will show that neither of these cases can occur. \cref{codim4-two-qubit-rule} rules out the $\operatorname{codim} K_U = 4$ case. 

\begin{prop}\label{codim4-two-qubit-rule}
Let \(p\) be an odd prime, and let \(U\in\mathcal{C}^{(3)}\) be a two-qudit gate. Set \(K_U:=\ker S_U\) and \(E_U:=\operatorname{im} S_U\). Then
$\operatorname{codim} K_U\neq 4.$
\end{prop}

\begin{proof}
We first rule out the \(\operatorname{codim}K_U=4\) case. Assume, for the sake of contradiction, that \(\operatorname{codim}K_U=4\). We have $E_U\cong \mathbb F^4_p/K_U \cong \mathbb F^4_p$. Thus, \(\operatorname{im}S_U\) would be an abelian $p$-subgroup of \(\operatorname{Sp}(\mathbb F^4_p)\) of order \(p^4\). The standard order formula gives
\begin{equation}
|\operatorname{Sp}(\mathbb F^4_p)|=p^4(p^2-1)(p^4-1).
\end{equation}
Thus, an abelian $p$-subgroup of order \(p^4\) would be a Sylow \(p\)-subgroup. However, Sylow \(p\)-subgroups of \(\operatorname{Sp}(\mathbb F^4_p)\) are non-abelian. We can verify this claim explicitly. Let
\begin{equation}
A_t:=
\begin{pmatrix}
1&t\\
0&1
\end{pmatrix},
\quad t \in \mathbb F_p.
\end{equation}
Consider the set
\begin{equation}\label{sylow-sp4}
P:=\left\{\begin{pmatrix}A_t & A_t C\\0 & (A_t^{-1})^T\end{pmatrix}:t\in\mathbb F_p,\ C=C^T\in M_2(\mathbb F_p)\right\}
\end{equation}
Let
\begin{equation}\label{M-elem}
M=
\begin{pmatrix}
A_t&A_t C\\
0&(A_t^{-1})^T
\end{pmatrix},
\qquad C=C^T \in M_2(\mathbb F_p).
\end{equation}
With respect to the standard symplectic matrix
\begin{equation}\label{standard-sym-matrix}
J=\begin{pmatrix}0&I\\-I&0\end{pmatrix},
\end{equation}
we compute
\begin{align}
M^TJM
&=
\begin{pmatrix}
A_t^T&0\\
C^TA_t^T&A_t^{-1}
\end{pmatrix}
\begin{pmatrix}
0&I\\
-I&0
\end{pmatrix}
\begin{pmatrix}
A_t&A_t C\\
0&(A_t^{-1})^T
\end{pmatrix} \label{sym-1} \\
&=
\begin{pmatrix}
A_t^T&0\\
C^TA_t^T&A_t^{-1}
\end{pmatrix}
\begin{pmatrix}
0&(A_t^{-1})^T\\
-A_t&-A_tC
\end{pmatrix}=
\begin{pmatrix}
0&I\\
-I&C^T-C
\end{pmatrix}
=
\begin{pmatrix}
0&I\\
-I&0
\end{pmatrix}
=J. \label{sym-2}
\end{align}
Hence, $P \subseteq \operatorname{Sp}(\mathbb F^4_p)$. It is straightforward to verify that $P$ is a subgroup. Moreover, $P$ is of order $p^4$. Indeed, \(t\) has \(p\) possible values and a symmetric \(2 \times 2\) matrix in $M_2(\mathbb F_p)$ has \(p^{3}\) possible values. Consequently, $|P|=p\cdot p^{3}=p^{4}$.  Hence, $P$ is a Sylow \(p\)-subgroup. $P$ is non-abelian. Indeed, let $A=A_1$ and $C=\begin{pmatrix}0&0\\0&1\end{pmatrix}$.
Define
\begin{equation}
D:=\begin{pmatrix}
A&0\\
0& (A^{-1})^T
\end{pmatrix},
\qquad
U_C:=\begin{pmatrix}I&C\\ 0&I\end{pmatrix}.
\end{equation}
Both \(D\) and \(U_C\) belong to \(P\). Note that we have
\begin{align}
DU_CD^{-1}
&=
\begin{pmatrix}
A&0\\
0&(A^{-1})^T
\end{pmatrix}
\begin{pmatrix}
I&C\\
0&I
\end{pmatrix}
\begin{pmatrix}
A^{-1}&0\\
0&A^T
\end{pmatrix}=
\begin{pmatrix}
I&ACA^T\\
0&I
\end{pmatrix}.
\end{align}
Since
\begin{equation}
ACA^T
=
\begin{pmatrix}
1&1\\
1&1
\end{pmatrix}
\neq
\begin{pmatrix}
0&0\\
0&1
\end{pmatrix}
=
C,
\end{equation}
we have $DU_CD^{-1}\neq U_C$. Hence, \(P\) is non-abelian. Since all Sylow \(p\)-subgroups are conjugate, no Sylow \(p\)-subgroup of \(\operatorname{Sp}(\mathbb F^4_p)\) is abelian. Therefore, \(\operatorname{Sp}(\mathbb F^4_p)\) cannot contain an abelian subgroup of order \(p^4\). This contradicts \(E_U\cong \mathbb F^4_p\). Hence, $\operatorname{codim}K_U\neq 4$. 
\end{proof}

\cref{codim3-two-qubit-rule} rules out the $\operatorname{codim} K_U = 3$ case. 
The strategy is to use \cref{wong} to identify \(E_U\) with a specific subgroup, thereby converting the symplectic data \(S_U\) into a surjective linear map from \(\mathbb F_p^4\) onto \(\operatorname{Sym}_2(\mathbb F_p)\). The affine cocycle identity then forces the \(L^*\) component of \(q_U\) to vanish, so that \(q_U(a)\)'s lie in a Lagrangian subspace. Consequently, the operators \(B_U(a)\) commute, contradicting the non-degeneracy of the symplectic form.

\begin{prop}\label{codim3-two-qubit-rule}
Let \(p\) be an odd prime, and let \(U\in\mathcal{C}^{(3)}\) be a two-qudit gate. Set \(K_U:=\ker S_U\) and \(E_U:=\operatorname{im} S_U\). Then
$\operatorname{codim} K_U\neq 3.$
\end{prop}

\begin{proof}
We divide the proof into several steps to make the otherwise lengthy argument easier to follow.
 
\begin{enumerate}
\item  First, we reduce by Clifford conjugation so that \(E_U\) lies in the fixed Sylow \(p\)-subgroup \(P\) in \cref{sylow-sp4}. 
\(E_U\cong \mathbb F^4_p/K_U\) is an abelian \(p\)-subgroup of \(\operatorname{Sp}(\mathbb F_p^4)\) of order \(p^3\). Since \(E_U\) is a \(p\)-subgroup, \(E_U\) is contained in a Sylow \(p\)-subgroup of \(\operatorname{Sp}(\mathbb F_p^4)\). Since all Sylow \(p\)-subgroups of \(\operatorname{Sp}(\mathbb F_p^4)\) are conjugate and \(P\) from \cref{sylow-sp4} is a Sylow \(p\)-subgroup, there exists \(R\in\operatorname{Sp}(\mathbb F_p^4)\) such that $RE_UR^{-1}\subseteq P$. Choose a Clifford lift \(C_R\) of \(R\). Let \(U':=C_RU\). Since \(C_R\in \mathcal C^{(2)}\), the gate \(U\) is semi-Clifford if and only if \(U'\) is semi-Clifford. Moreover,
\begin{align}
B_{U'}(a)
&=U'w(a)U'^* \\
&=C_RB_U(a)C_R^* \\
&=C_R (\alpha_U(a)w(q_U(a))M_{S_U(a)} )C_R^*=\alpha_{U'}(a)w(Rq_U(a))M_{RS_U(a)R^{-1}}\label{last-eq}.
\end{align}
for every \(a\in \mathbb F^4_p\). \cref{last-eq} follows because conjugating \(M_{S_U(a)}\) by \(C_R\) produces a gate in \(\mathcal C^{(2)}\) with symplectic part \(RS_U(a)R^{-1}\) up to a global phase. Hence, $S_{U'}(a)=RS_U(a)R^{-1}$ for every \(a\in \mathbb F^4_p\). It follows that $a\in K_{U'}$ if and only if $a\in K_U$. Thus, $K_{U'}=K_U$. Consequently, after replacing \(U\) by \(U'\), we may assume that \(E_U\subseteq P\).

\item Next, we use \cref{wong} to identify the 
unique abelian subgroup of largest order in \(P\), and then identify \(E_U\) with this subgroup.
Consider $\mathbb F_p^4$ with the symplectic basis $e_1,e_2,f_1,f_2$. Define $L:=\operatorname{span}\{e_1,e_2\}$ and  $L^*:=\operatorname{span}\{f_1,f_2\}$. Thus, $\mathbb F_p^4=L\oplus L^*$. Consider the complete isotropic flag
\begin{equation}\label{flag}
0\subseteq \langle e_1\rangle \subseteq \langle e_1, e_2 \rangle = L.
\end{equation}
We verify that \(P\) in \cref{sylow-sp4} is the Sylow \(p\)-subgroup attached to this flag as required in  \cref{wong}. Let $M$ be as in \cref{M-elem}. Since the lower-left block of \(M\) is zero, \(M\) preserves \(L\). Its induced action on \(L\) is given by \(A_t\). Since
\begin{equation}
A_te_1=e_1,
\qquad
A_te_2=e_2+te_1,
\end{equation}
the line \(W_1=\langle e_1\rangle\) is fixed pointwise, and the induced action on \(W_2/W_1\) is trivial. Hence, $P$ satisfies the assumption in \cref{wong}. Hence, \cref{wong} implies that $A(L)$ is the unique abelian subgroup of largest order in \(P\). We identify $A(L)$ explicitly. Consider the set
\begin{equation}
Q:=\left\{\begin{pmatrix}I&C\\0&I\end{pmatrix}:C=C^T\in M_2(\mathbb F_p)\right\}\subseteq P.
\end{equation}
Clearly, $Q$ is an abelian subgroup of $P$ and $Q$ fixes $L$. Hence, $Q \subseteq A(L)$. Conversely, if $M$ as in \cref{M-elem} fixes \(L\) pointwise then
\begin{equation}
M(x,0)=(A_tx,0)=(x,0)
\end{equation}
for every \(x\in \mathbb F_p^2\). Hence \(A_t=I\), so \(t=0\), and therefore \(M \in Q\). Hence, we have $A(L)=Q$. Since a symmetric \(2\times 2\) matrix has three independent entries, we have $|A(L)|=p^3.$ On the other hand, \(E_U\) is also an abelian subgroup of $P$ of  order $p^3$. Hence, we have $E_U=A(L)$ by \cref{wong}. That is, we have
\begin{equation}\label{ul-group}
E_U=\left\{\begin{pmatrix}I&C\\0&I\end{pmatrix}:C=C^T\in M_2(\mathbb F_p)\right\}:= A(L)
\end{equation}

\item We now express the symplectic data \(S_U\) as a linear map into \(\operatorname{Sym}_2(\mathbb F_p)\). For every \(a\in \mathbb F^4_p\) there is a unique symmetric matrix \(C(a)\in M_2(\mathbb F_p)\) such that
\begin{equation}\label{su-codim-3}
S_U(a)=\begin{pmatrix}I&C(a)\\0&I\end{pmatrix}.
\end{equation}
Since \(S_U:\mathbb F_p^4\to A(L)\) is a group homomorphism and multiplication in \(A(L)\) corresponds to addition in \(\operatorname{Sym}_2(\mathbb F_p)\), the map
\begin{align}
C:\mathbb F_p^4&\longrightarrow \operatorname{Sym}_2(\mathbb F_p),\\
a&\longmapsto C(a)
\end{align}
is \(\mathbb F_p\)-linear. Moreover, \(\ker C=\ker S_U=K_U\). Since \(\operatorname{codim}K_U=3\), we have
\begin{equation}
\operatorname{rank}C=\dim \mathbb F_p^4-\dim K_U=4-1=3.
\end{equation}
Since \(\dim \operatorname{Sym}_2(\mathbb F_p)=3\), it follows that \(C\) is surjective.

\item  We now use the cocycle equation for \(q_U\). Write $q_U(a)=(\alpha(a),\beta(a))\in L\oplus L^*$. Using \cref{su-codim-3}, we have
\begin{equation}
S_U(a)q_U(b)
=
\begin{pmatrix}
I&C(a)\\
0&I
\end{pmatrix}
\begin{pmatrix}
\alpha(b)\\
\beta(b)
\end{pmatrix}
=
\begin{pmatrix}
\alpha(b)+C(a)\beta(b)\\
\beta(b)
\end{pmatrix}.
\end{equation}
Substituting this into the affine cocycle equation gives
\begin{equation}
\begin{pmatrix}
\alpha(a+b)\\
\beta(a+b)
\end{pmatrix}
=q_U(a+b)=q_U(a)+S_U(a)q_U(b)=\begin{pmatrix}
\alpha(a)+\alpha(b)+C(a)\beta(b)\\
\beta(a)+\beta(b)
\end{pmatrix}.
\end{equation}
Hence, we have
\begin{align}
\alpha(a+b)&=\alpha(a)+\alpha(b)+C(a)\beta(b),\\
\beta(a+b)&=\beta(a)+\beta(b).
\end{align}
Since \(\mathbb F^4_p\) is abelian, the same comparison as in \cref{lem:kernel-cocycle-consequences} applied to \(q_U(a+b)=q_U(b+a)\) gives \(C(a)\beta(b)=C(b)\beta(a)\) for all \(a,b\in\mathbb F_p^4\).

\item We show that the surjectivity of \(C\) forces the \(\beta\) component of $q$ to vanish. Since \(C:\mathbb F_p^4\to \operatorname{Sym}_2(\mathbb F_p)\) is surjective, choose elements \(a_{11},a_{22},a_{12}\in \mathbb F_p^4\) such that
\begin{equation}
C(a_{11})=E_{11},\qquad C(a_{22})=E_{22},\qquad C(a_{12})=E_{12},
\end{equation}
where \(E_{11}\), \(E_{12}\), and \(E_{22}\) are the standard basis matrices of \(\operatorname{Sym}_2(\mathbb F_p)\). Write
\begin{equation}
\beta(a_{11})=(u_1,u_2),
\qquad
\beta(a_{22})=(v_1,v_2),
\qquad
\beta(a_{12})=(w_1,w_2).
\end{equation}
Applying \(C(a)\beta(b)=C(b)\beta(a)\) to the pairs \((a_{11},a_{22})\), \((a_{11},a_{12})\) and \((a_{22},a_{12})\) gives
\begin{align}
(v_1,0)=E_{11}\beta(a_{22})&=E_{22}\beta(a_{11})=(0,u_2),\\
(w_1,0)=E_{11}\beta(a_{12})&=E_{12}\beta(a_{11})=(u_2,u_1),\\
(0,w_2)=E_{22}\beta(a_{12})&=E_{12}\beta(a_{22})=(v_2,v_1).
\end{align}
Hence, \(u_{1,2}=v_{1,2}=w_{1,2}=0\). Therefore,
\begin{equation}
\beta(a_{11})=\beta(a_{22})=\beta(a_{12})=0.
\end{equation}
Now let \(a\in\mathbb F_p^4\) be arbitrary. Since \(E_{11},E_{22},E_{12}\) form a basis of \(\operatorname{Sym}_2(\mathbb F_p)\), there exist \(c_{11},c_{22},c_{12}\in\mathbb F_p\) such that
\begin{equation}
C(a)=c_{11}E_{11}+c_{22}E_{22}+c_{12}E_{12}.
\end{equation}
Set $a_0:=c_{11}a_{11}+c_{22}a_{22}+c_{12}a_{12}$. Then $C(a_0)=C(a)$, so \(a-a_0\in\ker C\). We next show that \(\beta\) vanishes on \(\ker C\). If \(k\in\ker C\), then \(C(k)=0\). Hence, we have
\begin{equation}
C(a)\beta(k)=C(k)\beta(a)=0
\end{equation}
for every \(a\in\mathbb F_p^4\). Since \(C\) is surjective, every symmetric matrix annihilates \(\beta(k)\). In particular, the identity matrix annihilates \(\beta(k)\), so $\beta(k)=0$. Thus \(\beta(a-a_0)=0\). Since \(\beta\) is linear, we get $\beta(a)=\beta(a_0)$. Since
\begin{equation}
\beta(a_0)=c_{11}\beta(a_{11})+c_{22}\beta(a_{22})+c_{12}\beta(a_{12})=0,
\end{equation}
we have \(\beta(a)=0\) for every \(a\in\mathbb F_p^4\). Hence, $\beta=0$. Consequently, $q_U(\mathbb F_p^4)\subseteq L$.

\item Finally, we show that the vanishing of \(\beta\) forces all \(B_U(a)\) to commute, contradicting the Weyl commutator relation. Since \(L\) is Lagrangian, the Weyl operators \(w(q_U(a))\) and \(w(q_U(b))\) commute for all \(a,b\in \mathbb F_p^4\). Since every element of \(E_U\) fixes \(L\) pointwise, we have \(S_U(a)q_U(b)=q_U(b)\) for all \(a,b \in \mathbb F_p^4\). Since \(M_{S_U(a)}\) is a lift of \(S_U(a)\), we have
\begin{equation}
M_{S_U(a)}w(q_U(b))M_{S_U(a)}^*=w(S_U(a)q_U(b))=w(q_U(b)).
\end{equation}
Thus,
\begin{equation}
M_{S_U(a)}w(q_U(b))=w(q_U(b))M_{S_U(a)}
\end{equation}
for all \(a,b\in \mathbb F_p^4\). By \cref{lem:upper-unipotent-commuting-lifts} applied with \(m=2\), the fixed Clifford representatives \(M_{S_U(a)}=M_{u_{C(a)}}\) commute pairwise. Putting everything together, we have
\begin{equation}\label{compare-1}
B_U(a)B_U(b)=B_U(b)B_U(a)
\end{equation}
for all \(a,b\in\mathbb F_p^4\). On the other hand, by definition \(B_U(a)=Uw(a)U^*\). Hence, we have
\begin{align}\label{compare-2}
B_U(a)B_U(b)=Uw(a)w(b)U^*=\omega^{[a,b]}Uw(b)w(a)U^*=\omega^{[a,b]}B_U(b)B_U(a).
\end{align}
Comparing \cref{compare-1} and \cref{compare-2} gives $\omega^{[a,b]}=1$ for all \(a,b\in\mathbb F_p^4\). Since \(\omega\) is a primitive \(p\)-th root of unity, this implies $[a,b]=0$ for all \(a,b\in\mathbb F_p^4\). This contradicts the non-degeneracy of the symplectic form. Therefore, \(\operatorname{codim}K_U\neq3\).
\end{enumerate}
This completes the proof.
\end{proof}

We can now prove \cref{two-qudit-c3-semi-clifford}.

\begin{proof}[Proof of \cref{two-qudit-c3-semi-clifford}]
This follows from \cref{codim2-gen}, \cref{codim4-two-qubit-rule}, \cref{codim3-two-qubit-rule}.
\end{proof}

The \(\operatorname{codim} K_U=3\) part of the proof of \cref{two-qudit-c3-semi-clifford} places \(E_U=\operatorname{im}S_U\) inside a Sylow \(p\)-subgroup of \(\operatorname{Sp}(\mathbb F_p^4)\), uses \cref{wong} to identify \(E_U\) with the subgroup fixing a Lagrangian \(L\) pointwise, and then applies the affine cocycle identity to constrain \(q_U(a)=(\alpha(a),\beta(a))\). In the two-qudit case, this forces \(\beta=0\), contradicting the non-degeneracy of the symplectic form. Similar arguments will be used below in the three-qudit case.

\subsection{Three Qudit Case}\label{three}
We now prove that every three-qudit gate of odd prime dimension in \(\mathcal{C}^{(3)}\) is semi-Clifford. To the best of our knowledge, this result has not previously appeared in the literature. It therefore provides a substantive application of the formalism developed in this work.

\begin{thm}
\label{three-qudit-c3-semi-clifford}
Let \(p\) be an odd prime. Every three-qudit gate in \(\mathcal C^{(3)}\) is semi-Clifford.
\end{thm}


Since the proof of \cref{three-qudit-c3-semi-clifford} requires checking several cases, we organize the argument below into multiple steps. Let \(U\in \mathcal C^{(3)}\) and \(K_U:=\ker S_U\). By \cref{codim2-gen}, if \(\operatorname{codim}K_U\leq 2\) then \(U\) is semi-Clifford. Thus, it remains to consider the cases \(\operatorname{codim}K_U=3,4,5,6\). We first prove \cref{lem:pointwise-lagrangian-coordinate-setup} which will be used repeatedly below.

\begin{lem}\label{lem:pointwise-lagrangian-coordinate-setup}
Let \(p\) be an odd prime. Let \(U\in\mathcal C^{(3)}\) be a three-qudit gate and suppose that \(E_U:=\operatorname{im}S_U\) fixes a Lagrangian subspace \(L\subseteq \mathbb F^6_p\) pointwise. Choose a complementary Lagrangian \(L^*\), so that $\mathbb F^6_p=L\oplus L^*.$ Then for every \(a\in\mathbb F^6_p\), there is a symmetric matrix \(C(a)\in\operatorname{Sym}_3(\mathbb F_p)\) such that
\begin{equation}\label{eqn:pointwise-lagrangian-upper-unipotent}
S_U(a)=u_{C(a)}:=\begin{pmatrix}I&C(a)\\0&I\end{pmatrix}.
\end{equation}
Moreover, the map $C:\mathbb F^6_p \to \operatorname{Sym}_3(\mathbb F_p)$
is linear. If $q_U(a)=(\alpha(a),\beta(a))\in L\oplus L^*,$
then \(\beta:\mathbb F^6_p\to L^*\) is linear and
$C(a)\beta(b)=C(b)\beta(a)$
for all \(a,b\in\mathbb F^6_p\).
\end{lem}

\begin{proof}
Write \(S_U(a)\) in block form relative to the decomposition \(\mathbb F_p^6=L\oplus L^*\):
\begin{equation}
S_U(a)=\begin{pmatrix}A&B\\C&D\end{pmatrix}.
\end{equation}
Since $S_U(a)$ fixes $L$ for each $a \in \mathbb F_p^6$, we have 
\begin{equation}
(x,0)=S_U(a)(x,0)=(Ax,Cx).
\end{equation}
for each $x \in L$. Thus, \(A=I\) and \(C=0\). Since \(S_U(a)\) is symplectic, a computation analogous to \cref{sym-1}--\cref{sym-2} gives \(D=I\) and \(B=B^T\). Therefore
\begin{equation}
S_U(a)=\begin{pmatrix}I&C(a)\\0&I\end{pmatrix}
\end{equation}
with \(C(a)=C(a)^T\). Thus, \(C(a)\in\operatorname{Sym}_3(\mathbb F_p)\). Since \(S_U:\mathbb F^6_p\to\operatorname{Sp}(\mathbb F^6_p)\) is a homomorphism and \(u_Cu_D=u_{C+D}\), the map $C:\mathbb F^6_p \to  \operatorname{Sym}_3(\mathbb F_p)$ is linear.  Write
\begin{equation}
q_U(a)=(\alpha(a),\beta(a))\in L\oplus L^*.
\end{equation}
Applying the cocycle identity and comparing the \(L\) and \(L^*\) components, as in \cref{codim3-two-qubit-rule}, gives
\begin{align}
\alpha(a+b)&=\alpha(a)+\alpha(b)+C(a)\beta(b),\\
\beta(a+b)&=\beta(a)+\beta(b).
\end{align}
Hence, \(\beta:\mathbb F^6_p\to L^*\) is linear. Since \(\mathbb F^6_p\) is abelian, the same cocycle comparison as in \cref{lem:kernel-cocycle-consequences} applied to \(q_U(a+b)=q_U(b+a)\)  gives $C(a)\beta(b)=C(b)\beta(a)$ for all \(a,b\in\mathbb F^6_p\).
\end{proof}

\cref{lem:three-qudit-codim-three} treats the \(\operatorname{codim}K_U=3\) case. The non-trivial part of the proof combines the two strategies used above. First, as in the codimension two case of \cref{codim2-gen}, we pass from \(K_U\) to \(W:=q_U(K_U)\). However, \(W\) is not non-degenerate in this case. The key difference is therefore that we find a non-degenerate symplectic plane \(H\subseteq W\) and embed \(E_U\) into \(\operatorname{Sp}(H^\perp)\cong\operatorname{Sp}(\mathbb F_p^4)\). Second, as in \cref{codim3-two-qubit-rule}, \cref{wong} identifies $\operatorname{im} S_U$ with the pointwise stabilizer of a Lagrangian plane in this four-dimensional symplectic space, $H^\perp$. The difference from the two-qudit case is that here the analysis is restricted to \(H^\perp\), rather than to all of \(\mathbb F_p^6\). Hence, we enlarge the fixed Lagrangian plane in \(H^\perp\) to a Lagrangian three-plane in \(\mathbb F_p^6\) before using the cocycle identities to derive a contradiction.

\begin{prop}\label{lem:three-qudit-codim-three}
Let \(p\) be an odd prime and let \(U\in\mathcal C^{(3)}\) be a three-qudit gate. Set \(K_U:=\ker S_U\) and \(E_U:=\operatorname{im}S_U\). If \(\operatorname{codim}K_U=3\), then \(U\) is semi-Clifford.
\end{prop}

\begin{proof}
Since \(\operatorname{codim}K_U=3\), we have \(\dim K_U=3\). Set \(r_K:=\dim(K_U\cap K_U^\perp)\). As in \cref{eq-1}--\cref{eq-2} of \cref{codim2-gen}, the number \(3-r_K\) is the rank of the alternating bilinear form obtained by restricting the symplectic form to \(K_U\). Since every alternating bilinear form is represented by a skew-symmetric matrix, \(3-r_K\) is even.  Hence, \(r_K=1\) or \(r_K=3\). If \(r_K=3\), then \(K_U\cap K_U^\perp=K_U\), so \(K_U\subseteq K_U^\perp\).
Since $\dim K_U = 3$ in $\mathbb F_p^6$, this implies $K_U$ is a Lagrangian subspace. Hence, \(U\) is semi-Clifford. 
We rule out the \(r_K=1\) case.  We divide the proof into several steps to make the
otherwise lengthy argument easier to follow. 

\begin{enumerate} 
\item First, we reduce from \(K_U\) to a non-degenerate symplectic plane \(H\), so that \(E_U\) can be embedded into a four-dimensional symplectic group. By \cref{lem:kernel-cocycle-consequences}, \(q_U|_{K_U}\) is linear, injective, and preserves the restricted symplectic form. Hence, $W:=q_U(K_U)$ has dimension \(3\) and the radical of the restricted form on \(W\) is $R:=q_U(K_U\cap K_U^\perp),$ which is one-dimensional. Choose any two-dimensional subspace \(H\) in \(W\) such that
$W=H\oplus R. $
The symplectic form is non-degenerate on $H$. Indeed, suppose that \(h\in H\) satisfies \([h,h']=0\) for every \(h'\in H\). Since \(R=\operatorname{rad}(W)\), we also have \([h,r]=0\) for every \(r\in R\). As \(W=H\oplus R\), it follows that \(h\in\operatorname{rad}(W)=R\). But \(h\in H\) and \(H\cap R=0\), so \(h=0\). By \cref{lem:kernel-cocycle-consequences}, \(E_U\) fixes \(W\) pointwise. Hence, \(E_U\) fixes both \(H\) and \(R\) pointwise. Since \(H\) is non-degenerate, the symplectic complement argument used in the codimension two case of \cref{codim2-gen} gives $\mathbb F^6_p=H\oplus H^\perp$ and that \(E_U\) preserves \(H^\perp\). Therefore, we have a homomorphism
\begin{align}
E_U &\longrightarrow \operatorname{Sp}(H^\perp),\\
T &\longmapsto T|_{H^\perp}.
\end{align}
As in \cref{lem:kernel-cocycle-consequences}, the map is injective. Thus, $E_U$ embeds into $\operatorname{Sp}(H^\perp) \cong \operatorname{Sp}(\mathbb F_p^4)$ as an abelian $p$-subgroup of order \(p^3\).

\item Next, as in the two-qudit codimension three proof, we use \cref{wong} to identify the restricted image on \(H^\perp\) with a subgroup fixing a Lagrangian plane \(L_0\) pointwise. Since \(E_U\) is a \(p\)-subgroup of \(\operatorname{Sp}(H^\perp)\), choose a Sylow \(p\)-subgroup \(P\) containing \(E_U\). Conjugating inside \(\operatorname{Sp}(H^\perp)\), and extending this conjugation to \(\mathbb F_p^6=H\oplus H^\perp\) by the identity on \(H\), we may assume that \(P\) is the standard Sylow \(p\)-subgroup used in \cref{codim4-two-qubit-rule}. After this change of coordinates, the complete isotropic flag
\begin{equation}
0\subseteq \langle e_1\rangle\subseteq \langle e_1,e_2\rangle=L_0
\end{equation}
is contained in \(H^\perp\). Choose the complementary Lagrangian \(L_0^*\) in these coordinates, so that \(H^\perp=L_0\oplus L_0^*\). By \cref{wong}, and exactly as in the proof of \cref{codim3-two-qubit-rule}, the unique abelian subgroup of largest order in \(P\) is \(A(L_0)\), the subgroup fixing \(L_0\) pointwise. In the coordinates given by the decomposition \(H^\perp=L_0\oplus L_0^*\), this subgroup is
\begin{equation}
A(L_0)
=
\left\{
\begin{pmatrix}
I&C\\
0&I
\end{pmatrix}
:
C=C^T\in M_2(\mathbb F_p)
\right\}.
\end{equation}
In particular, \(|A(L_0)|=p^3\). Since \(E_U\subseteq P\) is abelian and \(|E_U|=p^3\), we have $E_U=A(L_0)$. It is clear that 
\begin{equation}\label{L0-fix}
L_0 = \{v\in H^\perp : Tv=v \text{ for every } T\in E_U\}.
\end{equation}

\item We now enlarge the Lagrangian plane \(L_0\subseteq H^\perp\) to a full Lagrangian \(L\subseteq \mathbb F_p^6\). Since \(R\subseteq H^\perp\) and \(R\) is fixed pointwise by \(E_U\), we have \(R\subseteq L_0\). Choose a line \(\ell\subseteq H\) and set $L:=\ell\oplus L_0$.  
Then \(\dim L=1+2=3\). Moreover, \(L\) is isotropic. Indeed, if \(x=x_\ell+x_0\) and \(y=y_\ell+y_0\), with \(x_\ell,y_\ell\in \ell\) and \(x_0,y_0\in L_0\), then
\begin{equation}
[x,y]=[x_\ell,y_\ell]+[x_\ell,y_0]+[x_0,y_\ell]+[x_0,y_0].
\end{equation}
The first term is zero because \(\ell\) is isotropic in \(H\), the last term is zero because \(L_0\) is Lagrangian in \(H^\perp\), and the two middle terms are zero because \(H\perp H^\perp\). Hence, \([x,y]=0\) for all \(x,y\in L\). Since \(\dim L=3=\frac{1}{2}\dim \mathbb F^6_p\), \(L\) is Lagrangian. Choose a complementary Lagrangian \(L^*\) such that $\mathbb F^6_p=L\oplus L^*$. Note that \(E_U\) fixes \(\ell\) pointwise because it fixes \(H\) pointwise, and \(E_U\) fixes \(L_0\) pointwise because \(E_U=A(L_0)\). Hence, $E_U$ fixes $L$ pointwise. 

\item We now express the symplectic data \(S_U\) as a linear map into \(\operatorname{Sym}_3(\mathbb F_p)\). \cref{lem:pointwise-lagrangian-coordinate-setup} implies for every \(a\in\mathbb F^6_p\) we may write
\begin{equation}\label{eqn:three-qudit-c3-upper-unipotent-form}
S_U(a)=u_{C(a)}=\begin{pmatrix}I&C(a)\\0&I\end{pmatrix},
\qquad C(a)\in\operatorname{Sym}_3(\mathbb F_p),
\end{equation}
where \(C:\mathbb F^6_p\to\operatorname{Sym}_3(\mathbb F_p)\) is linear. Let $B:=C(\mathbb F^6_p).$ Since \(C(a)=0\) if and only if \(S_U(a)=I\), we have \(\ker C=K_U\). Hence, 
$\dim B =\dim \mathbb F_p^6-\dim K_U=6-3=3.$
Note that, unlike in the two-qudit case, \(C\) is not surjective.

\item   We compute \(B\) explicitly before using the cocycle identity to derive a contradiction. Using \cref{L0-fix} and the fact that \(E_U=A(L_0)\), the subspace of \(\mathbb F_p^6\) fixed by \(E_U\) is \(H\oplus L_0\), which has dimension \(4\). On the other hand, in the coordinates \(\mathbb F^6_p=L\oplus L^*\), the same fixed subspace is $L\oplus N(B),$ where
\begin{equation}\label{eqn:three-qudit-c3-NB}
N(B):=\bigcap_{a\in\mathbb F_p^6} \ker C(a) \subseteq L^*.
\end{equation}
Indeed, \(u_{C(a)}(x,y)=(x+C(a)y,y)\), so \((x,y)\) is fixed by every \(u_{C(a)}\) if and only if \(C(a)y=0\) for every \(a \in \mathbb F_p^6\). Therefore, \(\dim N(B)=1\). Choose \(0\neq r_0\in N(B)\). Since \(r_0\in N(B)\), we have
\begin{equation}
B\subseteq \{A\in\operatorname{Sym}_3(\mathbb F_p):Ar_0=0\}.
\end{equation}
We claim that the space on the right has dimension \(3\). Indeed, choose a basis of \(L^*\) whose third basis vector is \(r_0\). In this basis, the condition \(Ar_0=0\) says that the third column of \(A\) is zero. Since \(A\) is symmetric, the third row is also zero. Therefore every such \(A\) has the form
\begin{equation}
A=
\begin{pmatrix}
a_{11}&a_{12}&0\\
a_{12}&a_{22}&0\\
0&0&0
\end{pmatrix}.
\end{equation}
Thus, the space \(\{A\in\operatorname{Sym}_3(\mathbb F_p):Ar_0=0\}\) has the three free parameters \(a_{11},a_{12},a_{22}\), and hence has dimension \(3\). Since \(\dim B=3\), we have
\begin{equation}\label{eqn:three-qudit-c3-B-space}
B=\{A\in\operatorname{Sym}_3(\mathbb F_p):Ar_0=0\}.
\end{equation}

\item We now use the cocycle identity to constrain the \(L^*\) component of \(q_U\). In contrast with the two-qudit proof, this will force \(\beta\) into a line rather than forcing \(\beta=0\). Write
\begin{equation}\label{eqn:three-qudit-c3-q-alpha-beta}
q_U(a)=(\alpha(a),\beta(a))\in L\oplus L^*.
\end{equation}
By \cref{lem:pointwise-lagrangian-coordinate-setup}, the map \(\beta:\mathbb F^6_p\to L^*\) is linear and
\begin{equation}\label{eqn:three-qudit-c3-C-beta-symmetry}
C(a)\beta(b)=C(b)\beta(a)
\end{equation}
for all \(a,b\in\mathbb F^6_p\). If \(k\in K_U\), then \(C(k)=0\), so \(C(a)\beta(k)=0\) for every \(a\in\mathbb F_p^6\). Hence, \(\beta(k)\in N(B)\). Let \(\pi:L^*\to L^*/N(B)\) be the quotient map. It follows that $K_U=\ker C\subseteq\ker(\pi\circ\beta)$. Therefore, there is a unique linear map
\begin{align}
\widetilde T:\mathbb F_p^6/K_U &\longrightarrow L^*/N(B),\\
a+K_U &\longmapsto \beta(a)+N(B).
\end{align}
Since \(C:\mathbb F_p^6\to B\) is surjective with kernel \(K_U\), we have an isomorphism
\begin{align}
\overline C:\mathbb F_p^6/K_U &\xrightarrow{\sim} B,\\
 a+K_U & \longmapsto C(a).
\end{align}
Therefore, we can define
\begin{equation}
T:=\widetilde T\circ\overline C^{-1}:B\longrightarrow L^*/N(B).
\end{equation}
By construction, $T(C(a))=\beta(a)+N(B)$ for every \(a\in\mathbb F_p^6\). The entire factorization is summarized by the diagram in \cref{factoriztion}.

\begin{equation}\label{factoriztion}
\begin{tikzcd}[column sep=huge,row sep=large]
\mathbb F_p^6 \arrow[r,"\beta"] \arrow[d,two heads,"q"'] & L^* \arrow[r,"\pi"] & L^*/N(B) \\
\mathbb F_p^6/K_U \arrow[r,leftrightarrow,"\overline C"] \arrow[rru,"\widetilde T"',swap, dashed] & B \arrow[ur,"T"'] &
\end{tikzcd}
\end{equation}

\item We show that \(T\) is in fact the zero map. Choose coordinates on \(L^*\) with \(r_0=e_3\), together with the corresponding coordinates on \(L\). Then \cref{eqn:three-qudit-c3-B-space} identifies \(B\) with \(\operatorname{Sym}_2(\mathbb F_p)\), while \(L^*/N(B)\) is identified with \(\mathbb F_p^2\). Since \(N(B)\subseteq\ker A\) for every \(A\in B\), each linear map \(A:L^*\to L\) factors uniquely through the quotient \(L^*/N(B)\). Under the chosen coordinates, the resulting map is represented by the corresponding \(2\times2\) symmetric matrix \(A\). Thus, for \(x\in L^*\), its value on the coset \(x+N(B)\) is represented by \(Ax\).
Let \(A,D\in\operatorname{Sym}_2(\mathbb F_p)\), and choose \(a,b\in\mathbb F_p^6\) such that \(A=C(a)\) and \(D=C(b)\). By the definition of \(T\),
\begin{equation}
T(A)=\beta(a)+N(B),\qquad T(D)=\beta(b)+N(B).
\end{equation}
Therefore, \cref{eqn:three-qudit-c3-C-beta-symmetry} gives
\begin{equation}\label{eqn:three-qudit-c3-ATD-DTA}
AT(D)=DT(A)
\end{equation}
for all \(A,D\in\operatorname{Sym}_2(\mathbb F_p)\). Taking \(D=I\) gives \(AT(I)=T(A)\). Setting \(t:=T(I)\), we therefore have \(T(A)=At\) for every \(A\in\operatorname{Sym}_2(\mathbb F_p)\). Substituting this into \cref{eqn:three-qudit-c3-ATD-DTA} gives
\begin{equation}\label{eqn:three-qudit-c3-commutator-kills-t}
(AD-DA)t=0
\end{equation}
for all symmetric \(A,D\). Taking
\begin{equation}
A=\begin{pmatrix}1&0\\0&0\end{pmatrix},
\qquad
D=\begin{pmatrix}0&1\\1&0\end{pmatrix}
\end{equation}
gives
\begin{equation}
AD-DA=\begin{pmatrix}0&1\\-1&0\end{pmatrix},
\end{equation}
which is invertible over \(\mathbb F_p\). Thus \(t=0\), and hence \(T=0\). Since \(T(C(a))=\beta(a)+N(B)\), it follows that
\begin{equation}
\beta(\mathbb F_p^6)\subseteq N(B)=\mathbb F_pr_0.
\end{equation}
Write \(\beta(a)=\lambda(a)r_0\) for a linear functional \(\lambda:\mathbb F_p^6\to\mathbb F_p\), and define \(\nu(a):=r_0^T\alpha(a)\).

\item Finally, as in the two-qudit codimension three proof, we compare the commutator of \(B_U(a)\) and \(B_U(b)\) in two ways. Since \(C(a)r_0=0\) for every \(a\in\mathbb F^6_p\), each \(u_{C(a)}\) fixes \(q_U(b)\) for every \(b\in\mathbb F^6_p\). Therefore, \(M_{S_U(a)}\) commutes with \(w(q_U(b))\) for all \(a,b\). By \cref{lem:upper-unipotent-commuting-lifts}, applied with \(m=3\), the fixed Clifford representatives \(M_{S_U(a)}=M_{u_{C(a)}}\) commute pairwise. Therefore, we have
\begin{equation}\label{hi-1}
B_U(a)B_U(b)=\omega^{[q_U(a),q_U(b)]}B_U(b)B_U(a).
\end{equation} 
On the other hand, since \(B_U(a)=Uw(a)U^*\), we also have
\begin{equation}\label{hi-2}
B_U(a)B_U(b)=\omega^{[a,b]}B_U(b)B_U(a).
\end{equation}
Comparing \cref{hi-1} and \cref{hi-2} gives $[a,b]=[q_U(a),q_U(b)]$
for all \(a,b\in\mathbb F^6_p\). Since \(q_U(a)=(\alpha(a),\lambda(a)r_0)\), we obtain
\begin{equation}\label{eqn:three-qudit-c3-rank-two-form}
[a,b]=[q_U(a),q_U(b)]=\lambda(a)\nu(b)-\lambda(b)\nu(a).
\end{equation}
It remains to note that \(\nu\) is linear. By \cref{lem:pointwise-lagrangian-coordinate-setup}, we have
\begin{equation}
\alpha(a+b)=\alpha(a)+\alpha(b)+C(a)\beta(b).
\end{equation}
Since \(\beta(b)=\lambda(b)r_0\) and \(C(a)r_0=0\), this becomes $\alpha(a+b)=\alpha(a)+\alpha(b).$
Thus, \(\nu(a+b)=\nu(a)+\nu(b)\). Therefore, the right side of \cref{eqn:three-qudit-c3-rank-two-form} is an alternating form of rank at most \(2\), contradicting the fact that the symplectic form on \(\mathbb F^6_p\) has rank \(6\). Hence, the case \(r_K=1\) is impossible.
\end{enumerate}
This completes the proof.
\end{proof}

The remaining cases satisfy $\dim E_U=\operatorname{codim}K_U \geq 4$.  Before treating them separately, we establish in \cref{lem:three-qudit-large-image-obstruction} that such a large image cannot fix a Lagrangian subspace pointwise. The proof follows the same general argument as \cref{lem:three-qudit-codim-three}, but the higher rank assumption forces \(N(B)\) to vanish, allowing \(\beta\) to factor directly through \(C\). At the same time, however, \(B\) is now only an arbitrary higher dimensional subspace of \(\operatorname{Sym}_3(\mathbb F_p)\), rather than the rigid copy of \(\operatorname{Sym}_2(\mathbb F_p)\) that arises after quotienting in the codimension \(3\) case. Consequently, the proof requires a new dimension bound for the induced map \(T\) before obtaining a contradiction with the non-degeneracy of the symplectic form.

\begin{lem}\label{lem:three-qudit-large-image-obstruction}
Let \(p\) be an odd prime. Let \(U\in\mathcal C^{(3)}\) be a three-qudit gate, and set \(E_U:=\operatorname{im}S_U\). If \(\dim E_U\geq4\), then \(E_U\) does not fix any Lagrangian subspace of \(\mathbb F_p^6\) pointwise.
\end{lem}

\begin{proof}
Set \(d:=\dim E_U\geq4\). Suppose, for contradiction, that \(E_U\) fixes a Lagrangian subspace \(L\subseteq\mathbb F_p^6\) pointwise. Choose a complementary Lagrangian \(L^*\). By \cref{lem:pointwise-lagrangian-coordinate-setup}, we may write
\begin{equation}\label{eqn:large-image-upper-unipotent-form}
S_U(a)=u_{C(a)}=\begin{pmatrix}I&C(a)\\0&I\end{pmatrix},
\end{equation}
where \(C:\mathbb F_p^6\to\operatorname{Sym}_3(\mathbb F_p)\) is linear. By \cref{lem:pointwise-lagrangian-coordinate-setup}, if \(q_U(a)=(\alpha(a),\beta(a))\in L\oplus L^*\) then \(\beta:\mathbb F_p^6\to L^*\) is linear and satisfies
\begin{equation}\label{eqn:large-image-C-beta-symmetry}
C(a)\beta(b)=C(b)\beta(a)
\end{equation}
for all \(a,b\in\mathbb F_p^6\). As in the proof of \cref{lem:three-qudit-codim-three}, let \(B:=C(\mathbb F_p^6)\). Since \(\ker C=K_U\), we have
\begin{equation}
\dim B=\operatorname{codim}K_U=\dim E_U=d.
\end{equation}
Define $N(B)$ as in \cref{eqn:three-qudit-c3-NB}. If \(N(B)\neq0\), choose \(0\neq r\in N(B)\). Then
\begin{equation}
B\subseteq\{A\in\operatorname{Sym}_3(\mathbb F_p):Ar=0\}.
\end{equation}
The dimension calculation in the proof of \cref{lem:three-qudit-codim-three} shows that the space on the right has dimension \(3\). This contradicts \(\dim B=d\geq4\). Hence,
$N(B)=0$. If \(k\in K_U\), then \(C(k)=0\), so \cref{eqn:large-image-C-beta-symmetry} gives \(C(a)\beta(k)=0\) for every \(a\in\mathbb F_p^6\). Hence \(\beta(k)\in N(B)\), and therefore \(\beta(k)=0\). Thus, \(K_U=\ker C\subseteq\ker\beta\). Since \(C:\mathbb F_p^6\to B\) is surjective with kernel \(K_U\), the same factorization argument used in \cref{lem:three-qudit-codim-three}, now with \(L^*/N(B)=L^*\), gives a unique linear map
\begin{equation}\label{eqn:large-image-T-map}
T:B\longrightarrow L^*
\end{equation}
such that
$T(C(a))=\beta(a)$
for every \(a\in\mathbb F_p^6\). Let \(A,D\in B\), and choose \(a,b\in\mathbb F_p^6\) such that \(A=C(a)\) and \(D=C(b)\). By \cref{eqn:large-image-C-beta-symmetry}, we have
\begin{equation}\label{eqn:large-image-ATD-DTA}
AT(D)=DT(A)
\end{equation}
for all \(A,D\in B\). Set
$s:=\dim\operatorname{im}T.$
If \(A\in\ker T\), then \cref{eqn:large-image-ATD-DTA} gives
$AT(D)=0$
for every \(D\in B\). Since \(T(B)=\operatorname{im}T\), the matrix \(A\) therefore annihilates \(\operatorname{im}T\). A symmetric \(3\times3\) matrix annihilating a fixed \(s\)-dimensional subspace is determined by a symmetric \((3-s)\times(3-s)\) block. Consequently,
\begin{equation}\label{eqn:large-image-dimension-inequality}
d-s\leq\frac{(3-s)(4-s)}{2}.
\end{equation}
Suppose first that \(d=5\) or \(d=6\). Then \cref{eqn:large-image-dimension-inequality} forces \(s=0\). Hence \(T=0\), so \(\beta=0\) and
$q_U(\mathbb F_p^6)\subseteq L.$
The same commutator argument used at the end of \cref{codim3-two-qubit-rule}, together with \cref{lem:upper-unipotent-commuting-lifts} applied with \(m=3\), implies that the operators \(B_U(a)\) commute pairwise. On the other hand,
\begin{equation}
B_U(a)B_U(b)=\omega^{[a,b]}B_U(b)B_U(a)
\end{equation}
for all \(a,b\in\mathbb F_p^6\). This contradicts the non-degeneracy of the symplectic form. Hence \(d\neq5,6\). It remains to consider \(d=4\). By \cref{eqn:large-image-dimension-inequality}, either \(s=0\) or \(s=1\). If \(s=0\), then \(T=0\), hence \(\beta=0\), and the same commuting-lift argument used above gives a contradiction. Therefore,  \(s=1\). Thus,
$T(A)=\lambda(A)r$
for some non-zero \(r\in L^*\) and non-zero linear functional \(\lambda:B\to\mathbb F_p\). Choose \(D_0\in B\) such that \(\lambda(D_0)=1\). Taking \(D=D_0\) in \cref{eqn:large-image-ATD-DTA} gives
\begin{equation}\label{eqn:large-image-Ar-lambda}
Ar=\lambda(A)D_0r
\end{equation}
for every \(A\in B\). Set \(v:=D_0r\) and \(\xi:=r^Tv\). Then
\begin{equation}\label{eqn:large-image-rAr}
r^TAr=\xi\lambda(A)
\end{equation}
for every \(A\in B\). For \(a\in\mathbb F_p^6\), write
\begin{equation}
A_a:=C(a),\qquad \lambda_a:=\lambda(A_a),\qquad \nu_a:=r^T\alpha(a).
\end{equation}
Then
$\beta(a)=T(A_a)=\lambda_ar.$
Using the commuting lifts from \cref{lem:upper-unipotent-commuting-lifts} and comparing the commutators of \(B_U(a)\) and \(B_U(b)\), we obtain
\begin{equation}\label{eqn:large-image-commutator-formula}
[a,b]=[q_U(a),q_U(b)]+\frac12\left(\beta(a)^TA_b\beta(a)-\beta(b)^TA_a\beta(b)\right).
\end{equation}
Substituting \(\beta(a)=\lambda_ar\) and using \cref{eqn:large-image-rAr} gives
\begin{equation}\label{eqn:large-image-pre-eta}
[a,b]=\lambda_a\nu_b-\lambda_b\nu_a+\frac{\xi}{2}\left(\lambda_a^2\lambda_b-\lambda_b^2\lambda_a\right).
\end{equation}
Define
$\eta_a:=\nu_a-\frac{\xi}{2}\lambda_a^2.$
By \cref{lem:pointwise-lagrangian-coordinate-setup}, we have
\begin{equation}
\alpha(a+b)=\alpha(a)+\alpha(b)+A_a\beta(b).
\end{equation}
Using \(\beta(b)=\lambda_br\) and \cref{eqn:large-image-Ar-lambda}, we obtain
\begin{equation}
\nu_{a+b}=\nu_a+\nu_b+\xi\lambda_a\lambda_b.
\end{equation}
Since \(\lambda_{a+b}=\lambda_a+\lambda_b\), it follows that
\begin{align}
\eta_{a+b}
&=\nu_{a+b}-\frac{\xi}{2}\lambda_{a+b}^2=\nu_a+\nu_b+\xi\lambda_a\lambda_b-\frac{\xi}{2}(\lambda_a+\lambda_b)^2=\eta_a+\eta_b.
\end{align}
Thus \(\eta:\mathbb F_p^6\to\mathbb F_p\) is linear. \cref{eqn:large-image-pre-eta} therefore becomes
\begin{equation}\label{eqn:large-image-rank-two-form}
[a,b]=\lambda_a\eta_b-\lambda_b\eta_a.
\end{equation}
As in the proof of \cref{lem:three-qudit-codim-three}, the alternating form on the right has rank at most \(2\), whereas the symplectic form on \(\mathbb F_p^6\) has rank \(6\). This contradiction rules out \(d=4\), and completes the proof.
\end{proof}

\cref{lem:three-qudit-codim-four-impossible} rules out the codimension \(4\) case. Since \(\dim K_U=2\), the restriction of the symplectic form to \(K_U\) is either non-degenerate or identically zero. The non-degenerate case reduces immediately to the non-existence of an abelian \(p\)-subgroup of order \(p^4\) in \(\operatorname{Sp}(\mathbb F_p^4)\).
In the isotropic case, \(E_U\) fixes the isotropic plane \(W=q_U(K_U)\) pointwise. We therefore analyze \(E_U\) using a homomorphism induced by its action on a complementary symplectic plane, and apply \cref{lem:three-qudit-large-image-obstruction} whenever the resulting structure forces \(E_U\) to fix a Lagrangian subspace pointwise.

\begin{prop}\label{lem:three-qudit-codim-four-impossible}
Let \(p\) be an odd prime and let \(U\in\mathcal C^{(3)}\) be a three-qudit gate. Set \(K_U:=\ker S_U\) and \(E_U:=\operatorname{im}S_U\). Then \(\operatorname{codim}K_U\neq4\).
\end{prop}

\begin{proof}
Suppose, for contradiction, that \(\operatorname{codim}K_U=4\). Then $\dim K_U=2$ and  $\dim E_U=4.$
Since an alternating form is skew-symmetric and has even rank, the restriction of the symplectic form to \(K_U\) is either non-degenerate or identically zero. Suppose first that \(K_U\) is non-degenerate. By \cref{lem:kernel-cocycle-consequences}, \(W:=q_U(K_U)\) is a non-degenerate symplectic plane fixed pointwise by \(E_U\). Hence, \(\mathbb F_p^6=W\oplus W^\perp\), \(E_U\) acts on \(W^\perp\) and we obtain an injective homomorphism
\begin{equation}
E_U\lhook\joinrel\longrightarrow\operatorname{Sp}(W^\perp)\cong\operatorname{Sp}(\mathbb F_p^4).
\end{equation}
Thus, \(\operatorname{Sp}(\mathbb F_p^4)\) would contain an abelian \(p\)-subgroup of order \(p^4\).
This is impossible by \cref{codim4-two-qubit-rule}. Therefore, the symplectic form on \(K_U\) is identically zero. That is, \(K_U\) is isotropic. We rule out this case to derive a contradiction. We divide the proof into several steps to make the otherwise lengthy argument easier to follow.

\begin{enumerate} 
\item We first construct a decomposition of \(\mathbb F_p^6\) adapted to the isotropic plane \(W:=q_U(K_U)\)\footnote{This differs from the non-degenerate case where \(W\cap W^\perp=0\), and from the one-dimensional-radical case in the proof of \cref{lem:three-qudit-codim-three}, where a non-degenerate symplectic plane \(H\subseteq W\) could be split off. Here \(W\) is totally isotropic, so \(W\subseteq W^\perp\) and \(W\) contains no non-degenerate symplectic plane. We obtain the required symplectic plane from \(W^\perp/W\) instead.}. By \cref{lem:kernel-cocycle-consequences}, \(W\) is an isotropic plane fixed pointwise by \(E_U\). Since \(\dim W=2\), we have \(\dim W^\perp=4\). The radical of the restricted symplectic form on \(W^\perp\) is
\begin{equation}
\operatorname{rad}(W^\perp)=W^\perp\cap(W^\perp)^\perp=W^\perp\cap W=W,
\end{equation}
where we have used that \(W\subseteq W^\perp\). Thus, the quotient \(W^\perp/W\) is a two-dimensional non-degenerate symplectic space. We first choose a symplectic plane representing this quotient. Choose a basis \(w_1,w_2\) of \(W\) and extend it to a basis \(w_1,w_2,h_1,h_2\) of \(W^\perp\). Setting \(H:=\operatorname{span}\{h_1,h_2\}\), we have
\begin{equation}
W^\perp=W\oplus H,\qquad W^\perp/W\cong H.
\end{equation}
The restricted symplectic form on \(H\) is non-degenerate. Indeed, if \(h\in H\) and \([h,h']=0\) for all \(h'\in H\), then since \(h\in W^\perp\) we also have \([h,w]=0\) for all \(w\in W\). As \(W^\perp=W\oplus H\), it follows that \(h\in\operatorname{rad}(W^\perp)=W\). Since \(H\cap W=0\), we obtain \(h=0\). Since \(H\) is non-degenerate, we have
$\mathbb F_p^6=H\oplus H^\perp.$
Moreover, \(W\subseteq H^\perp\) because \(H\subseteq W^\perp\), and \(\dim H^\perp=4\). Since \(W\) is isotropic and has dimension \(2\), it is therefore Lagrangian in the four-dimensional symplectic space \(H^\perp\). Hence, there exists a Lagrangian subspace \(W^*\subseteq H^\perp\) such that
$H^\perp=W\oplus W^*,$
and the symplectic pairing identifies \(W^*\) with the dual space of \(W\). Consequently,
\begin{equation}\label{eqn:three-qudit-c4-witt}
\mathbb F_p^6=W\oplus H\oplus W^*.
\end{equation}

\item We next choose a basis adapted to the decomposition in \cref{eqn:three-qudit-c4-witt} and use it to determine explicitly the form of every symplectic transformation fixing \(W\) pointwise. Since the pairing between \(W\) and \(W^*\) is non-degenerate, we may choose bases \(w_1,w_2\) of \(W\) and \(w_1^*,w_2^*\) of \(W^*\) such that
\begin{equation}
[w_i,w_j^*]=\delta_{ij},\qquad [w_i,w_j]=[w_i^*,w_j^*]=0.
\end{equation}
Since \(H\) is a two-dimensional non-degenerate symplectic space, we may also choose a symplectic basis \(h_1,h_2\) satisfying
\begin{equation}
[h_1,h_2]=1,\qquad [h_2,h_1]=-1,\qquad [h_1,h_1]=[h_2,h_2]=0.
\end{equation}
Thus, relative to the ordered basis $w_1,w_2,h_1,h_2,w_1^*,w_2^*$
the symplectic form is represented by
\begin{equation}\label{eqn:three-qudit-c4-symplectic-matrix}
\Omega:=\begin{pmatrix}0&0&I\\0&J&0\\-I&0&0\end{pmatrix}, \qquad J=\begin{pmatrix}0&1\\-1&0\end{pmatrix}.
\end{equation}
We now determine the form of a symplectic transformation fixing \(W\) pointwise. First, such a transformation preserves \(W^\perp=W\oplus H\). Indeed, if \(x\in W^\perp\) and \(w\in W\), then
\begin{equation}
[Tx,w]=[Tx,Tw]=[x,w]=0,
\end{equation}
so \(Tx\in W^\perp\). Moreover, \(T\) induces the identity on \(\mathbb F_p^6/W^\perp\). Indeed, for every \(x\in\mathbb F_p^6\) and \(w\in W\),
\begin{equation}
[Tx-x,w]=[Tx,Tw]-[x,w]=0,
\end{equation}
and hence \(Tx-x\in W^\perp\). Therefore, relative to the decomposition in \cref{eqn:three-qudit-c4-witt}, fixing \(W\) pointwise gives the first block column \((I,0,0)^T\), preservation of \(W^\perp=W\oplus H\) forces the lower-middle block to vanish, and triviality on the quotient forces the lower-right block to be \(I\). Thus, \(T\) has the block form
\begin{equation}
T=\begin{pmatrix}I&A&B\\0&g&C\\0&0&I\end{pmatrix}
\end{equation}
for some \(2\times2\) matrices \(A,B,C,g\). Imposing the symplectic condition
$T^T\Omega T=\Omega,$
we compute
\begin{align}
T^T\Omega T
&=
\begin{pmatrix}
I&0&0\\
A^T&g^T&0\\
B^T&C^T&I
\end{pmatrix}
\begin{pmatrix}
0&0&I\\
0&J&0\\
-I&0&0
\end{pmatrix}
\begin{pmatrix}
I&A&B\\
0&g&C\\
0&0&I
\end{pmatrix}\\
&=
\begin{pmatrix}
I&0&0\\
A^T&g^T&0\\
B^T&C^T&I
\end{pmatrix}
\begin{pmatrix}
0&0&I\\
0&Jg&JC\\
-I&-A&-B
\end{pmatrix}=
\begin{pmatrix}
0&0&I\\
0&g^TJg&A^T+g^TJC\\
-I&C^TJg-A&B^T+C^TJC-B
\end{pmatrix}.
\end{align}
Comparing this with $\Omega$ gives
\begin{equation}
g^TJg=J,\qquad A^T+g^TJC=0,\qquad B-B^T=C^TJC.
\end{equation}
The first identity says that \(g\in\operatorname{Sp}(H)\cong\operatorname{Sp}(\mathbb F^2_p)\). Using \(g^TJg=J\), the second identity gives
$C=gJA^T.$
Substituting this into the third identity gives
$B-B^T=AJA^T.$
Since \(p\) is odd, this is equivalent to writing
$B=Z+\frac12AJA^T$
for a unique symmetric matrix \(Z=Z^T\). Consequently, every symplectic transformation fixing \(W\) pointwise has a unique expression
\begin{equation}\label{eqn:three-qudit-c4-stabilizer}
T(g,A,Z):=\begin{pmatrix}I&A&Z+\frac12AJA^T\\0&g&gJA^T\\0&0&I\end{pmatrix},
\end{equation}
where $g\in\operatorname{Sp}(H)\cong\operatorname{Sp}(\mathbb F^2_p)$, $A\in M_2(\mathbb F_p)$ and $ Z\in\operatorname{Sym}_2(\mathbb F_p).$ Finally, direct block multiplication gives
\begin{equation}\label{eqn:three-qudit-c4-product}
T(g_1,A_1,Z_1)T(g_2,A_2,Z_2)=T\left(g_1g_2,A_2+A_1g_2,Z_1+Z_2+\operatorname{sym}(A_1g_2JA_2^T)\right),
\end{equation}
where $\operatorname{sym}(X):=\frac12(X+X^T).$

\item We next isolate a linear component of \(q_U\) that will provide the contradiction in both remaining cases. Write
\begin{equation}\label{eqn:three-qudit-c4-q}
q_U(a)=(\alpha(a),\gamma(a),\beta(a))\in W\oplus H\oplus W^*.
\end{equation}
Since every matrix in \cref{eqn:three-qudit-c4-stabilizer} acts as the identity on the \(W^*\) component, taking \(W^*\) components in the cocycle identity gives $\beta(a+b)=\beta(a)+\beta(b)$. Thus, $\beta:\mathbb F_p^6\to W^*$ is linear. We claim that
$\operatorname{rank}\beta=2.$
Let \(k\in K_U\) and \(a\in\mathbb F_p^6\). By \cref{lem:kernel-cocycle-consequences}, we have
$[q_U(k),q_U(a)]=[k,a].$
Since \(q_U(k)\in W\), while \(W\) is orthogonal to both \(W\) and \(H\), only the \(W^*\) component of \(q_U(a)\) contributes to the left hand side. Hence, we have
\begin{equation}\label{eqn:three-qudit-c4-beta-pairing}
[q_U(k),\beta(a)]=[k,a].
\end{equation}
\(q_U|_{K_U}:K_U\to W\) is an isomorphism. Since \(\beta(a)\in W^*\), it defines the functional on \(W\) given by \(w\mapsto[w,\beta(a)]\), and pulling this functional back along \(q_U|_{K_U}\) gives \(k\mapsto[q_U(k),\beta(a)]\). Hence, through \cref{eqn:three-qudit-c4-beta-pairing}, each \(\beta(a)\in W^*\) represents the linear functional on \(K_U\)
$k\longmapsto[k,a].$
On the other hand, non-degeneracy of the ambient symplectic form implies that
\begin{align}
\mathbb{F}_p^6 &\longrightarrow K_U^*,\\
a &\longmapsto \left(k\longmapsto[k,a]\right).
\end{align}
is surjective. 
If \(f\in K_U^*\), then there exists \(a\in\mathbb F_p^6\) such that
$f(k)=[k,a]$
for every \(k\in K_U\). By \cref{eqn:three-qudit-c4-beta-pairing},
\begin{equation}
f(k)=[k,a]=[q_U(k),\beta(a)].
\end{equation}
Thus, under the identification \(W^*\cong K_U^*\), the vector \(\beta(a)\) corresponds exactly to \(f\). Hence, \(\beta:\mathbb F_p^6\to W^*\) is surjective. As \(\dim W^*=2\), this proves $\operatorname{rank}\beta=2.$

\item We now record the action of \(E_U\) on the symplectic plane \(H\), which reduces the remaining argument to two possible cases. Taking the middle diagonal block in \cref{eqn:three-qudit-c4-stabilizer} defines a homomorphism
\begin{align}
\rho:E_U&\longrightarrow\operatorname{Sp}(\mathbb F^2_p),\\
T(g,A,Z)&\longmapsto g.
\end{align}
That \(\rho\) is a homomorphism follows from \cref{eqn:three-qudit-c4-product}. Since
$|\operatorname{Sp}(\mathbb F^2_p)|=p(p^2-1),$
every \(p\)-subgroup of \(\operatorname{Sp}(\mathbb F^2_p)\) has order at most \(p\). Hence, 
$\dim \rho(E_U)\in\{0,1\}.$
We rule out the possibility that $\dim \rho(E_U)=0$. The $\dim \rho(E_U)=1$ case is proved in \cref{appendix:rank-one-codim4-three}.  Let \(\dim \rho(E_U)=0\). In this case, all $g$'s are equal to the identity. We determine the possible \(A\)'s in \cref{eqn:three-qudit-c4-stabilizer} and then use the cocycle to contradict $\operatorname{rank}(\beta)=2$.  We proceed in several steps.

\begin{enumerate}[label=(\alph*)]

\item  Define
\begin{equation}
\mathcal A:=\{A\in M_2(\mathbb F_p):T(I,A,Z)\in E_U\text{ for some }Z\}.
\end{equation}
If \(A_1,A_2\in\mathcal A\), then there exist \(Z_1,Z_2\) such that \(T(I,A_1,Z_1),T(I,A_2,Z_2)\in E_U\), and we have
\begin{equation}
T(I,A_1,Z_1)T(I,A_2,Z_2)=T\left(I,A_1+A_2,Z_1+Z_2+\operatorname{sym}(A_1JA_2^T)\right)\in E_U.
\end{equation}
by \cref{eqn:three-qudit-c4-product}. Thus, \(A_1+A_2\in\mathcal A\). Clearly, \(0\in\mathcal A\). If \(A\in\mathcal A\), then for some \(Z\) we have \(T(I,A,Z)\in E_U\). By \cref{eqn:three-qudit-c4-product}, we have
\begin{equation}
T(I,A,Z)T(I,-A,-Z)=T\left(I,0,\operatorname{sym}(-AJA^T)\right)=T(I,0,0),
\end{equation}
where the final equality follows because \(AJA^T\) is skew-symmetric. Hence
\begin{equation}
T(I,A,Z)^{-1}=T(I,-A,-Z)\in E_U,
\end{equation}
so \(-A\in\mathcal A\). Therefore, \(\mathcal A\) is a linear subspace of \(M_2(\mathbb F_p)\). Let \(T(I,A,Z),T(I,D,Y)\in E_U\). By \cref{eqn:three-qudit-c4-product},
\begin{align}
T(I,A,Z)T(I,D,Y)
&=
T\left(I,A+D,Z+Y+\operatorname{sym}(AJD^T)\right),\\
T(I,D,Y)T(I,A,Z)
&=
T\left(I,A+D,Z+Y+\operatorname{sym}(DJA^T)\right).
\end{align}
Since \(E_U\) is abelian, these products are equal and their last components agree. Therefore,
$\operatorname{sym}(AJD^T)=\operatorname{sym}(DJA^T).$
Since we have
\begin{align}
\operatorname{sym}(AJD^T) &= \frac12(AJD^T-DJA^T),\\
\operatorname{sym}(DJA^T) &= \frac12(DJA^T-AJD^T),
\end{align}
equating these two expressions gives $2(AJD^T-DJA^T)=0$ and hence 
\begin{equation}\label{eqn:three-qudit-c4-AJD}
AJD^T=DJA^T
\end{equation}
for all \(A,D\in\mathcal A\). 

\item 
We claim that either the matrices $A: H \to W$ in \(\mathcal A\) have a common non-zero kernel vector $v \in H$ or $\mathcal A=\operatorname{span}(A_0)$ for some invertible matrix \(A_0\).  
Suppose first that \(A_0\in\mathcal A\) is invertible. For \(A\in\mathcal A\), write \(A=A_0D\). Applying \cref{eqn:three-qudit-c4-AJD} to \(A_0\) and \(A\) gives $JD^T=DJ$.  Writing $D=\begin{pmatrix}a&b\\c&d\end{pmatrix}$ gives
\begin{equation}
\begin{pmatrix}0&1\\-1&0\end{pmatrix}\begin{pmatrix}a&c\\b&d\end{pmatrix}
=
\begin{pmatrix}b&d\\-a&-c\end{pmatrix}
=
\begin{pmatrix}-b&a\\-d&c\end{pmatrix}
=
\begin{pmatrix}a&b\\c&d\end{pmatrix}\begin{pmatrix}0&1\\-1&0\end{pmatrix}.
\end{equation}
Hence, \(b=c=0\) and \(d=a\), so \(D=aI\). Therefore, \(\mathcal A=\operatorname{span}(A_0)\). Suppose instead that every matrix in \(\mathcal A\) is singular and \(\mathcal A \neq 0\). After changing bases we may choose a non-zero element of \(\mathcal A\) of the form\footnote{
Every non-zero singular \(2\times2\) matrix has rank \(1\), so it may be written as \(A_0=uv^T\) with \(u,v\neq0\), where \(u\) is a vector in \(W\) and \(v\) is a vector in \(H\). Since \(A_0:H\to W\), we may change the bases of its codomain \(W\) and domain \(H\) independently. Choose a basis change \(P\in\operatorname{GL}(W)\) such that \(P^{-1}u=e_1\), changing the basis of \(W^*\) dually so that the symplectic pairing is preserved. Independently, since \(\operatorname{SL}_2(\mathbb F_p) = \operatorname{Sp}(H)\) acts transitively on the non-zero vectors of \(H\), choose a symplectic basis change \(Q\in\operatorname{SL}_2(\mathbb F_p)\) such that \(v^TQ=e_1^T\). Under these changes of basis, the matrix of \(A_0\) becomes
$P^{-1}A_0Q=(P^{-1}u)(v^TQ)=e_1e_1^T$. 
}
\begin{equation}
A_0=\begin{pmatrix}1&0\\0&0\end{pmatrix}.
\end{equation}
For \(A=(a_{ij})\in\mathcal A\), every matrix \(A_0+tA\) is singular. Hence, we have
\begin{equation}
0=\det(A_0+tA)=ta_{22}+t^2\det A=ta_{22}.
\end{equation}
for every $t \in \mathbb F_p$.  Hence,  \(a_{22}=0\) and 
\(a_{12}a_{21}=0\). If one element of \(\mathcal A\) has \(a_{12}\neq0\) and another has \(a_{21}\neq0\), their sum is invertible, which is a contradiction. Thus, either all matrices in \(\mathcal A\) have the common kernel line \(\mathbb F_p e_2\), or all of their images are contained in the common image line \(\mathbb F_p e_1\). Suppose we are in the latter case and that there is no common non-zero kernel vector. Then every matrix in \(\mathcal A\) has the form \(u\mu\), where \(u\) is a fixed non-zero column vector spanning the common image line and \(\mu\) is a row vector. 
We can choose independent row vectors \(\mu,\nu\) such that \(u\mu,u\nu\in\mathcal A\)\footnote{If all row vectors \(\mu\) occurring in factorizations \(u\mu\in\mathcal A\) were proportional to a fixed nonzero row vector \(\mu_0\), then every matrix \(u\mu\in\mathcal A\) would satisfy \(\ker(u\mu)=\ker\mu_0\), giving a common kernel line.}.
Applying \cref{eqn:three-qudit-c4-AJD} with $A=u \mu$ and $D=u \nu$ gives
\begin{equation}
(\mu J\nu^T)uu^T=(\nu J\mu^T)uu^T=-(\mu J\nu^T)uu^T,
\end{equation}
and hence
$2(\mu J\nu^T)uu^T=0.$
This is impossible because \(p\) is odd, \(u\neq0\), and \(\mu J\nu^T\neq0\) for independent \(\mu,\nu\). Thus the common-image case also has a common kernel, proving the claim.

\item  We now rule out the possibility that \(\mathcal A\) has a common non-zero kernel line. Let \(v\neq 0\) and suppose that \(\ell:=\operatorname{span}\{v\}\subseteq H\) satisfies \(Av=0\) for every \(A\in\mathcal A\). We have
\begin{equation}
T(I,A,Z)\begin{pmatrix}0\\v\\0\end{pmatrix}
=
\begin{pmatrix}Av\\v\\0\end{pmatrix}
=
\begin{pmatrix}0\\v\\0\end{pmatrix}.
\end{equation}
Thus, every element of \(E_U\) fixes \(\ell\) pointwise. Since \(E_U\) fixes \(W\) pointwise by construction, it follows that \(E_U\) fixes \(W\oplus\ell\) pointwise. Moreover, \(W\) is isotropic, \(\ell\subseteq H\subseteq W^\perp\), and \(\ell\) is one-dimensional, so \(W\oplus\ell\) is isotropic. Since
\begin{equation}
\dim(W\oplus\ell)=2+1=3=\frac12\dim\mathbb F_p^6,
\end{equation}
it is Lagrangian. This contradicts \cref{lem:three-qudit-large-image-obstruction}  because  \(\dim E_U=4\). 

\item By (c), we have  $\mathcal A= \operatorname{span}(A_0)$ for some invertible \(A_0\in M_2(\mathbb F_p)\). We now use this description to show that \(\beta\) has rank at most \(1\). Consider the $\mathbb F_p$-linear map
\begin{align}
\varphi:E_U&\longrightarrow\mathcal A,\\
T(I,A,Z)&\longmapsto A.
\end{align}
Its image is \(\mathcal A=\operatorname{span}(A_0)\), which has dimension \(1\). Since \(\dim E_U=4\), rank-nullity gives
$\dim_{\mathbb F_p}\ker\phi=4-1=3.$
The kernel consists of those elements of \(E_U\) whose \(A\) component is zero. Hence, 
\begin{equation}
\ker\varphi\subseteq Z_0:=\left\{T(I,0,Z):Z=Z^T\in\operatorname{Sym}_2(\mathbb F_p)\right\}.
\end{equation}
Since \(\dim \operatorname{Sym}_2(\mathbb F_p)=3\), we have $\dim Z_0 = 3$.
Thus, \(\ker\varphi\) and \(Z_0\) have the same dimension. Hence, 
$\ker\varphi=Z_0.$
Since \(A_0\in\mathcal A\), choose \(h\in\mathbb F_p^6\) such that
\begin{equation}
S_U(h)=T(I,A_0,Z_h)
\end{equation}
for some symmetric matrix \(Z_h\). Now let \(b\in\mathbb F_p^6\) satisfy \(S_U(b)\in Z_0\), so that
$S_U(b)=T(I,0,Z_b)$
for some symmetric matrix \(Z_b\).
Since \(\mathbb F_p^6\) is abelian, the affine cocycle equation applied to \(h+b=b+h\) gives
\begin{equation}
q_U(h)+S_U(h)q_U(b)=q_U(b)+S_U(b)q_U(h).
\end{equation}
Write \(q_U(a)=(\alpha(a),\gamma(a),\beta(a))\). The \(H\) component of \(T(I,A,Z)(\alpha,\gamma,\beta)\) is
$\gamma+JA^T\beta.$
Comparing the \(H\) components in the cocycle identity gives
\begin{equation}
\gamma(h)+\gamma(b)+JA_0^T\beta(b)=\gamma(b)+\gamma(h),
\end{equation}
and hence
$JA_0^T\beta(b)=0.$
Since \(J\) and \(A_0\) are invertible, we obtain \(\beta(b)=0\). It follows that \(\beta\) vanishes on \(S_U^{-1}(Z_0)\). Since \(\ker S_U=K_U\) and \(S_U\) maps \(S_U^{-1}(Z_0)\) onto \(Z_0\), we have
$S_U^{-1}(Z_0)/K_U\cong Z_0.$
Therefore,
\begin{equation}
\dim S_U^{-1}(Z_0)=\dim K_U+\dim Z_0=2+3=5.
\end{equation}
Thus \(\beta\) vanishes on a five-dimensional subspace of \(\mathbb F_p^6\), so
$\operatorname{rank}\beta\leq 6-5=1$, a contradiction.
\end{enumerate}

\item The $\dim \rho(E_U)=1$ case is proved in \cref{appendix:rank-one-codim4-three}.
\end{enumerate}
This completes the proof. 
\end{proof}

\cref{lem:three-qudit-codim-five-impossible} rules out the codimension \(5\) case. In this case, \(K_U\) is one-dimensional, so \(W:=q_U(K_U)\) is an isotropic line fixed pointwise by \(E_U\). We then analyze the induced action of \(E_U\) on a four-dimensional symplectic complement of \(W\) and show that every possible image forces \(E_U\) to fix a Lagrangian subspace pointwise, contradicting \cref{lem:three-qudit-large-image-obstruction}. Since several details parallel the proof of \cref{lem:three-qudit-codim-four-impossible}, we relegate the proof to \cref{appendix:codim5-proof}. 

\begin{prop}\label{lem:three-qudit-codim-five-impossible}
Let \(p\) be an odd prime and let \(U\in\mathcal C^{(3)}\) be a three-qudit gate. Set \(K_U:=\ker S_U\) and \(E_U:=\operatorname{im}S_U\). Then \(\operatorname{codim}K_U\neq5\).
\end{prop}

\cref{last} rules out the $\operatorname{codim}K_U=6$ case. 

\begin{prop}\label{last}
Let \(p\) be an odd prime and let \(U\in\mathcal C^{(3)}\) be a three-qudit gate. Set \(K_U:=\ker S_U\) and \(E_U:=\operatorname{im}S_U\). Then \(\operatorname{codim}K_U\neq6\).
\end{prop}

\begin{proof}
In this case \(K_U=0\), so \(E_U\cong\mathbb F_p^6\). In particular,
$\dim_{\mathbb F_p}E_U=6$
and \(|E_U|=p^6\).  Choose a symplectic basis
$e_1,e_2,e_3,f_1,f_2,f_3$
of \(\mathbb F_p^6\), and define $W_i:=\operatorname{span}\{e_1,\ldots,e_i\}$ for 
$0\leq i\leq 3.$
Then
\begin{equation}
0=W_0\subset W_1\subset W_2\subset W_3
\end{equation}
is a complete isotropic flag, with \(W_3\) Lagrangian. Relative to the decomposition
\begin{equation}
\mathbb F_p^6=W_3\oplus\operatorname{span}\{f_1,f_2,f_3\},
\end{equation}
the associated Sylow \(p\)-subgroup is
\begin{equation}\label{hi}
P:=\left\{\begin{pmatrix}R&RS\\0&R^{-T}\end{pmatrix}:R\in U_3(\mathbb F_p),\ S\in\operatorname{Sym}_3(\mathbb F_p)\right\},
\end{equation}
where
\begin{equation}
U_3(\mathbb F_p):=\left\{\begin{pmatrix}1&a&b\\0&1&c\\0&0&1\end{pmatrix}:a,b,c\in\mathbb F_p\right\}.
\end{equation}
If \(R \in U_3(\mathbb F_p)\), then $R$ preserves each \(W_i\) and acts trivially on every successive quotient \(W_i/W_{i-1}\). Moreover, preservation of the symplectic form by a block upper-triangular matrix
\begin{equation}
\begin{pmatrix}R&B\\0&D\end{pmatrix}
\end{equation}
forces \(D=R^{-T}\) and \(R^{-1}B\) to be symmetric. Hence, \(B=RS\) for a unique \(S\in\operatorname{Sym}_3(\mathbb F_p)\), giving \cref{hi}.
There are three free parameters in \(R\in U_3(\mathbb F_p)\) and six free parameters in \(S\in\operatorname{Sym}_3(\mathbb F_p)\). Therefore,
$|P|=p^{3+6}=p^9.$
On the other hand,
\begin{equation}
|\operatorname{Sp}_6(\mathbb F_p)|=p^9(p^2-1)(p^4-1)(p^6-1).
\end{equation}
Hence,  \(P\) is a Sylow \(p\)-subgroup of \(\operatorname{Sp}_6(\mathbb F_p)\). In fact, \(P=P(W_0,W_1,W_2,W_3)\) in the notation of \cref{wong}.
Since \(E_U\) is a \(p\)-subgroup of \(\operatorname{Sp}(\mathbb F_p^6)\), it is contained in some Sylow \(p\)-subgroup. Since all Sylow \(p\)-subgroups are conjugate, after conjugating inside \(\operatorname{Sp}(\mathbb F_p^6)\), we may assume without loss of generality that
$E_U\subseteq P.$
By \cref{wong}, \(A(W_3)\) unique abelian subgroup of largest order in \(P\). In fact, $A(W_3)$ can be identified with 
\begin{equation}
A(W_3)=\left\{\begin{pmatrix}I&S\\0&I\end{pmatrix}:S\in\operatorname{Sym}_3(\mathbb F_p)\right\},
\end{equation}
Since $|A(W_3)|=|E_U|=p^6$, we have $E_U=A(W_3)$. Therefore, \(E_U\) fixes the Lagrangian \(W_3\) pointwise.
Since \(\dim E_U=6\geq4\), this contradicts \cref{lem:three-qudit-large-image-obstruction}. 
\end{proof}

We can now prove \cref{three-qudit-c3-semi-clifford}.

\begin{proof}
This follows from \cref{codim2-gen}, \cref{lem:three-qudit-codim-three}, \cref{lem:three-qudit-codim-four-impossible}, \cref{lem:three-qudit-codim-five-impossible} and \cref{last}. 
\end{proof}


\addtocontents{toc}{\protect\setcounter{tocdepth}{1}}
\appendix
\markboth{APPENDIX}{APPENDIX}
\section{Facts about Weyl Operators}\label[appendix]{appendix:support}
In the main text, we have repeatedly used the fact that if \(\partial_a U=U(1)I\) for every \(a\in V_n\), then \(U\) is proportional to a single Weyl operator. We prove a more general version of this statement in \cref{lem:single-Weyl-support}.

\begin{lem}\label{lem:single-Weyl-support}
Let \(U\in \mathcal U(\mathcal H_n)\) such that \(\partial_a U\in U(1)I\) for every \(a\in V_n\). Equivalently, suppose that there exists a phase function \(\lambda:V_n\to U(1)\) such that \(\tau_a(U)=\lambda(a)U\) for every \(a\in V_n\). Then there exist \(v\in V_n\) and \(\zeta\in U(1)\) such that
$U=\zeta w(v)$. 
\end{lem}

\begin{proof}
Expand \(U\) in the Weyl basis as \(U=\sum_{v\in V_n} c_v w(v)\). Applying \(\tau_a\) and using \cref{tau-on-Weyl}, we obtain
\begin{equation}
\sum_{v\in V_n} \lambda(a)c_v w(v)=\lambda(a)U=\tau_a(U)=\sum_{v\in V_n} c_v \tau_a(w(v))=\sum_{v\in V_n} c_v \omega^{[a,v]}w(v).
\end{equation}
Uniqueness of the Weyl expansion implies that \(c_v\omega^{[a,v]}=\lambda(a)c_v\) for all \(a,v\in V_n\). Hence, if \(c_v\neq 0\), then \(\omega^{[a,v]}=\lambda(a)\) for every \(a\in V_n\). We claim that at most one coefficient \(c_v\) can be non-zero. Indeed, if \(c_v, ,c_{v'}\neq 0\), then
\begin{equation}
\omega^{[a,v]}=\lambda(a)=\omega^{[a,v']}
\end{equation}
for every \(a\in V_n\). Therefore, \(\omega^{[a,v-v']}=1\) for every \(a\in V_n\). Since \(\omega\) is a primitive \(p\)-th root of unity, we must have \([a,v-v']=0\) for every \(a\in V_n\). The non-degeneracy of the symplectic form forces \(v-v'=0\), and hence \(v=v'\). Thus, \(U\) has support on a single Weyl basis element. Since \(U\neq 0\), there exist \(v\in V_n\) and \(\zeta\in\mathbb C\) such that \(U=\zeta w(v)\). Finally, since \(U\) and \(w(v)\) are unitary, \(|\zeta|=1\). Hence, \(\zeta\in U(1)\), as required.
\end{proof}

\cref{cor:tau-invariants-are-scalars} records a basic consequence of \cref{lem:single-Weyl-support}. Namely, if \(\partial_a U\in U(1)I\) for every \(a\in V_n\), then \(U\) is proportional to a single Weyl operator. As noted above, we have repeatedly used this fact in the main text.

\begin{cor}\label{cor:tau-invariants-are-scalars}
Let \(U\in \mathcal U(\mathcal H_n)\). If \(\partial_a U=I\) for every \(a\in V_n\), then \(U\in U(1)I\).
\end{cor}

\begin{proof}
Since \(\partial_a U=\tau_a(U)U^*\), the assumption \(\partial_a U=I\) implies that \(\tau_a(U)=U\) for every \(a\in V_n\). Thus, \(U\) satisfies the hypothesis of \cref{lem:single-Weyl-support} with the constant phase function \(\lambda(a)=1\). Hence, there exist \(v\in V_n\) and \(\zeta\in U(1)\) such that $U=\zeta w(v)$. Applying \(\tau_a\) and using \cref{tau-on-Weyl}, we get
\begin{equation}
\zeta w(v)=U=\tau_a(U)=\zeta \omega^{[a,v]}w(v) \quad \Longrightarrow \quad
\zeta (\omega^{[a,v]}-1)w(v) = 0.
\end{equation}
for every $a \in V_n$. It follows that $\omega^{[a,v]}=1$ for every \(a\in V_n\). As in \cref{lem:single-Weyl-support}, we conclude that \(v=0\). Therefore, \(U=\zeta w(0)=\zeta I\), and hence \(U\in U(1)I\).
\end{proof}

\cref{lem:Weyl-label-phase-uniqueness} and \cref{lem:weyl-scalar-label} record two further basic facts about Weyl operators that are invoked in the main text.

\begin{lem}\label{lem:Weyl-label-phase-uniqueness}
Let \(x,y\in V_n\) and let \(\lambda,\lambda'\in U(1)\). If $\lambda w(x)=\lambda' w(y)$, then \(x=y\) and \(\lambda=\lambda'\).
\end{lem}

\begin{proof}
Multiplying the equality \(\lambda w(x)=\lambda' w(y)\) on the left by \(w(x)^*\), we obtain $\lambda I=\lambda' w(x)^*w(y)$.  Since \(w(x)^*w(y)\) is a scalar multiple of \(w(y-x)\), there exists \(\mu\in U(1)\) such that $\lambda I=\lambda'\mu w(y-x).$
We claim that \(x=y\). Indeed, if \(x\neq y\), then \(y-x\neq 0\), and so \(\operatorname{Tr}(w(y-x))=0\) by \cref{weyl-inner}. Taking traces in the preceding equality gives
\begin{equation}
0 \neq \lambda \operatorname{Tr}(I)=\lambda'\mu \operatorname{Tr}(w(y-x))=0.
\end{equation}
This is a contradiction. Therefore, \(x=y\). Substituting \(x=y\) into \(\lambda w(x)=\lambda' w(y)\), we obtain $\lambda w(x)=\lambda' w(x)$.  Since \(w(x)\) is unitary, it is invertible. Hence, \(\lambda=\lambda'\).
\end{proof}

\begin{cor}\label{lem:weyl-scalar-label}
Let \(a\in V_n\). If \(w(a)\in U(1)I\), then \(a=0\).
\end{cor}

\begin{proof}
Since \(w(a)\in U(1)I\), there exists \(\lambda\in U(1)\) such that $w(a)=\lambda I.$ Since \(w(0)=I\), this can be rewritten as $w(a)=\lambda w(0)$. By \cref{lem:Weyl-label-phase-uniqueness}, we get \(a=0\).
\end{proof}

\section{Cohomological Results} \label[appendix]{appendix:cohom-facts}
In \cref{appendix:cohom-facts}, we discuss and prove various cohomological results invoked in \cref{third-clifford}. We use standard results from group cohomology, singular cohomology, and algebraic topology, and refer the reader to~\cite{weibel1994introduction,hatcher-alg-top} for further background. The main purpose of \cref{appendix:cohom-facts} is to show that, for a finite abelian group, a \(U(1)\)-valued \(H^2\)-class is completely determined by the commutator bicharacter of its associated central extension. This allows us to convert the abstract \(H^2\)-obstructions appearing in \cref{third-clifford} into explicit equations that can be checked on a basis. We establish this result in several steps. \cref{commutator-bicharacter-well-defined} shows that a normalized \(2\)-cocycle \(R\) determines an alternating bicharacter which depends only on the cohomology class \([R]\). 

\begin{lem}\label{commutator-bicharacter-well-defined}
Let \(A\) be an abelian group, and let \(U(1)\) be equipped with the trivial \(A\)-action. If \(R\in Z^2(A,U(1))\) is a normalized \(2\)-cocycle, then the map
\begin{equation}
\operatorname{comm}_R(a,b):=\frac{R(a,b)}{R(b,a)}
\end{equation}
is an alternating bicharacter on \(A\), and hence defines an element of \(\operatorname{Hom}_{\mathbb Z}(\Lambda^2 A,U(1))\). Moreover, if \(R\) and \(R'\) are cohomologous normalized \(2\)-cocycles, then \(\operatorname{comm}_R=\operatorname{comm}_{R'}\).
\end{lem}

\begin{proof}
Recall that a cohomology class \([R]\in H^2(A,U(1))\) determines an equivalence class of central extensions of \(A\) by \(U(1)\). If \(R\) is a normalized cocycle representative, then \(R(0,a)=R(a,0)=1\) for all \(a\in A\)\footnote{It can be shown that every \(2\)-cocycle is cohomologous to a normalized \(2\)-cocycle.}. In this case, an explicit central extension is given by
\begin{equation} 
1\longrightarrow U(1)\longrightarrow E_R\longrightarrow A\longrightarrow 1.
\end{equation}
Here \(E_R=U(1)\times A\) as a set, the map $U(1) \to E_R$ is the inclusion map and $E_R \to A$ is the projection map. The multiplication law on $E_R$ is
\begin{equation} 
(z,a)(w,b)=(zwR(a,b),a+b).
\end{equation}
The cocycle identity gives associativity, and the normalization of \(R\) implies that \((1,0)\) is the identity element and that \(U(1) \times \{0\}\) is contained in the center of $E_R$. Thus, this defines a central extension of \(A\) by \(U(1)\). Since \(A\) is abelian, every commutator in \(E_R\) maps trivially to \(A\). Therefore, every commutator in \(E_R\) is in the kernel of the map $E_R \to A$, which is precisely \(U(1)\times\{0\}\). Therefore, for \(a,b\in A\), choose lifts \(\widetilde a,\widetilde b\in E_R\) and define
\begin{equation}
[\widetilde a,\widetilde b]\in U(1)\times\{0\}\cong U(1).
\end{equation}
This is independent of the chosen lifts. Indeed, if \(\widetilde a'\) and \(\widetilde b'\) are other lifts, then \(\widetilde a'=u\widetilde a\) and \(\widetilde b'=v\widetilde b\) for some \(u,v\in U(1)\times\{0\}\). Since \(U(1)\times\{0\}\) is contained in the center of \(E_R\), the elements \(u\) and \(v\) commute with all elements of \(E_R\). Indeed, we have 
\begin{align}
[u\widetilde a,v\widetilde b]
&= u\widetilde a v\widetilde b \widetilde a^{-1}u^{-1}\widetilde b^{-1}v^{-1} \\
&= \widetilde a\widetilde b\widetilde a^{-1}\widetilde b^{-1} \\
&= [\widetilde a,\widetilde b].
\end{align}
Thus, the commutator is independent of the chosen lifts and the central extension \(E_R\) defines a well-defined commutator pairing. 
We now compute this commutator pairing explicitly using the standard lift \((1,a)\) of each element \(a\in A\). For $a,b \in A$, we have
\begin{align}
(1,a)(1,b) &= (R(a,b),a+b), \\
(1,b)(1,a) &= (R(b,a),a+b).
\end{align}
We compute the commutator directly. Since $(1,x)^{-1}=(R(x,-x)^{-1},-x)$\footnote{
This follows from the cocycle identity and normalization applied to the triple \((x,-x,x)\).} for all $x \in A$, we have
\begin{align}
[(1,a),(1,b)]
&=(1,a)(1,b)(1,a)^{-1}(1,b)^{-1}\\
&=(1,a)(1,b)(R(a,-a)^{-1},-a)(R(b,-b)^{-1},-b)\\
&=(R(a,b),a+b)(R(a,-a)^{-1},-a)(R(b,-b)^{-1},-b)\\
&=(R(a,b)R(a,-a)^{-1}R(a+b,-a),b)(R(b,-b)^{-1},-b)\\
&=(R(a,b)R(a,-a)^{-1}R(a+b,-a),0).
\end{align}
Noting that $R(a,-a)^{-1}R(a+b,-a)=R(b,a)^{-1}$\footnote{This follows from the cocycle identity and normalization applied to the triple \((b,a,-a)\).}, we have
\begin{equation}
[(1,a),(1,b)]=\left(\frac{R(a,b)}{R(b,a)},0\right).
\end{equation}
for all $a,b \in A$. Hence, the commutator pairing is given by the map
\begin{align}
A\times A &\longrightarrow U(1) \\
(a,b) &\longmapsto \frac{R(a,b)}{R(b,a)}.
\end{align}
This is exactly the map \(\operatorname{comm}_R\). It remains to check that the commutator pairing is an alternating map. For any \(x,x',y\in E_R\), we have 
\begin{align}
[xx',y]
&=xx'y(xx')^{-1}y^{-1}  \\
&=xx'yx'^{-1}x^{-1}y^{-1} \\
&=x[x',y]yx'x'^{-1}x^{-1}y^{-1} \\
&=x[x',y]yx^{-1}y^{-1} \\
&=[x',y]xyx^{-1}y^{-1}=[x',y][x,y]=[x,y][x',y].
\end{align}
Similarly, one has \( [x,yy']=[x,y][x,y'] \). Together with $\operatorname{comm}_R(a,a)=R(a,a)/R(a,a)=1$, these identities show that the commutator pairing is an alternating bicharacter. By the universal property of alternating maps, every alternating bicharacter $A\times A\longrightarrow U(1)$ determines a unique homomorphism $\Lambda^2 A\longrightarrow U(1).$
Therefore \(\operatorname{comm}_R\) defines an element of $\operatorname{Hom}_{\mathbb Z}(\Lambda^2 A,U(1)).$ Suppose that \([R]=[R']\). Then $R'=R\,\delta\alpha$ for some \(1\)-cochain \(\alpha:A\to U(1)\), where
\begin{equation}
(\delta\alpha)(a,b)
=\frac{\alpha(a)\alpha(b)}{\alpha(a+b)}
=\frac{\alpha(b)\alpha(a)}{\alpha(b+a)}=(\delta\alpha)(b,a).
\end{equation}
Therefore, we have
\begin{equation}
\operatorname{comm}_{R'}(a,b)=\frac{R'(a,b)}{R'(b,a)}=\frac{R(a,b)(\delta\alpha)(a,b)}{R(b,a)(\delta\alpha)(b,a)}=\frac{R(a,b)}{R(b,a)}=\operatorname{comm}_R(a,b).
\end{equation}
This completes the proof.
\end{proof}

\cref{commutator-bicharacter-well-defined} shows that the assignment \(R\mapsto \operatorname{comm}_R\) induces a homomorphism
\begin{align}\label{map-to-show-iso}
H^2(A,U(1)) &\longrightarrow \operatorname{Hom}_{\mathbb Z}(\Lambda^2 A,U(1)), \\
[R] &\longmapsto \left(a\wedge b\longmapsto \frac{R(a,b)}{R(b,a)}\right).
\end{align}
We want to show that this homomorphism is an isomorphism when \(A\) is a finite abelian group. In \cref{h2-finite-abelian-alt2}, we use the universal coefficient theorem to obtain an abstract isomorphism
\begin{equation}
H^2(A,U(1)) \cong \operatorname{Hom}_{\mathbb Z}(\Lambda^2A,U(1)).
\end{equation}
We then show in \cref{comm-map-iso} that the specific homomorphism in \cref{map-to-show-iso} realizes this isomorphism.

\begin{lem}\label{h2-finite-abelian-alt2}
Let \(A\) be a finite abelian group. Then \(H^2(A,U(1)) \cong \operatorname{Hom}_{\mathbb Z}(\Lambda^2 A,U(1))\), and both groups are finite.
\end{lem}

\begin{proof}
We compute group cohomology using the classifying space \(BA\). Since \(A\) is discrete and the \(A\)-action on \(U(1)\) is trivial, group cohomology agrees with the singular cohomology of \(BA\) with $U(1)$ coefficients. That is, 
\begin{equation}
H^m(A,U(1))\cong H^m(BA;U(1)).
\end{equation}
By abuse of notation, we continue to denote the singular cohomology groups \(H^m(BA;U(1))\) by \(H^m(A,U(1))\). Recall that \(BA\) is an Eilenberg--Mac Lane space \(K(A,1)\). The universal coefficient theorem for singular cohomology gives a short exact sequence of abelian groups
\begin{equation}\label{uct}
0 \longrightarrow \operatorname{Ext}^1_{\mathbb Z}(H_1(BA;\mathbb Z),U(1)) \longrightarrow H^2(BA;U(1)) \longrightarrow \operatorname{Hom}_{\mathbb Z}(H_2(BA;\mathbb Z),U(1)) \longrightarrow 0.
\end{equation}
Since \(U(1)\) is a divisible group and divisible abelian groups are injective objects in the category of abelian groups, we have $\operatorname{Ext}^1_{\mathbb Z}(H_1(BA;\mathbb Z),U(1)) = 0$. By \cref{uct}, we have
\begin{equation}\label{uct-cons}
H^2(A,U(1)) \cong \operatorname{Hom}_{\mathbb Z}(H_2(BA;\mathbb Z),U(1)).
\end{equation}
It remains to identify \(H_2(BA;\mathbb Z)\). Since \(A\) is a finite abelian group, we have $A\cong \bigoplus_{i=1}^r \Z_{m_i}$. Since both \(BA\) and \(\prod_{i=1}^r B\mathbb Z_{m_i}\) are \(K(A,1)\)-spaces, uniqueness of Eilenberg--Mac Lane spaces up to homotopy equivalence implies that we have
\begin{equation}\label{eq:BA-product-decomposition}
BA\simeq \prod_{i=1}^r B\mathbb Z_{m_i}.
\end{equation}
We use the standard topological model $B\mathbb Z_m \simeq \mathbb S^\infty/\mathbb Z_m$,  where \(\mathbb Z_m\) acts freely on \(\mathbb S^\infty \subseteq \mathbb C^\infty\) by scalar multiplication by \(m\)-th roots of unity. The quotient \(\mathbb S^\infty/\mathbb Z_m\) is the infinite lens space \(L(\infty,m)\), and it is a \(K(\mathbb Z_m,1)\) space. It is well known that
\begin{equation}\label{eq:homology-infinite-lens-space}
H_k(L(\infty,m);\mathbb Z)\cong
\begin{cases}
\mathbb Z, & k=0,\\
\mathbb Z_m, & k>0 \text{ odd},\\
0, & k>0 \text{ even}.
\end{cases}
\end{equation}
We now apply the Künneth theorem. Note that the possible \(\operatorname{Tor}\)-terms in degree \(2\) vanish, since they involve \(H_0(B\mathbb Z_{m_i};\mathbb Z)\cong \mathbb Z\), and \(\operatorname{Tor}_1^{\mathbb Z}(\mathbb Z,M)=\operatorname{Tor}_1^{\mathbb Z}(M,\mathbb Z)=0\) for every abelian group \(M\). Thus the degree-two homology receives contributions only from tensor products of degree-one homology groups from distinct factors. Hence, we have
\begin{equation}\label{eq:H2-BA-kunneth}
H_2(BA;\mathbb Z) \cong \bigoplus_{1\leq i<j\leq r} H_1(B\mathbb Z_{m_i};\mathbb Z)\otimes_{\mathbb Z} H_1(B\mathbb Z_{m_j};\mathbb Z) \cong \bigoplus_{1\leq i<j\leq r} \mathbb Z_{m_i}\otimes_{\mathbb Z} \mathbb Z_{m_j}.
\end{equation}
On the other hand, $\Lambda^2 A$ decomposes as
\begin{equation}
\Lambda^2\left(\bigoplus_{i=1}^r \mathbb Z_{m_i}\right) \cong \left(\bigoplus_{i=1}^r \Lambda^2 \mathbb Z_{m_i}\right) \oplus \left(\bigoplus_{1\leq i<j\leq r} \mathbb Z_{m_i}\otimes_{\mathbb Z}\mathbb Z_{m_j}\right).
\end{equation}
Since each \(\mathbb Z_{m_i}\) is cyclic,  we have \(\Lambda^2 \mathbb Z_{m_i}=0\). Therefore, we have
\begin{equation}\label{h2-ba}
H_2(BA;\mathbb Z) \cong \bigoplus_{1\leq i<j\leq r} \mathbb Z_{m_i}\otimes_{\mathbb Z} \mathbb Z_{m_j} \cong \Lambda^2 A,
\end{equation}
\cref{uct-cons} and \cref{h2-ba} imply \(H^2(A,U(1))\cong \operatorname{Hom}_{\mathbb Z}(\Lambda^2 A,U(1))\).
\end{proof}

We now use \cref{commutator-bicharacter-well-defined} and \cref{h2-finite-abelian-alt2} to show that the map in \cref{map-to-show-iso} is an isomorphism.

\begin{cor}\label{comm-map-iso}
Let \(A\) be a finite abelian group. Then the commutator pairing
\begin{align}
H^2(A,U(1)) &\longrightarrow \operatorname{Hom}_{\Z}(\Lambda^2A,U(1)), \\
[R] &\longmapsto \operatorname{comm}_R
\end{align}
is an isomorphism of finite groups.
\end{cor}

\begin{proof}
By \cref{h2-finite-abelian-alt2}, we have an abstract isomorphism \(H^2(A,U(1)) \cong \operatorname{Hom}_{\mathbb Z}(\Lambda^2 A,U(1))\) of finite groups. Hence, it suffices to show that the commutator pairing from \cref{commutator-bicharacter-well-defined} is injective. Let $[R] \in H^2(A,U(1)) $ such that \(\operatorname{comm}_R(a,b)=1\) for all \(a,b\in A\). Equivalently, in the corresponding central extension \(E_R\), the commutator of any two lifts of elements of \(A\) is trivial. Let \(x,y\in E_R\). If \(a,b\in A\) are their images, then \(x=u\widetilde a\) and \(y=v\widetilde b\) for some lifts \(\widetilde a,\widetilde b\) of \(a,b\) and \(u,v\in U(1)\). Since \(U(1)\) lies in the center of \(E_R\), we have \([x,y]=[\widetilde a,\widetilde b]\). By assumption this commutator is trivial, so \([x,y]=1\) for all $x,y \in E_R$. Thus, \(E_R\) is an abelian group. Hence, we obtain a short exact sequence of abelian groups
\begin{equation}\label{ses-ab}
0 \longrightarrow U(1) \longrightarrow E_R \longrightarrow A \longrightarrow 0.
\end{equation}
Since \(U(1)\) is a divisible group, it is injective in the category of abelian groups. Therefore, the short exact sequence in \cref{ses-ab} splits in the category of abelian groups. Choose a splitting \(s:A\to E_R\) such that \(s(a)=(\alpha(a),a)\), where $\alpha : A \to U(1)$ is a map. The condition that \(s\) is a homomorphism gives
\begin{equation}
(\alpha(a+b),a+b)=s(a+b)=s(a)s(b)=(\alpha(a)\alpha(b)R(a,b),a+b).
\end{equation}
Therefore, we have
\begin{equation}\label{preceding}
R(a,b)=\frac{\alpha(a+b)}{\alpha(a)\alpha(b)}.
\end{equation}
With the convention that, for a \(1\)-cochain \(\beta\),
$(\delta\beta)(a,b)=\frac{\beta(a)\beta(b)}{\beta(a+b)},$
\cref{preceding} gives \(R=\delta(\alpha^{-1})\). Hence, \(R\) is a $2$-coboundary, so \([R]=0\). Thus \(\ker \Phi=0\), and therefore \(\Phi\) is injective. This completes the proof.
\end{proof}

\cref{cor:coboundary-iff-trivial-commutator} gives an explicit criterion for the vanishing of a class in \(H^2(A,U(1))\) in terms of the commutator bicharacter of a normalized cocycle representative. When \(A\) is a finite-dimensional vector space over \(\mathbb F_p\), it suffices to check this criterion on pairs of basis vectors, thereby reducing the abstract cohomological obstruction to finitely many explicit conditions.

\begin{cor}\label{cor:coboundary-iff-trivial-commutator}
Let \(A\) be a finite abelian group, and let \(R\in Z^2(A,U(1))\) be a normalized \(2\)-cocycle. Then $[R]=0\in H^2(A,U(1))$ if and only if $\frac{R(a,b)}{R(b,a)}=1$ for all \(a,b\in A\). In particular, if \(V\) is a finite-dimensional vector space over \(\mathbb F_p\) with basis \(e_1,\ldots,e_m\), then for any normalized \(2\)-cocycle \(R\in Z^2(V,U(1))\), one has $[R]=0\in H^2(V,U(1))$ if and only if
\begin{equation}\label{basis-comm}
\frac{R(e_i,e_j)}{R(e_j,e_i)}=1
\end{equation}
for all \(1\leq i,j\leq m\).
\end{cor}

\begin{proof}
By \cref{comm-map-iso}, the class \([R]\) vanishes if and only if its image under the commutator pairing in \(\operatorname{Hom}_{\mathbb Z}(\Lambda^2 A,U(1))\) vanishes. Since this image is the alternating bicharacter \(\operatorname{comm}_R\), it follows that \([R]=0\) if and only if $\frac{R(a,b)}{R(b,a)}=1$ for all \(a,b\in A\). Take \(A=V\). Since \(\operatorname{comm}_R\) is biadditive, it is determined by its values on pairs of basis vectors. Indeed, if \(v=\sum_i \lambda_i e_i\) and \(w=\sum_j \mu_j e_j\), then biadditivity gives
\begin{equation}
\operatorname{comm}_R(v,w)
=
\prod_{i,j}\operatorname{comm}_R(e_i,e_j)^{\lambda_i\mu_j}.
\end{equation}
Thus \(\operatorname{comm}_R\) is determined by its values on pairs of basis vectors. Therefore, \(\operatorname{comm}_R\) is trivial on \(V\times V\) if and only if \(\operatorname{comm}_R(e_i,e_j)=1\) for all \(1\leq i,j\leq m\). This is exactly \cref{basis-comm}.  
\end{proof}

\section{Cohomological Parametrization in Characteristic Two}
\label[appendix]{appendix:characteristic-two-third-level}
In \cref{second-level-symplectic-quotient} and the discussion following it in \cref{subsec:second-clifford-level}, as well as throughout \cref{third-clifford}, we restrict our discussion to odd prime \(p\). In \cref{appendix:characteristic-two-third-level}, we develop the corresponding \(p=2\) analogs of the results in \cref{subsec:second-clifford-level} and \cref{third-level-linear-algebraic}--\cref{generator-level-criteria-third}. In particular, we show that the cohomological description of \(\mathcal C^{(3)}\) persists when \(p=2\), but now involves an additional obstruction in \(H^2(V_n,V_{n,S})\), which must vanish before the phase obstruction in \(H^2(V_n,U(1))\) can be considered.

\subsection{The Second Clifford Level in Characteristic Two}
We first discuss the \(p=2\) analog of \cref{second-level-symplectic-quotient} and the discussion following it in \cref{subsec:second-clifford-level}. When \(p=2\), the short exact sequence in \cref{eqn:clifford-symplectic-exact-sequence} continues to hold, but the choice of representatives \(M_S\) requires a modification. Indeed, by \cref{symplectic-invariance-weyl-multiplier}, for \(S\in\operatorname{Sp}(V_n)\), the projective representations \(a\mapsto w(a)\) and \(a\mapsto w(Sa)\) need not have the same multiplier. Their multipliers are \(i^{\mu(a,b)}\) and \(i^{\mu(Sa,Sb)}\), respectively. Both multipliers are normalized \(U(1)\)-valued \(2\)-cocycles on \(V_n\). Therefore,
\begin{equation}
R_S(a,b):=i^{\mu(a,b)-\mu(Sa,Sb)}
\end{equation}
is a normalized \(U(1)\)-valued \(2\)-cocycle. Its commutator bicharacter is trivial. Indeed, using \cref{qubit-mult-comm} and the symplecticity of \(S\), we have
\begin{equation}
\frac{i^{\mu(a,b)}}{i^{\mu(b,a)}} = (-1)^{[a,b]} = (-1)^{[Sa,Sb]} = \frac{i^{\mu(Sa,Sb)}}{i^{\mu(Sb,Sa)}}.
\end{equation}
Hence, $\frac{R_S(a,b)}{R_S(b,a)}=1$ for all \(a,b\in V_n\). By \cref{cor:coboundary-iff-trivial-commutator}, \(R_S\) is a \(2\)-coboundary. Thus, there exists a $1$-cochain \(\gamma_S:V_n\to U(1)\) such that
\begin{equation}\label{eqn:qubit-symplectic-phase-correction}
\frac{\gamma_S(a)\gamma_S(b)}{\gamma_S(a+b)} = i^{\mu(a,b)-\mu(Sa,Sb)}
\end{equation}
for all \(a,b\in V_n\). Setting \(a=b=0\) shows that \(\gamma_S(0)=1\). Define $\widetilde w_S(a):=\gamma_S(a)w(Sa)$. We have,
\begin{align}
\widetilde w_S(a)\widetilde w_S(b)
&=
\gamma_S(a)\gamma_S(b)i^{\mu(Sa,Sb)}w(S(a+b))\\
&=
i^{\mu(a,b)}\gamma_S(a+b)w(S(a+b))=
i^{\mu(a,b)}\widetilde w_S(a+b).
\end{align}
Thus, \(a\mapsto\widetilde w_S(a)\) is a projective representation of \(V_n\) with the same multiplier as \(a\mapsto w(a)\). By the finite Stone--von Neumann theorem in characteristic two, there exists \(M_S\in\mathcal U(\mathcal H_n)\) such that
\begin{equation}\label{eqn:qubit-clifford-representative-action}
M_Sw(a)M_S^*=\gamma_S(a)w(Sa)
\end{equation}
for every \(a\in V_n\). Hence, \(M_S\in\mathcal C^{(2)}\) and \(\Psi(M_S)=S\), so \(\Psi\) remains surjective when \(p=2\). The proof that \(\Psi\) is a group homomorphism and that \(\ker\Psi=\mathcal C^{(1)}\) is unchanged from the odd prime case in \cref{second-level-symplectic-quotient}.
Fix one such representative \(M_S\) for every \(S\in\operatorname{Sp}(V_n)\), with \(M_{\operatorname{Id}_{V_n}}=I\). The normal form in \cref{clifford-normal-form} therefore remains valid. As in the odd prime case, we have \(M_SM_TM_{ST}^*\in\ker\Psi=\mathcal C^{(1)}\). Therefore, there exist \(\kappa(S,T)\in U(1)\) and \(c(S,T)\in V_n\) such that
\begin{equation}\label{eqn:qubit-clifford-section-product}
M_SM_T=\kappa(S,T)w(c(S,T))M_{ST}.
\end{equation}
In the odd prime case, we showed that the representatives may be chosen so that \(c(S,T)=0\). This need not hold a priori when \(p=2\), and the resulting correction term is responsible for the additional structure that appears in the characteristic-two cohomological parametrization of \(\mathcal C^{(3)}\).  

\subsection{Affine, Symplectic and Phase Data}
\label{appendix:qubit-affine-symplectic-phase}
For \(S,T\in\operatorname{Sp}(V_n)\), there are \(\kappa(S,T)\in U(1)\) and \(c(S,T)\in V_n\) such that \cref{eqn:qubit-clifford-section-product} is satisfied. Since the section is normalized by, we have
\begin{equation}
\label{eqn:qubit-c-normalization}
c(\operatorname{Id}_{V_n},S) = c(S,\operatorname{Id}_{V_n}) = 0
\end{equation}
for every \(S\in\operatorname{Sp}(V_n)\). \cref{lem:qubit-clifford-extension-cocycle} shows that \(c(\cdot,\cdot)\) satisfies a cocycle identity.

\begin{lem}
\label{lem:qubit-clifford-extension-cocycle}
The map
\begin{equation}
c:
\operatorname{Sp}(V_n)\times\operatorname{Sp}(V_n)
\longrightarrow
V_n
\end{equation}
is a normalized \(2\)-cocycle for the natural action of \(\operatorname{Sp}(V_n)\) on \(V_n\), given by
\(S\cdot v:=Sv\). Explicitly,
\begin{equation}
\label{eqn:qubit-c-cocycle}
c(S,T)+c(ST,R) = Sc(T,R)+c(S,TR)
\end{equation}
for all \(S,T,R\in\operatorname{Sp}(V_n)\).
\end{lem}

\begin{proof}
By associativity, $(M_SM_T)M_R=M_S(M_TM_R).$ Using \cref{eqn:qubit-clifford-section-product} on the left-hand side, and collecting scalar factors into a phase \(\lambda\in U(1)\), we obtain
\begin{equation}
(M_SM_T)M_R = \lambda\, w(c(S,T)+c(ST,R))M_{STR}.
\end{equation}
On the other hand, using \cref{eqn:qubit-clifford-representative-action} and \cref{eqn:qubit-clifford-section-product} on the right-hand side gives, for some \(\lambda'\in U(1)\),
\begin{equation}
M_S(M_TM_R) = \lambda'\,
w(Sc(T,R)+c(S,TR))M_{STR}.
\end{equation}
\cref{lem:Weyl-label-phase-uniqueness} implies \cref{eqn:qubit-c-cocycle}. The normalization follows from \cref{eqn:qubit-c-normalization}.
\end{proof}

The appearance of the cocycle \(c\) is a new feature of the \(p=2\) case. Before examining how this cocycle enters the cohomological interpretation of \(\mathcal{C}^{(3)}\) in characteristic two, we first derive the \(p=2\) analog of \cref{third-level-classification}.

\begin{prop}\label{prop:qubit-third-level-classification}
For each \(a\in V_n\), write
\begin{equation}\label{eqn:qubit-third-level-B-normal-form} B(a)=\alpha_B(a)w(q_B(a))M_{S_B(a)}, \end{equation}
where \(\alpha_B(a)\in U(1)\), \(q_B(a)\in V_n\), and \(S_B(a)\in\operatorname{Sp}(V_n)\). Then \(B\in Y^3_2(V_n)\) if and only if
\begin{align}
S_B(a+b) &= S_B(a)S_B(b), \label{eqn:qubit-third-level-S-equation}\\
q_B(a+b) &= q_B(a)+S_B(a)q_B(b)+c(S_B(a),S_B(b)), \label{eqn:qubit-third-level-q-equation}\\
\alpha_B(a+b) &= \alpha_B(a)\alpha_B(b)\gamma_{S_B(a)}(q_B(b))\kappa(S_B(a),S_B(b)) \notag\\
&\qquad\times i^{\mu(q_B(a),S_B(a)q_B(b))+\mu(q_B(a)+S_B(a)q_B(b),c(S_B(a),S_B(b)))-\mu(a,b)}. \label{eqn:qubit-third-level-alpha-equation}
\end{align}
for all \(a,b\in V_n\).
\end{prop}

\begin{proof}
By \cref{rem:equiv-abstract-qubit}, we have \(B\in Y^3_2(V_n)\) if and only if
$B(a)B(b)=i^{\mu(a,b)}B(a+b)$ for all \(a,b\in V_n\). Using \cref{eqn:qubit-clifford-representative-action}, we obtain
\begin{align}
B(a)B(b) &= \alpha_B(a)\alpha_B(b)w(q_B(a))M_{S_B(a)}w(q_B(b))M_{S_B(b)}\\
&= \alpha_B(a)\alpha_B(b)\gamma_{S_B(a)}(q_B(b))w(q_B(a))w(S_B(a)q_B(b))M_{S_B(a)}M_{S_B(b)}.
\end{align}
Applying \cref{qubit-mult-comm} and \cref{eqn:qubit-clifford-section-product}, we get
\begin{align}
B(a)B(b) &= \alpha_B(a)\alpha_B(b)\gamma_{S_B(a)}(q_B(b))\kappa(S_B(a),S_B(b)) \notag\\
& \times i^{\mu(q_B(a),S_B(a)q_B(b))}w(q_B(a)+S_B(a)q_B(b))w(c(S_B(a),S_B(b)))M_{S_B(a)S_B(b)}\\
&= \alpha_B(a)\alpha_B(b)\gamma_{S_B(a)}(q_B(b))\kappa(S_B(a),S_B(b)) \notag\\
& \times i^{\mu(q_B(a),S_B(a)q_B(b))+\mu(q_B(a)+S_B(a)q_B(b),c(S_B(a),S_B(b)))} \notag\\
& \times w(q_B(a)+S_B(a)q_B(b)+c(S_B(a),S_B(b)))M_{S_B(a)S_B(b)}.
\end{align}
On the other hand,
\begin{equation}
i^{\mu(a,b)}B(a+b)=i^{\mu(a,b)}\alpha_B(a+b)w(q_B(a+b))M_{S_B(a+b)}.
\end{equation}
Comparing the two expressions and using the uniqueness of the normal form as in \cref{clifford-normal-form} gives \cref{eqn:qubit-third-level-S-equation,eqn:qubit-third-level-q-equation,eqn:qubit-third-level-alpha-equation}.
\end{proof}

As in the odd prime case, we would like to describe the symplectic and affine data in terms of an affine symplectic group. However, the presence of the additional cocycle \(c\) requires us to replace the usual affine symplectic group by a twisted analog. Consider \(\operatorname{ASp}_c(V_n)\cong V_n\rtimes_c \operatorname{Sp}(V_n)\), where the underlying set \(V_n\times\operatorname{Sp}(V_n)\) is endowed with multiplication
\begin{equation}
\label{eqn:qubit-twisted-affine-product}
(q,S)(r,T) := (q+Sr+c(S,T),ST).
\end{equation}

\cref{lem:qubit-twisted-affine-extension} will allows us to describe the symplectic and affine data in terms of \(\operatorname{ASp}_c(V_n)\).

\begin{lem}\label{lem:qubit-twisted-affine-extension}
The multiplication in \cref{eqn:qubit-twisted-affine-product} defines a group structure on \(\operatorname{ASp}_c(V_n)\), and there is a short exact sequence
\begin{equation}\label{eqn:qubit-twisted-affine-extension}
1\longrightarrow V_n\longrightarrow \operatorname{ASp}_c(V_n)\longrightarrow \operatorname{Sp}(V_n)\longrightarrow 1.
\end{equation}
Moreover, the normal form in \cref{clifford-normal-form} induces an isomorphism
\begin{equation}\label{eqn:qubit-clifford-quotient-twisted-affine}
\overline{\mathcal C}^{(2)}\cong \operatorname{ASp}_c(V_n).
\end{equation}
\end{lem}

\begin{proof}
Associativity of the multiplication in \cref{eqn:qubit-twisted-affine-product} follows from the cocycle identity in \cref{eqn:qubit-c-cocycle}, while \cref{eqn:qubit-c-normalization} implies that \((0,\operatorname{Id}_{V_n})\) is the identity element. For \((q,S)\in\operatorname{ASp}_c(V_n)\), an inverse is given by
\begin{equation}
(q,S)^{-1}= (-S^{-1}\bigl(q+c(S,S^{-1})\bigr),S^{-1} ).
\end{equation}
Indeed, the product of \((q,S)\) with this element on the right is \((0,\operatorname{Id}_{V_n})\), and the cocycle identity implies that the same element is also a left inverse. Hence, \cref{eqn:qubit-twisted-affine-product} defines a group structure on \(\operatorname{ASp}_c(V_n)\). The inclusion \(V_n\to\operatorname{ASp}_c(V_n)\) is given by \(q\mapsto(q,\operatorname{Id}_{V_n})\), while the projection \(\operatorname{ASp}_c(V_n)\to\operatorname{Sp}(V_n)\) is given by \((q,S)\mapsto S\). The projection is surjective and its kernel is precisely the image of \(V_n\), which gives \cref{eqn:qubit-twisted-affine-extension}. For the second statement, define
\begin{align}
\Xi:\operatorname{ASp}_c(V_n) &\longrightarrow \overline{\mathcal C}^{(2)},\\
(q,S) &\longmapsto [w(q)M_S].
\end{align}
Since the analog of the normal form in \cref{clifford-normal-form} continues to hold when \(p=2\), it follows that \(\Xi\) is a bijection. Using \cref{eqn:qubit-clifford-representative-action,eqn:qubit-clifford-section-product} and discarding scalar phases in \(\overline{\mathcal C}^{(2)}\), we obtain
\begin{align}
[w(q)M_S][w(r)M_T] &= [w(q)w(Sr)w(c(S,T))M_{ST}]= [w(q+Sr+c(S,T))M_{ST}].
\end{align}
Thus, \(\Xi\) is a group homomorphism and hence an isomorphism.
\end{proof}

By \cref{lem:qubit-twisted-affine-extension}, the first two equations in \cref{prop:qubit-third-level-classification} admit a structural formulation analogous to \cref{cor:third-level-affine-symplectic-homomorphism}. Namely, \cref{eqn:qubit-third-level-S-equation,eqn:qubit-third-level-q-equation} hold if and only if
\begin{align}
\Phi_{q,S}:V_n &\longrightarrow \operatorname{ASp}_c(V_n), \label{eqn:qubit-twisted-affine-homomorphism}\\
a &\longmapsto (q(a),S(a))
\end{align}
is a group homomorphism.

\subsection{Cohomological Interpretation}
\label{appendix:qubit-cohomological-interpretation}

Fix a homomorphism \(S:V_n\to\operatorname{Sp}(V_n)\) and let \(V_{n,S}\) defined as in \cref{cohom-inter}. Pull back the cocycle \(c\) along \(S\) by setting
\begin{equation}\label{eqn:qubit-pulled-back-cocycle}
c_S(a,b):=c(S(a),S(b)).
\end{equation}
Since \(S\) is a homomorphism and \(c\) satisfies \cref{eqn:qubit-c-cocycle}, we have
\begin{equation}\label{eqn:qubit-pulled-back-cocycle-identity}
S(a)c_S(b,d)+c_S(a,b+d)=c_S(a,b)+c_S(a+b,d)
\end{equation}
for all \(a,b,d\in V_n\). Hence, \(c_S\in Z^2(V_n,V_{n,S})\). For a \(1\)-cochain \(q:V_n\to V_n\), we use the additive coboundary convention
\begin{equation}\label{eqn:qubit-additive-coboundary}
(\delta q)(a,b):=q(a)+S(a)q(b)-q(a+b).
\end{equation}
With this convention, \cref{eqn:qubit-third-level-q-equation} is equivalent to $\delta q=-c_S$. Since \(V_n\) is a vector space over \(\mathbb F_2\), the minus sign may be replaced by a plus sign. We retain it to emphasize the standard coboundary equation.

\begin{cor}\label{cor:qubit-affine-cohomological-obstruction}
Let \(S:V_n\to\operatorname{Sp}(V_n)\) be a homomorphism. A map \(q:V_n\to V_n\) satisfying \cref{eqn:qubit-third-level-q-equation} exists if and only if
\begin{equation}\label{eqn:qubit-affine-obstruction-class}
[c_S]=0\in H^2(V_n,V_{n,S}).
\end{equation}
If this obstruction vanishes and \(q_0\) is one solution, then every solution is of the form \(q=q_0+z\) for some \(z\in Z^1(V_n,V_{n,S})\). Replacing \(q\) by \(q+q_v\), where
\begin{equation}
q_v(a):=S(a)v-v,
\end{equation}
corresponds to conjugating the twisted affine homomorphism \(\Phi_{q,S}:V_n\to\operatorname{ASp}_c(V_n)\) by the translation \((-v,\operatorname{Id}_{V_n})\in\operatorname{ASp}_c(V_n)\). Consequently, when \([c_S]=0\), twisted affine homomorphisms with fixed symplectic part \(S\), modulo conjugation by translations, form a torsor over \(H^1(V_n,V_{n,S})\).
\end{cor}

\begin{proof}
A solution \(q\) exists if and only if \(c_S\) is a \(2\)-coboundary, which is equivalent to \([c_S]=0\in H^2(V_n,V_{n,S})\). Suppose that \(q_0\) is one solution. If \(q\) is any other solution, then \(\delta(q-q_0)=0\), so \(q-q_0\in Z^1(V_n,V_{n,S})\). Conversely, if \(z\in Z^1(V_n,V_{n,S})\), then
\begin{equation}
\delta(q_0+z)=\delta q_0+\delta z=-c_S,
\end{equation}
so \(q_0+z\) is again a solution. Thus, the set of solutions is a torsor over \(Z^1(V_n,V_{n,S})\).

For the second statement, let \(v\in V_n\). Using the multiplication law in \(\operatorname{ASp}_c(V_n)\) and the normalization \(c(\operatorname{Id}_{V_n},S)=c(S,\operatorname{Id}_{V_n})=0\), we have
\begin{align}
(-v,\operatorname{Id}_{V_n})(q(a),S(a))(v,\operatorname{Id}_{V_n}) &= (-v+q(a),S(a))(v,\operatorname{Id}_{V_n})\\
&= (q(a)+S(a)v-v,S(a))\\
&= (q(a)+q_v(a),S(a)).
\end{align}
Thus, conjugation by the translation \((-v,\operatorname{Id}_{V_n})\) replaces \(q\) by \(q+q_v\). The maps \(q_v\) are precisely the \(1\)-coboundaries in \(B^1(V_n,V_{n,S})\). Hence, after quotienting by translation conjugacy, the torsor over \(Z^1(V_n,V_{n,S})\) descends to a torsor over \(H^1(V_n,V_{n,S})\).
\end{proof}

We now describe the phase obstruction. Assume that \(S\) and \(q\) satisfy \cref{eqn:qubit-third-level-S-equation,eqn:qubit-third-level-q-equation}, and define
\begin{equation}
T_{q,S}(a):=w(q(a))M_{S(a)}.
\end{equation}

\begin{lem}\label{lem:qubit-phase-obstruction-cocycle}
Define \(F_{q,S}:V_n\times V_n\to U(1)\) by
\begin{align}
F_{q,S}(a,b) &:= i^{\mu(q(a),S(a)q(b))+\mu(q(a)+S(a)q(b),c_S(a,b))} \gamma_{S(a)}(q(b))\kappa(S(a),S(b)). 
\label{eqn:qubit-Fqs}
\end{align}
Then
\begin{equation}
T_{q,S}(a)T_{q,S}(b)=F_{q,S}(a,b)T_{q,S}(a+b),
\end{equation}
and hence \(F_{q,S}\) is a normalized \(U(1)\)-valued \(2\)-cocycle on \(V_n\). Consequently,
\begin{equation}\label{eqn:qubit-Rqs}
R_{q,S}(a,b):=F_{q,S}(a,b)i^{-\mu(a,b)}
\end{equation}
is also a normalized \(U(1)\)-valued \(2\)-cocycle.
\end{lem}

\begin{proof}
The calculation in \cref{prop:qubit-third-level-classification}, together with \cref{eqn:qubit-third-level-q-equation}, gives
\begin{equation}
T_{q,S}(a)T_{q,S}(b)=F_{q,S}(a,b)T_{q,S}(a+b).
\end{equation}
Associativity of multiplication implies that \(F_{q,S}\) satisfies the \(2\)-cocycle identity. Moreover, \(q(0)=0\), \(S(0)=\operatorname{Id}_{V_n}\), and \(M_{\operatorname{Id}_{V_n}}=I\), so $T_{q,S}(0)=I$. Setting either argument equal to \(0\) in the preceding multiplication law therefore gives $F_{q,S}(a,0)=F_{q,S}(0,a)=1.$
Hence, \(F_{q,S}\) is a normalized \(U(1)\)-valued \(2\)-cocycle. The multiplication law
\begin{equation}
w(a)w(b)=i^{\mu(a,b)}w(a+b)
\end{equation}
implies, by associativity of the Weyl operators, that \(i^{\mu(a,b)}\) satisfies the \(2\)-cocycle identity. Since \(w(0)=I\), setting either argument equal to \(0\) gives $i^{\mu(a,0)}=i^{\mu(0,a)}=1.$ Thus, \((a,b)\mapsto i^{\mu(a,b)}\) is also a normalized \(U(1)\)-valued \(2\)-cocycle. Therefore, the ratio in \cref{eqn:qubit-Rqs} is a normalized \(U(1)\)-valued \(2\)-cocycle.
\end{proof}

\begin{cor}\label{cor:qubit-phase-cohomological-obstruction}
Let \(S:V_n\to\operatorname{Sp}(V_n)\) be a homomorphism and let \(q:V_n\to V_n\) satisfy \cref{eqn:qubit-third-level-q-equation}. There exists a phase function \(\alpha:V_n\to U(1)\) such that
\begin{equation}
B(a):=\alpha(a)w(q(a))M_{S(a)}
\end{equation}
belongs to \(Y^3_2(V_n)\) if and only if
\begin{equation}\label{eqn:qubit-phase-obstruction-class}
[R_{q,S}]=0\in H^2(V_n,U(1)).
\end{equation}
\end{cor}

\begin{proof}
By \cref{prop:qubit-third-level-classification}, the phase equation is equivalent to
\begin{equation}
\frac{\alpha(a+b)}{\alpha(a)\alpha(b)}=R_{q,S}(a,b)
\end{equation}
for all \(a,b\in V_n\). Recall that for a \(1\)-cochain \(\phi:V_n\to U(1)\), its \(2\)-coboundary is given by \((\delta\phi)(a,b)=\frac{\phi(a)\phi(b)}{\phi(a+b)}\). Hence, the preceding equation is equivalent to $R_{q,S}=\delta(\alpha^{-1})$. Therefore, such an \(\alpha\) exists if and only if \(R_{q,S}\) is a \(2\)-coboundary, which is equivalent to \cref{eqn:qubit-phase-obstruction-class}.
\end{proof}

\cref{prop:qubit-section-independence} shows that the obstruction classes above do not depend on the choice of representatives \(S\mapsto M_S\), provided the corresponding coordinates are transformed accordingly. 

\begin{prop}\label{prop:qubit-section-independence}
Let \(S\mapsto M_S\) and \(S\mapsto M'_S\) be two normalized choices of representatives in \(\mathcal C^{(2)}\). Then there exist functions
\begin{align}
\eta:\operatorname{Sp}(V_n)  \longrightarrow U(1), \qquad
d:\operatorname{Sp}(V_n)  \longrightarrow V_n
\end{align}
such that
\begin{equation}\label{eqn:qubit-change-section}
M'_S=\eta(S)w(d(S))M_S.
\end{equation}
If \(c\) and \(c'\) are the corresponding \(V_n\)-valued cocycles, then
\begin{equation}\label{eqn:qubit-c-section-change}
c'(S,T)=c(S,T)+d(S)+Sd(T)-d(ST).
\end{equation}
Consequently, $[c']=[c]\in H^2(\operatorname{Sp}(V_n),V_n)$.  For a fixed homomorphism \(S:V_n\to\operatorname{Sp}(V_n)\), the corresponding pulled-back classes satisfy $[c'_S]=[c_S]\in H^2(V_n,V_{n,S})$. Moreover, if \(q\) is compatible with \(c_S\), define
\begin{equation}\label{eqn:qubit-q-section-change}
q'(a):=q(a)-d(S(a)).
\end{equation}
Then \(q'\) is compatible with \(c'_S\), and the corresponding phase-obstruction classes satisfy $[R^{\prime}_{q',S}]=[R_{q,S}]\in H^2(V_n,U(1))$. 
\end{prop}

\begin{proof}
Since \(M'_S\) and \(M_S\) have the same image under \(\Psi\), the operator \(M'_SM_S^*\) lies in \(\ker\Psi=\mathcal C^{(1)}\). Thus, \cref{eqn:qubit-change-section} follows from the analog of \cref{clifford-normal-form}. Comparing the labels in the products \(M'_SM'_T\) and \(M'_{ST}\) and using \cref{eqn:qubit-clifford-representative-action} gives
\begin{equation}
c'(S,T)+d(ST)=d(S)+Sd(T)+c(S,T).
\end{equation}
This is precisely \cref{eqn:qubit-c-section-change}. Hence, \(c'\) and \(c\) differ by the \(2\)-coboundary of the \(1\)-cochain \(d\), and therefore determine the same class in \(H^2(\operatorname{Sp}(V_n),V_n)\). Pulling back along \(S\) gives \([c'_S]=[c_S]\). Define $d_S(a):=d(S(a)).$ By \cref{eqn:qubit-c-section-change}, we have
\begin{equation}
c'_S=c_S+\delta d_S.
\end{equation}
If \(\delta q=-c_S\), then \cref{eqn:qubit-q-section-change} gives
\begin{equation}
\delta q'=\delta q-\delta d_S=-c_S-\delta d_S=-c'_S.
\end{equation}
Thus, \(q'\) is compatible with the new section. Finally, define
\begin{align}
T(a) :=w(q(a))M_{S(a)}, \qquad
T'(a):=w(q'(a))M'_{S(a)}.
\end{align}
Using \cref{eqn:qubit-change-section}, \cref{eqn:qubit-q-section-change} and the characteristic-two multiplication law, we obtain $T'(a)=\lambda(a)T(a),$ where $\lambda(a):=\eta(S(a))i^{\mu(q'(a),d(S(a)))}.$ Therefore, if \(F_{q,S}\) and \(F'_{q',S}\) denote the corresponding projective multipliers, then $F'_{q',S}=F_{q,S}\cdot\delta\lambda.$ Since \(i^{\mu(a,b)}\) is unchanged, it follows that $R'_{q',S}=R_{q,S}\cdot\delta\lambda.$
Hence, \(R'_{q',S}\) and \(R_{q,S}\) determine the same class in \(H^2(V_n,U(1))\). 
\end{proof}

The discussion above shows that, for a fixed symplectic homomorphism \(S\), the characteristic-two parametrization involves two successive obstruction classes. The first class $[c_S]\in H^2(V_n,V_{n,S})$ determines whether compatible affine data exist. When this class vanishes, the translation classes of compatible affine data form a torsor over \(H^1(V_n,V_{n,S})\). For a compatible pair \((q,S)\), the second class $[R_{q,S}]\in H^2(V_n,U(1))$ determines whether compatible phase data exist. This gives the characteristic-two analogue of the \(H^1\) and \(H^2\) description developed in \cref{cohom-inter}.

\subsection{Generator Level Criteria}
\label{appendix:qubit-generator-level-criteria}

We conclude \cref{appendix:qubit-affine-symplectic-phase} by giving the characteristic-two analogues of \cref{cor:third-level-generator-equations} and \cref{cor-cohom}. Let \(e_1,\ldots,e_{2n}\) be an \(\mathbb F_2\)-basis of \(V_n\), and write \(S_i:=S(e_i)\) and \(q_i:=q(e_i)\).

\begin{cor}\label{cor:qubit-third-level-generator-equations}
A pair of maps \(S:V_n\to\operatorname{Sp}(V_n)\) and \(q:V_n\to V_n\) satisfies \cref{eqn:qubit-third-level-S-equation,eqn:qubit-third-level-q-equation} if and only if
\begin{align}
S_iS_j &= S_jS_i, \label{eqn:qubit-generator-S-commute}\\
S_i^2 &= \operatorname{Id}_{V_n}, \label{eqn:qubit-generator-S-order}\\
q_i+S_iq_j+c(S_i,S_j) &= q_j+S_jq_i+c(S_j,S_i), \label{eqn:qubit-generator-q-commute}\\
q_i+S_iq_i+c(S_i,S_i) &= 0. \label{eqn:qubit-generator-q-order}
\end{align}
for all \(1\leq i,j\leq 2n\).
\end{cor}

\begin{proof}
By \cref{eqn:qubit-twisted-affine-homomorphism}, specifying \((q,S)\) is equivalent to specifying a group homomorphism
\begin{align}
\Phi_{q,S}:V_n &\longrightarrow \operatorname{ASp}_c(V_n),\\
a &\longmapsto (q(a),S(a)).
\end{align}
The additive group \(V_n\) is generated by \(e_1,\ldots,e_{2n}\), subject to the relations
\begin{equation}
e_i+e_j=e_j+e_i,\qquad 2e_i=0.
\end{equation}
Hence, \(\Phi_{q,S}\) is a homomorphism if and only if the elements \((q_i,S_i)\) commute pairwise and have order dividing \(2\). Using \cref{eqn:qubit-twisted-affine-product}, pairwise commutativity is equivalent to \cref{eqn:qubit-generator-S-commute,eqn:qubit-generator-q-commute}. Moreover,
\begin{equation}
(q_i,S_i)^2=(q_i+S_iq_i+c(S_i,S_i),S_i^2).
\end{equation}
Thus, \((q_i,S_i)^2=(0,\operatorname{Id}_{V_n})\) is equivalent to \cref{eqn:qubit-generator-S-order,eqn:qubit-generator-q-order}.
\end{proof}

Once the twisted affine data have been fixed, the phase obstruction can again be tested on a basis using the commutator bicharacter from \cref{appendix:cohom-facts}.

\begin{cor}\label{cor:qubit-phase-generator-obstruction}
Let \(q\) and \(S\) satisfy \cref{eqn:qubit-third-level-S-equation,eqn:qubit-third-level-q-equation}. Then $[R_{q,S}]=0\in H^2(V_n,U(1))$ if and only if
\begin{equation}\label{eqn:qubit-phase-generator-test}
\frac{R_{q,S}(e_i,e_j)}{R_{q,S}(e_j,e_i)}=1
\end{equation}
for all \(1\leq i,j\leq 2n\). Equivalently,
\begin{equation}\label{eqn:qubit-phase-generator-test-expanded}
\frac{\gamma_{S_i}(q_j)\kappa(S_i,S_j)}{\gamma_{S_j}(q_i)\kappa(S_j,S_i)}i^{\Delta_{ij}}=1,
\end{equation}
where
\begin{align}
\Delta_{ij} &:= \mu(q_i,S_iq_j)-\mu(q_j,S_jq_i) \notag\\
&\quad+\mu(q_i+S_iq_j,c(S_i,S_j))-\mu(q_j+S_jq_i,c(S_j,S_i)) \notag\\
&\quad-\mu(e_i,e_j)+\mu(e_j,e_i)
\end{align}
in $\mathbb Z_4$
\end{cor}

\begin{proof}
By \cref{lem:qubit-phase-obstruction-cocycle}, \(R_{q,S}\) is a normalized \(U(1)\)-valued \(2\)-cocycle on the finite abelian group \(V_n\). Therefore, \cref{cor:coboundary-iff-trivial-commutator} implies that its cohomology class vanishes if and only if its commutator bicharacter is trivial on pairs of basis vectors. This gives \cref{eqn:qubit-phase-generator-test}. Substituting \cref{eqn:qubit-Fqs,eqn:qubit-Rqs} gives \cref{eqn:qubit-phase-generator-test-expanded}.
\end{proof}

\section{Commuting Clifford Lifts for Unipotent Symplectic Matrices}
\label[appendix]{appendix:upper-unipotent-clifford-lifts}

In both \cref{two-qudit-c3-semi-clifford} and \cref{three-qudit-c3-semi-clifford}, we used the fact that the relevant unipotent symplectic subgroup admits pairwise commuting Clifford lifts. We now prove this fact in a form that applies to an arbitrary number of qudits.

\begin{lem}\label{lem:upper-unipotent-commuting-lifts}
Let \(p\) be an odd prime and let \(m\geq 1\). Write $V_m=\mathbb F_p^m\oplus \mathbb F_p^m$
with its standard symplectic form. For each symmetric matrix \(C\in\operatorname{Sym}_m(\mathbb F_p)\), define
\begin{equation}
u_C:=\begin{pmatrix}I&C\\0&I\end{pmatrix}\in\operatorname{Sp}(V_m).
\end{equation}
Then there exist Clifford unitaries \(N_C\) such that
\begin{equation}
N_Cw(s,t)N_C^*=w(s+Ct,t)
\end{equation}
for all \((s,t)\in V_m\), and $N_CN_D=N_{C+D}$
for all \(C,D\in\operatorname{Sym}_m(\mathbb F_p)\). Moreover, if \(M_{u_C}\) denotes any fixed choice of Clifford representative implementing \(u_C\), then the operators \(M_{u_C}\) commute pairwise.
\end{lem}

\begin{proof}
The assignment \(C\mapsto u_C\) is a homomorphism from the additive group \(\operatorname{Sym}_m(\mathbb F_p)\) to \(\operatorname{Sp}(V_m)\). We now choose explicit Clifford lifts of these symplectic matrices using diagonal quadratic phase gates and the \(m\)-qudit Fourier transform. For \(C\in\operatorname{Sym}_m(\mathbb F_p)\), define the diagonal quadratic phase gate
\begin{equation}
D_C\ket{x}:=\omega^{\frac12 x^TCx}\ket{x}.
\end{equation}
We first compute the conjugation action of \(D_C\) on Weyl operators. For \(s,t,x\in\mathbb F_p^m\), we have
\begin{align}
D_Cw(s,t)D_C^*\ket{x}&=D_C (\omega^{-\frac12 s\cdot t}Z^tX^s\omega^{-\frac12 x^TCx}\ket{x} )\\
&=D_C (\omega^{-\frac12 s\cdot t-\frac12 x^TCx+t\cdot(x+s)}\ket{x+s} )\\
&=\omega^{-\frac12 s\cdot t-\frac12 x^TCx+t\cdot(x+s)+\frac12(x+s)^TC(x+s)}\ket{x+s}.
\end{align}
Since \(C=C^T\), we have
\begin{equation}
(x+s)^TC(x+s)-x^TCx=2x^TCs+s^TCs.
\end{equation}
Therefore, we have
\begin{align}
D_Cw(s,t)D_C^*\ket{x}&=\omega^{-\frac12 s\cdot t+t\cdot(x+s)+x^TCs+\frac12 s^TCs}\ket{x+s}\\
&=\omega^{-\frac12 s\cdot(t+Cs)+(t+Cs)\cdot(x+s)}\ket{x+s}\\
&=\omega^{-\frac12 s\cdot(t+Cs)}Z^{t+Cs}X^s\ket{x}=w(s,t+Cs)\ket{x}.
\end{align}
Since this holds for every computational basis vector \(\ket{x}\), we obtain
\begin{equation}\label{eqn:DC-conjugation-general-m}
D_Cw(s,t)D_C^*=w(s,t+Cs).
\end{equation}
Let \(\mathcal F\) denote the \(m\)-qudit Fourier transform
\begin{equation}
\mathcal F\ket{x}:=\frac{1}{p^{m/2}}\sum_{y\in\mathbb F_p^m}\omega^{x\cdot y}\ket{y}.
\end{equation}
We have, 
\begin{align}
\mathcal F X^s\mathcal F^*\ket{x}
&=\frac{1}{p^{m/2}}\sum_{y\in\mathbb F_p^m}\omega^{-x\cdot y}\mathcal F\ket{y+s}\\
&=\frac{1}{p^m}\sum_{y,z\in\mathbb F_p^m}\omega^{-x\cdot y+(y+s)\cdot z}\ket{z}\\
&=\frac{1}{p^m}\sum_{z\in\mathbb F_p^m}\omega^{s\cdot z}\left(\sum_{y\in\mathbb F_p^m}\omega^{y\cdot(z-x)}\right)\ket{z}=\omega^{s\cdot x}\ket{x}=Z^s\ket{x}.
\end{align}
Thus, $\mathcal F X^s\mathcal F^*=Z^s$. Similarly, we have
\begin{align}
\mathcal F Z^t\mathcal F^*\ket{x}
&=\frac{1}{p^{m/2}}\sum_{y\in\mathbb F_p^m}\omega^{-x\cdot y}\mathcal F Z^t\ket{y}\\
&=\frac{1}{p^{m/2}}\sum_{y\in\mathbb F_p^m}\omega^{-x\cdot y+t\cdot y}\mathcal F\ket{y}\\
&=\frac{1}{p^m}\sum_{y,z\in\mathbb F_p^m}\omega^{(t-x)\cdot y+y\cdot z}\ket{z}=\ket{x-t}=X^{-t}\ket{x}.
\end{align}
Thus, $\mathcal F Z^t\mathcal F^*=X^{-t}$. 
Therefore, we have
\begin{align}
\mathcal Fw(s,t)\mathcal F^*
&=\omega^{-\frac12 s\cdot t}\mathcal F Z^tX^s\mathcal F^*\\
&=\omega^{-\frac12 s\cdot t}X^{-t}Z^s\\
&=\omega^{-\frac12 s\cdot t}\omega^{s\cdot t}Z^sX^{-t}=\omega^{\frac12 s\cdot t}Z^sX^{-t}=w(-t,s).
\end{align}
Define $N_C:=\mathcal F^*D_{-C}\mathcal F$. Then \(N_C\) implements the upper unipotent symplectic transformation \(u_C\). Indeed, we have
\begin{align}
N_Cw(s,t)N_C^*
&=\mathcal F^*D_{-C}\mathcal Fw(s,t)\mathcal F^*D_{-C}^*\mathcal F\\
&=\mathcal F^*D_{-C}w(-t,s)D_{-C}^*\mathcal F\\
&=\mathcal F^*w(-t,s+Ct)\mathcal F=w(s+Ct,t)=w(u_C(s,t)).
\end{align}
Thus, \(N_C\) is a Clifford lift of \(u_C\). Moreover, these lifts commute because
\begin{align}
N_CN_D
&=\mathcal F^*D_{-C}\mathcal F\mathcal F^*D_{-D}\mathcal F\\
&=\mathcal F^*D_{-C}D_{-D}\mathcal F\\
&=\mathcal F^*D_{-(C+D)}\mathcal F=N_{C+D}=N_{D+C}=N_DN_C.
\end{align}
Finally, let \(M_{u_C}\) denote the fixed Clifford representative implementing \(u_C\). Both \(M_{u_C}\) and \(N_C\) implement the same conjugation action on Weyl operator and \(N_C^*M_{u_C}\) commutes with every Weyl operator. Hence, there exists \(\eta(C)\in U(1)\) such that $M_{u_C}=\eta(C)N_C.$
Consequently, the fixed lifts \(M_{u_C}\) also commute pairwise. Indeed, we have
\begin{align}
M_{u_C}M_{u_D}
&=\eta(C)\eta(D)N_CN_D\\
&=\eta(C)\eta(D)N_DN_C\\
&=M_{u_D}M_{u_C}.
\end{align}
This completes the proof.
\end{proof}

\section{Omitted Proofs}\label[appendix]{appendix:omitted}
In \cref{appendix:omitted}, we provide the details of proofs that were omitted from the main text.

\subsection{Proof of Lemma 2.10}\label[appendix]{appendix:semi-clifford-proof}
\begin{proof}
Let \(A_Z\subseteq\mathcal C^{(1)}\) denote the maximal abelian subgroup generated by the \(Z\)-type Weyl operators. Explicitly,
\begin{equation}
A_Z:=\{\lambda Z^t:\lambda\in U(1),\ t\in\mathbb F_p^n\}.
\end{equation}
Suppose that \(U\) is semi-Clifford. By definition, there exist \(C_1,C_2\in\mathcal C^{(2)}\) and a diagonal unitary operator \(D\) such that $U=C_1DC_2$. Since \(D\) is diagonal, it commutes with \(A_Z\). Hence, $DA_ZD^*=A_Z$. Set
\begin{equation}
A:=C_2^*A_ZC_2,\qquad A':=C_1A_ZC_1^*.
\end{equation}
Since \(C_1\) and \(C_2\) act on \(\mathcal C^{(1)}\) by automorphisms, both \(A\) and \(A'\) are maximal abelian subgroups of \(\mathcal C^{(1)}\). Moreover, we have
\begin{equation}
UAU^*=C_1DC_2(C_2^*A_ZC_2)C_2^*D^*C_1^*=C_1DA_ZD^*C_1^*=C_1A_ZC_1^*=A'.
\end{equation}
Thus, \(U\) maps a maximal abelian subgroup of \(\mathcal C^{(1)}\) to another maximal abelian subgroup of \(\mathcal C^{(1)}\). Conversely, suppose that there exist maximal abelian subgroups \(A,A'\subseteq\mathcal C^{(1)}\) such that $UAU^*=A'$.  It is known that \(\mathcal C^{(2)}\) acts transitively on maximal abelian subgroups of \(\mathcal C^{(1)}\). Hence, there exist \(C_R,C_L\in\mathcal C^{(2)}\) such that $C_RA_ZC_R^*=A$ and $ C_LA'C_L^*=A_Z$. 
Define $V:=C_LUC_R$. We have
\begin{align}
V A_Z V^*
&= C_L U C_R A_Z C_R^* U^* C_L^* = C_L U A U^* C_L^*= C_L A' C_L^*= A_Z.
\end{align}
Hence, \(V\) determines an automorphism of \(A_Z\) via conjugation that fixes scalar phases. This automorphism is implemented by a \(K\in\mathcal C^{(2)}\). This is because any automorphism of \(A_Z\) fixing scalar phases has the form $Z^t\mapsto \chi(t)Z^{Mt}$, where \(M\in \operatorname{GL}(\mathbb F^n_p)\) and \(\chi:\mathbb F_p^n\to U(1)\) is a character.\footnote{Since \(A_Z/U(1) \cong \mathbb F_p^n\), any automorphism of \(A_Z\) fixing scalar phases descends to an automorphism of \(\mathbb F_p^n\). Hence it is given by some \(M\in \operatorname{GL}(\mathbb F^n_p)\). Therefore, the image of \(Z^t\) must lie in the coset corresponding to \(Mt\), so it has the form $Z^t\longmapsto \chi(t)Z^{Mt}$
for some phase \(\chi(t)\in U(1)\). Since the automorphism is a group homomorphism and \(Z^{t+s}=Z^tZ^s\), we get
\[
\chi(t+s)Z^{(Mt+Ms)}
=
\chi(t+s)Z^{M(t+s)}
=
\chi(t)\chi(s)Z^{Mt+Ms}.
\]
Thus \(\chi(t+s)=\chi(t)\chi(s)\), so \(\chi:\mathbb F_p^n\to U(1)\) is a character. Therefore, the automorphism of \(A_Z\) is implemented by some Clifford operator \(K\in\mathcal C^{(2)}\) normalizing \(A_Z\).
}
The linear part \(t\mapsto Mt\) is implemented by a permutation of the computational basis, hence by an element of \(\mathcal C^{(2)}\). The character factor \(\chi(t)\) is implemented by conjugation with an \(X\)-type Weyl operator\footnote{
For \(L\in\operatorname{GL}(\mathbb F^n_p)\), define the permutation \(P_L\) by $P_L\ket{x}=\ket{Lx}$. 
Then
\[
P_L Z^t P_L^*\ket{x}=\omega^{t\cdot L^{-1}x}\ket{x}=\omega^{((L^{-1})^{T}t)\cdot x}\ket{x}=Z^{(L^{-1})^{T}t}\ket{x}.
\]
Thus, choosing \(L\) such that \((L^{-1})^{T}=M\), the unitary \(P_L\) implements \(Z^t\mapsto Z^{Mt}\). \(P_L \in \mathcal{C}^{(2)}\) since it maps Weyl operators to Weyl operators. Moreover, every character \(\chi:\mathbb F_p^n\to U(1)\) has the form \(\chi(t)=\omega^{s\cdot t}\) for some \(s\in\mathbb F_p^n\). Conjugation by \(X^s\) gives
$X^s Z^t (X^s)^*=\omega^{-s\cdot t}Z^t.$
Choosing \(s\) with the appropriate sign therefore implements the factor \(Z^t\mapsto \chi(t)Z^t\).}. Hence, we have $VK^*aKV^*=a$ for every \(a\in A_Z\). Hence, \(VK^*\) commutes with \(A_Z\). Since the algebra generated by \(A_Z\) is the diagonal maximal abelian algebra, its commutant is itself. Therefore \(VK^*\) is diagonal. Write $D:=VK^*$. Then \(D\) is diagonal and $V=DK$. Consequently, we have
$U=C_L^*DKC_R^*.$
Since \(C_L^*, KC_R^* \in \mathcal{C}^{(2)}\), \(U\) is semi-Clifford in the sense of \cref{def:semi-clifford}.
\end{proof}

\subsection{One-Dimensional Image Case in Proposition 5.16}
\label[appendix]{appendix:rank-one-codim4-three}
We complete the proof of Proposition 5.16 by ruling out the case \(\dim_{\mathbb F_p}\rho(E_U)=1\). Set \(E_0:=\ker(\rho)\). Since $\dim E_U=4$ and 
$\dim \rho(E_U)=1$, 
we have
$\dim E_0=3$. 
Since \(\rho(E_U)\) has order \(p\), it is a Sylow \(p\)-subgroup of \(\operatorname{Sp}(\mathbb F_p^2)\). All Sylow \(p\)-subgroups are conjugate, so after changing the symplectic basis of \(H\) we may assume
\begin{equation}\label{eqn:three-qudit-c4-standard-p-subgroup}
\rho(E_U)=\left\{\begin{pmatrix}1&t\\0&1\end{pmatrix}:t\in\mathbb F_p\right\}.
\end{equation}
Choose \(h\in\mathbb F_p^6\) such that
\begin{equation}
(\rho\circ S_U)(h)=g:=\begin{pmatrix}1&1\\0&1\end{pmatrix},
\end{equation}
and set
\begin{equation}
\ell:=\ker(g-I)=\operatorname{im}(g-I)=\mathbb F_p\begin{pmatrix}1\\0\end{pmatrix}.
\end{equation}
Write \(S_U(h)=T(g,A_h,Z_h)\). If \(e\in\mathbb F_p^6\) satisfies \(S_U(e)\in E_0\), write \(S_U(e)=T(I,A_e,Z_e)\). Since \(E_U\) is abelian, \cref{eqn:three-qudit-c4-product} gives
$A_e+A_h=A_h+A_eg,$
and hence
\begin{equation}\label{eqn:three-qudit-c4-Ae-g}
A_e(g-I)=0.
\end{equation}
Since \(\ell=\operatorname{im}(g-I)\), it follows that \(A_e\ell=0\). We claim that \(A_h\ell\neq0\). Indeed, if \(A_h\ell=0\), then \(S_U(h)\) fixes \(\ell\) pointwise because \(g\) acts as the identity on \(\ell\), while every element of \(E_0\) fixes \(\ell\) pointwise by \cref{eqn:three-qudit-c4-Ae-g}. Since \(E_U/E_0\cong\rho(E_U)\) is generated by the image of \(S_U(h)\), the whole group \(E_U\) fixes \(\ell\) pointwise. Hence \(E_U\) fixes \(W\oplus\ell\) pointwise. Since \(W\) is isotropic, \(\ell\subseteq H\subseteq W^\perp\), and
\begin{equation}
\dim(W\oplus\ell)=2+1=3,
\end{equation}
the subspace \(W\oplus\ell\) is Lagrangian. This contradicts \cref{lem:three-qudit-large-image-obstruction}. Therefore,
$A_h\ell\neq0.$
Let
\begin{equation}
v:=\begin{pmatrix}1\\0\end{pmatrix}\in\ell,\qquad u:=A_hv\neq0.
\end{equation}
Since $g-I=E_{12}$, \cref{eqn:three-qudit-c4-Ae-g} implies that every \(A_e\) has the form
$A_e=(0\;\;w)$
for some column vector \(w\). Comparing the \(Z\) -components of \(S_U(h)S_U(e)=S_U(e)S_U(h)\) using \cref{eqn:three-qudit-c4-product} gives
$uw^T+wu^T=0.$
Since \(p\) is odd and \(u\neq0\), this forces \(w=0\). Thus \(A_e=0\) for every element of \(E_0\). Therefore,
$E_0\subseteq Z_0.$
Both \(E_0\) and \(Z_0\) have \(\mathbb F_p\)-dimension \(3\). Since \(E_0\subseteq Z_0\), it follows that \(E_0=Z_0\).
We now show that \(\beta\) has rank at most \(1\). Since \(q_U(K_U)=W\), we have
$\beta(k)=0$
for every \(k\in K_U\). For \(Z\in\operatorname{Sym}_2(\mathbb F_p)\), choose \(b_Z\in\mathbb F_p^6\) such that
\begin{equation}
S_U(b_Z)=T(I,0,Z).
\end{equation}
If \(b_Z'\) is another such choice, then \(b_Z-b_Z'\in K_U\), and hence \(\beta(b_Z)=\beta(b_Z')\). Thus \(\beta(b_Z)\) depends only on \(Z\), and we obtain a linear map
\begin{align}
\Theta:\operatorname{Sym}_2(\mathbb F_p) &\longrightarrow W^*,\\
Z &\longmapsto \beta(b_Z).
\end{align}
Applying the affine cocycle equation to \(b_Z+b_D=b_D+b_Z\) and comparing the \(W\) components gives
\begin{equation}\label{eqn:three-qudit-c4-ZThetaD}
Z\Theta(D)=D\Theta(Z)
\end{equation}
for all \(Z,D\in\operatorname{Sym}_2(\mathbb F_p)\). Taking \(D=I\) gives
$\Theta(Z)=Z\Theta(I).$
Set \(t:=\Theta(I)\). Substituting \(\Theta(Z)=Zt\) into \cref{eqn:three-qudit-c4-ZThetaD} gives
$(ZD-DZ)t=0$. Taking
\begin{equation}
Z=\begin{pmatrix}1&0\\0&0\end{pmatrix},\qquad D=\begin{pmatrix}0&1\\1&0\end{pmatrix}
\end{equation}
gives
\begin{equation}
ZD-DZ=\begin{pmatrix}0&1\\-1&0\end{pmatrix},
\end{equation}
which is invertible. Hence \(t=0\), and therefore
$\Theta=0.$
Thus,  \(\beta\) vanishes on \(S_U^{-1}(E_0)\).
Since \(\ker S_U=K_U\) and \(S_U\) maps \(S_U^{-1}(E_0)\) onto \(E_0\), we have
$S_U^{-1}(E_0)/K_U\cong E_0.$
Therefore,
\begin{equation}
\dim S_U^{-1}(E_0)=\dim K_U+\dim_{\mathbb F_p}E_0=2+3=5.
\end{equation}
Thus, \(\beta\) vanishes on a five-dimensional subspace of \(\mathbb F_p^6\), so
$\operatorname{rank}\beta\leq6-5=1,$
a contradiction.

\subsection{Proof of Proposition 5.17}\label[appendix]{appendix:codim5-proof}

\begin{proof}
Suppose, for contradiction, that \(\operatorname{codim}K_U=5\).  Then \(\dim K_U=1\) and
$\dim E_U=5.$
The symplectic form on $K_U$ is identically zero.  Set \(W:=q_U(K_U)\). 
By \cref{lem:kernel-cocycle-consequences}, \(W\) fixed pointwise by \(E_U\) and symplectic form on $W$ vanishes. As in \cref{lem:three-qudit-codim-four-impossible}, we choose a symplectic decomposition adapted to \(W\) by choosing a four-dimensional non-degenerate symplectic subspace \(H\subseteq W^\perp\), and choosing an isotropic line \(W^*\subseteq H^\perp\) dual to \(W\). Thus,
\begin{equation}\label{eqn:three-qudit-c5-witt}
\mathbb F_p^6=W\oplus H\oplus W^*.
\end{equation}
Choose dual bases of \(W\) and \(W^*\), choose a symplectic basis of \(H\), and let \(J\) denote the matrix of the symplectic form on \(H\). With the same ordering of the basis as in \cref{lem:three-qudit-codim-four-impossible}, the symplectic form on $\mathbb F_p^6$ is represented by the matrix in \cref{eqn:three-qudit-c4-symplectic-matrix} with block sizes \(1,4,1\) in place of \(2,2,2\). As in \cref{lem:three-qudit-codim-four-impossible}, 
every symplectic transformation fixing \(W\) pointwise therefore has a unique expression
\begin{equation}\label{eqn:three-qudit-c5-stabilizer}
T(g,A,z):=\begin{pmatrix}1&A&z\\0&g&gJA^T\\0&0&1\end{pmatrix},
\end{equation}
where \(g\in\operatorname{Sp}(H)\), \(A\in M_{1\times4}(\mathbb F_p)\), and \(z\in\mathbb F_p\). 
A direct block multiplication gives
\begin{equation}\label{eqn:three-qudit-c5-product}
T(g_1,A_1,z_1)T(g_2,A_2,z_2)=T\left(g_1g_2,A_2+A_1g_2,z_1+z_2+A_1g_2JA_2^T\right).
\end{equation}
As in \cref{lem:three-qudit-codim-four-impossible}, we record the induced action on the non-degenerate symplectic component \(H \cong \mathbb F_p^4\). Taking the middle diagonal block defines a homomorphism
\begin{align}
\rho:E_U&\longrightarrow \operatorname{Sp}(\mathbb F_p^4),\\
T(g,A,z)&\longmapsto g.
\end{align}
Set \(r:=\dim \rho(E_U) \) and \(E_0:=\ker\rho\). Hence, 
$\dim E_0=5-r$. Moreover, \(\rho(E_U)\) is an abelian \(p\)-subgroup of \(\operatorname{Sp}(\mathbb F_p^4)\). Since such a  subgroup has order at most \(p^3\), we have \(r\leq3\).
We next analyze the \(A\) components of the elements of \(E_0\), analogous  to the analysis in \cref{lem:three-qudit-codim-four-impossible}. Define
\begin{equation}
\mathcal A:=\{A\in M_{1\times4}(\mathbb F_p):T(I,A,z)\in E_0\text{ for some }z\}.
\end{equation}
As in \cref{lem:three-qudit-codim-four-impossible},  \(\mathcal A\) is a linear subspace. As in \cref{lem:three-qudit-codim-four-impossible}, we 
get $AJD^T=DJA^T$
for all \(A,D\in\mathcal A\). On the other hand, \(AJD^T\) is a scalar and since \(J^T=-J\), we have
\begin{equation}
AJD^T=(AJD^T)^T=DJ^TA^T=-DJA^T.
\end{equation}
Hence, \(2AJD^T=0\). Since \(p\) is odd, we have
$AJD^T=0$
for all \(A,D\in\mathcal A\).
Thus, \(M_{1\times4}(\mathbb F_p)\) becomes a four-dimensional symplectic vector space with alternating form
\begin{equation}
\omega(A,D):=AJD^T.
\end{equation}
Indeed, \(\omega(A,A)=AJA^T=0\) because \(J\) is skew-symmetric, and the form is non-degenerate since \(AJD^T=0\) for every \(D\) implies \(AJ=0\)\footnote{
If we choose \(D=e_i\), where \(e_i\) is the \(i\)-th standard basis row vector of \(\mathbb F_p^4\), then \(AJD^T\) is the \(i\)-th entry of \(AJ\). Since \(AJD^T=0\) for every \(D\), each entry of \(AJ\) vanishes. Hence, \(AJ=0\).
}, hence \(A=0\) because \(J\) is invertible. By the condition $AJD^T=0$, the restriction of \(\omega\) to \(\mathcal A\) vanishes identically, so \(\mathcal A\) is isotropic. Since an isotropic subspace of a four-dimensional symplectic vector space has dimension at most \(2\), we conclude that
$\dim\mathcal A\leq2.$ 
The homomorphism
\begin{align}
\phi : E_0 &\longrightarrow \mathcal A,\\
T(I,A,z) &\longmapsto A.
\end{align}
has kernel of dimension at most \(1\), since $\ker \phi$ is contained in \(\{T(I,0,z):z\in\mathbb F_p\}\). Since the homomorphism \(E_0\to\mathcal A\), \(T(I,A,z)\mapsto A\), has image of dimension at most \(2\) and kernel of dimension at most \(1\), the rank-nullity theorem gives
$\dim E_0\leq2+1=3.$
Since \(\dim E_0=5-r\), we obtain \(r\geq2\). Hence,
$r=2,3.$
As in \cref{lem:three-qudit-codim-four-impossible}, we rule out the possible ranks of the induced action separately.

\begin{enumerate}

\item Suppose first that \(r=2\). Since \(\dim E_0=3\), 
rank-nullity gives
$\dim\operatorname{im}\phi\geq2.$
Since \(\operatorname{im}\phi=\mathcal A\) and we already know that \(\dim\mathcal A\leq2\), it follows that \(\dim\mathcal A=2\). After choosing a symplectic basis of \(H\), we identify \(H\) with \(\mathbb F_p^4\), so that each \(A\in M_{1\times 4}(\mathbb F_p)\) defines a linear functional on \(H\), while \(JA^T\) defines a vector in \(H\). Set
\begin{equation}
L_0:=\{h\in H:Ah=0\text{ for every }A\in\mathcal A\}.
\end{equation}
For \(A,D\in\mathcal A\), we have
\begin{equation}
D(JA^T)=DJA^T=-AJD^T=0.
\end{equation}
Hence \(JA^T\in L_0\) for every \(A\in\mathcal A\), and therefore
\begin{equation}
\{JA^T:A\in\mathcal A\}\subseteq L_0.
\end{equation}
Since \(J\) is invertible, the map \(A\mapsto JA^T\) is injective, and hence its image has dimension \(2\). On the other hand, since \(\dim\mathcal A=2\), the space \(L_0\) is the common kernel of two linearly independent linear functionals on the four-dimensional space \(H\), so \(\dim L_0=2\). It follows that
\begin{equation}
L_0=\{JA^T:A\in\mathcal A\}.
\end{equation}
Finally, for \(A,D\in\mathcal A\), the symplectic pairing of the corresponding vectors in \(L_0\) is
\begin{equation}
(JA^T)^T J(JD^T)=AJD^T=0.
\end{equation}
Thus, \(L_0\) is isotropic. Since \(\dim L_0=2\) and \(H\) is four-dimensional and symplectic, \(L_0\) is Lagrangian.
Let \(T(g,B,z)\in E_U\) and \(T(I,A,y)\in E_0\). Since \(E_U\) is abelian, comparison of the \(A\) components in \cref{eqn:three-qudit-c5-product} gives
$Ag=A.$
Hence, for every \(A\in\mathcal A\),
\begin{equation}
gJA^T=Jg^{-T}A^T=JA^T,
\end{equation}
so \(g\) fixes \(L_0\) pointwise. Comparing the last components in the same two products gives
\begin{equation}
BJA^T=AgJB^T=AJB^T=-BJA^T.
\end{equation}
Since \(p\) is odd, \(BJA^T=0\). Thus \(B\) vanishes on \(L_0\). By \cref{eqn:three-qudit-c5-stabilizer}, every element of \(E_U\) therefore fixes \(L_0\) pointwise. Since \(E_U\) also fixes \(W\) pointwise, it fixes \(W\oplus L_0\) pointwise. The subspace \(W\oplus L_0\) is three-dimensional and isotropic, and hence is Lagrangian. This contradicts \cref{lem:three-qudit-large-image-obstruction}.

\item Suppose now that \(r=3\). Then \(\rho(E_U)\) is an abelian \(p\)-subgroup of \(\operatorname{Sp}(H)\cong\operatorname{Sp}(\mathbb F_p^4)\) of order \(p^3\). By \cref{wong}, exactly as in the proof of \cref{lem:three-qudit-codim-three}, after changing the symplectic basis of \(H\) we may write $H=L_0\oplus L_0^*$ and 
\begin{equation}\label{eqn:three-qudit-c5-rank-three-image}
\rho(E_U)=\left\{g_C:=\begin{pmatrix}I&C\\0&I\end{pmatrix}:C=C^T\in\operatorname{Sym}_2(\mathbb F_p)\right\},
\end{equation}
where \(L_0\) is Lagrangian in \(H\) and is fixed pointwise by \(\rho(E_U)\).
We show that the \(A\) component of every element of \(E_U\) vanishes on \(L_0\). First let \(T(I,A,z)\in E_0\), and write \(A=(a\;\;b)\) relative to \(H=L_0\oplus L_0^*\), with \(a,b\in M_{1\times2}(\mathbb F_p)\). As in the one-dimensional image case of \cref{lem:three-qudit-codim-four-impossible}, commuting this element with an element of \(E_U\) whose image under \(\rho\) is \(g_C\) gives \(Ag_C=A\). Hence
\begin{equation}
(a\;\;aC+b)=(a\;\;b)
\end{equation}
for every \(C\in\operatorname{Sym}_2(\mathbb F_p)\). Taking \(C=I\) gives \(a=0\). Thus, every element of \(E_0\) has an \(A\) component that vanishes on \(L_0\).
For each \(C\in\operatorname{Sym}_2(\mathbb F_p)\), choose an element \(T(g_C,A_C,z_C)\in E_U\) and write \(A_C=(a_C\;\;b_C)\). Since changing the choice changes \(A_C\) by the \(A\) component of an element of \(E_0\), the row vector \(a_C\) is independent of the choice. Since \(E_U\) is abelian, comparison of the \(A\) components of the products corresponding to \(C\) and \(D\) gives
\begin{equation}\label{eqn:three-qudit-c5-aC-relation}
a_CD=a_DC
\end{equation}
for all  \(C,D \in\operatorname{Sym}_2(\mathbb F_p)\). This is the same relation that appears in the codimension \(4\) argument in \cref{lem:three-qudit-codim-four-impossible}. Taking \(D=I\) and setting \(t:=a_I^T\), we obtain
$a_C^T=Ct.$
Substituting this into \cref{eqn:three-qudit-c5-aC-relation} gives
\begin{equation}
(DC-CD)t=0
\end{equation}
for all \(C,D \in\operatorname{Sym}_2(\mathbb F_p)\). As in \cref{lem:three-qudit-codim-four-impossible}, taking
\begin{equation}
C=\begin{pmatrix}1&0\\0&0\end{pmatrix},\qquad D=\begin{pmatrix}0&1\\1&0\end{pmatrix}
\end{equation}
gives an invertible commutator, so \(t=0\). Hence \(a_C=0\) for every \(C\). Therefore, the \(A\) component of every element of \(E_U\) vanishes on \(L_0\).
By \cref{eqn:three-qudit-c5-rank-three-image}, the \(g\) component of every element of \(E_U\) fixes \(L_0\) pointwise. Together with \cref{eqn:three-qudit-c5-stabilizer}, this shows that \(E_U\) fixes \(L_0\) pointwise. As in \cref{lem:three-qudit-codim-four-impossible}, adjoining the already fixed isotropic line \(W\) gives the Lagrangian subspace \(W\oplus L_0\), which is fixed pointwise by \(E_U\). This again contradicts \cref{lem:three-qudit-large-image-obstruction}.
\end{enumerate}
This completes the proof.
\end{proof}

\printbibliography

\end{document}